\documentclass[aps,superscriptaddress,twocolumn,10pt]{revtex4-2}

\usepackage{amsmath,mathtools,amsthm,amssymb}
\usepackage{tabularx}
\usepackage{tabularray}

\usepackage{graphicx}

\usepackage{color}
\usepackage[dvipsnames]{xcolor}
\usepackage{bbold}
\usepackage{enumitem}
\usepackage{physics}
\usepackage{comment}
\usepackage{array}
\usepackage{graphics}
\usepackage{wrapfig}
\graphicspath{{figures/}}
\usepackage{biolinum}
\usepackage{bbm}
\usepackage{dsfont}
\usepackage{mathrsfs}
\usepackage{mathdots}
\usepackage{enumitem}
\usepackage{pifont}
\definecolor{myrefcolor}{rgb}{0.067,0.5,0.5}
\definecolor{myurlcolor}{rgb}{0.1,0,0.9}
\usepackage{biolinum}

\usepackage[
    breaklinks,
    pdftex,
    colorlinks=true,
    linkcolor=myrefcolor,
    citecolor=myrefcolor,
    urlcolor=myrefcolor
]{hyperref}

\usepackage[capitalize]{cleveref}
\DeclareMathOperator{\poly}{poly}

\newtheorem{assmt}{Assumption}
\newtheorem*{theorem*}{Theorem}
\newcounter{thm}
\newtheorem{theorem}[thm]{Theorem}
\newtheorem{lemma}[thm]{Lemma}

\newtheorem{definition}{Definition}

\newtheorem{corollary}{Corollary}

\theoremstyle{remark}
\newtheorem{remark}{Remark}

\newcommand{\be}{\begin{equation}\begin{aligned}\hspace{0pt}}
\newcommand{\bbb}{\begin{equation*}\begin{aligned}}
\newcommand{\ee}{\end{aligned}\end{equation}}
\newcommand{\eee}{\end{aligned}\end{equation*}}
\newcommand{\fu}{Dahlem Center for Complex Quantum Systems, Freie Universit\"{a}t Berlin, 14195 Berlin, Germany}

\usepackage{algorithm}
\usepackage{algpseudocode}

\newcommand{\id}{\mathds{1}}

\newcommand{\haar}[0]{\operatorname{Haar}}

\definecolor{airforceblue}{rgb}{0.36, 0.54, 0.66}
\newcommand{\nocontentsline}[3]{}
\let\origcontentsline\addcontentsline
\newcommand\stoptoc{\let\addcontentsline\nocontentsline}
\newcommand\resumetoc{\let\addcontentsline\origcontentsline}

\allowdisplaybreaks
\begin{document}
\title{Entanglement theory with limited computational resources}
\author{Lorenzo Leone}
\thanks{Contributed equally. \{\href{mailto:lorenzo.leone@fu-berlin.de}{lorenzo.leone}, \href{mailto:jacopo.rizzo@fu-berlin.de}{jacopo.rizzo}\}@fu-berlin.de}
\author{Jacopo Rizzo}
\thanks{Contributed equally. \{\href{mailto:lorenzo.leone@fu-berlin.de}{lorenzo.leone}, \href{mailto:jacopo.rizzo@fu-berlin.de}{jacopo.rizzo}\}@fu-berlin.de}
\author{Jens Eisert}
\author{Sofiene Jerbi}
\affiliation{\fu}

\begin{abstract}

The precise quantification of the ultimate efficiency in manipulating quantum resources lies at the core of quantum information theory. However, purely information-theoretic measures fail to capture the actual computational complexity involved in performing certain tasks. In this work, we rigorously address this issue within the realm of entanglement theory, a cornerstone of quantum information science. We consider two key figures of merit: the computational distillable entanglement and the computational entanglement cost, quantifying the optimal rate of entangled bits (ebits) that can be extracted from or used to dilute many identical copies of $n$-qubit bipartite pure states, using computationally efficient local operations and classical communication (LOCC). We demonstrate that computational entanglement measures diverge  significantly from their information-theoretic counterparts. While the von Neumann entropy captures information-theoretic rates for pure-state transformations, we show that under computational constraints, the min-entropy instead governs optimal entanglement distillation. Meanwhile, efficient entanglement dilution incurs in a major cost, requiring maximal (\(\Tilde{\Omega}(n)\)) ebits even for nearly unentangled states. Surprisingly, in the worst-case scenario, even if an efficient description of the state exists and is fully known, one gains no advantage over state-agnostic protocols.
Our results establish a stark, maximal separation of \(\Tilde{\Omega}(n)\) vs. \(o(1)\) between computational and information-theoretic entanglement measures. Finally, our findings yield new sample-complexity bounds for measuring and testing the von Neumann entropy, fundamental limits on efficient state compression, and efficient LOCC tomography protocols.
\end{abstract}
\maketitle

\stoptoc
\onecolumngrid
\subsection{Introduction}

A central idea in modern science is that physical processes are inherently computational. This view, rooted in the physical Church-Turing thesis \cite{church1936an,turing1937computable,deutsch1985quantum}, suggests that all natural processes—including those governing quantum systems—can be understood as computations characterized by well-defined algorithmic rules. In this framework, quantum mechanics is not just a theory of physical phenomena, but also a computational model, limited by the constraints of information processing. As a result, the study of quantum information should not only focus on identifying transformations that are theoretically possible but also those that are computationally feasible within these limits. However, traditional approaches in quantum information theory often assume access to unbounded computational resources within an information-theoretic framework \cite{Wilde_2013}. While this assumption is undoubtedly mathematically useful, it leads to idealized models of quantum resource manipulation. These models completely overlook the computational complexities involved in transforming and processing quantum states, which are crucial to understanding the true characteristics of quantum information.

In this work, we address this issue within the framework of \textit{entanglement theory} \cite{Horodecki_2009}, arguably the most fundamental and well-known \textit{quantum resource}, both from a foundational and operational perspective. Entanglement not only underlies key advantages in computation \cite{shor1997polynomial}, super-dense coding \cite{PhysRevLett.69.2881}, quantum teleportation \cite{PhysRevLett.70.1895}, and cryptography \cite{Pirandola_2020} but also enables departures from classical theory, revealing the intrinsically non-classical nature of quantum information \cite{PhysicsPhysiqueFizika.1.195}. The basic setting of entanglement theory involves two distant parties, \textit{Alice} (\(A\)) and \textit{Bob} (\(B\)), who can only \textit{exchange classical information and perform local quantum operations} (LOCC) on their respective systems. LOCC transformations naturally classify quantum states into two categories: separable states, which are easily accessible, and entangled states, which serve as valuable resources. Two central objectives in entanglement theory are the \textit{quantification} of entanglement using entanglement monotones \cite{Vidal_2000,EntanglementUniquenessTheorem,Horodecki_2009,PhDJens,Virmani}
-- non-decreasing functions of the state under LOCC—and the \textit{manipulation} of entanglement using only LOCC operations, which are the ones freely available to the parties involved. 
Although entanglement theory is a very well-established field with numerous definitive results \cite{Bennett_1996,bennett1992communication,PhysRevLett.70.1895,Terhal_2000,PhysRevA.54.3824, PhysRevLett.76.722, Virmani, Nielsen_1999,Horodecki_2009,Lami_2023, gour2024resourcesquantumworld}, the natural impact of \textit{computational constraints} remains largely unexplored, 
with only very recent works hinting at some of the limitations 
that this framework may introduce \cite{aaronson2023quantumpseudoentanglement,arnonfriedman2023computationalentanglementtheory,bouland2023publickeypseudoentanglementhardnesslearning}. Thus, understanding how these restrictions redefine the actual accessible entanglement, both in terms of its quantification and the ability to manipulate it, constitutes a fundamental and pressing open problem—one that, remarkably, has yet to be thoroughly investigated.

Here, we determine the ultimate computational limits of entanglement theory within the most well-established and presumably most important setting: bipartite pure-state entanglement \cite{Bennett_1996,Nielsen_1999}. 
For pure states, the \textit{von Neumann entropy} of the reduced density matrix $S_1(\psi_A) = -\tr(\psi_A\log\psi_A) $ 
fundamentally captures both entanglement quantification and manipulation, determining the optimal asymptotic rates of state conversions under LOCC when many \emph{identical and independently distributed} (i.i.d.) copies of the input and output states are considered \cite{Bennett_1996}. Thus, in an asymptotic, information-theoretic sense, the von Neumann entropy is essentially the unique measure of entanglement for what concerns bipartite pure states \cite{EntanglementUniquenessTheorem,Nielsen_2000}.

However, as we show in this work, when computational limitations are taken into account, and the LOCC operations are also required to be \textit{computationally efficient}, this picture dramatically changes: the von Neumann entropy becomes fundamentally inaccessible, and fails to capture the optimal state-conversion rates achievable with limited computing power. This should not come entirely as a surprise. For example, it is well known that estimating the von Neumann entropy is computationally hard for large quantum systems \cite{9163139, wang_et_al:LIPIcs.ESA.2024.101}.
However, as one of the many consequences of our results, we prove that this computational inaccessibility takes place even when the von Neumann entropy is of order $O(1)$, and the entanglement is thus very low. This discussion immediately prompts the question: if the von Neumann entropy becomes inaccessible, 
\begin{center} 
\textit{what forms of entanglement can computationally bounded agents observe and manipulate?}
\end{center} 
To address this, we consider two key figures of merit: the \textit{computational distillable entanglement} and the \textit{computational entanglement cost}. When restricted to computationally efficient LOCC, the first quantifies the maximum rate of purified entanglement, in the form of entangled bits (ebits) that can be extracted from a set of pure i.i.d.\ quantum states, while the second captures the minimal rate of entanglement required to prepare a target bipartite i.i.d. source starting from pre-shared ebits. 
\begin{figure*}[t!]
    \centering
    \includegraphics[width=0.95\linewidth]{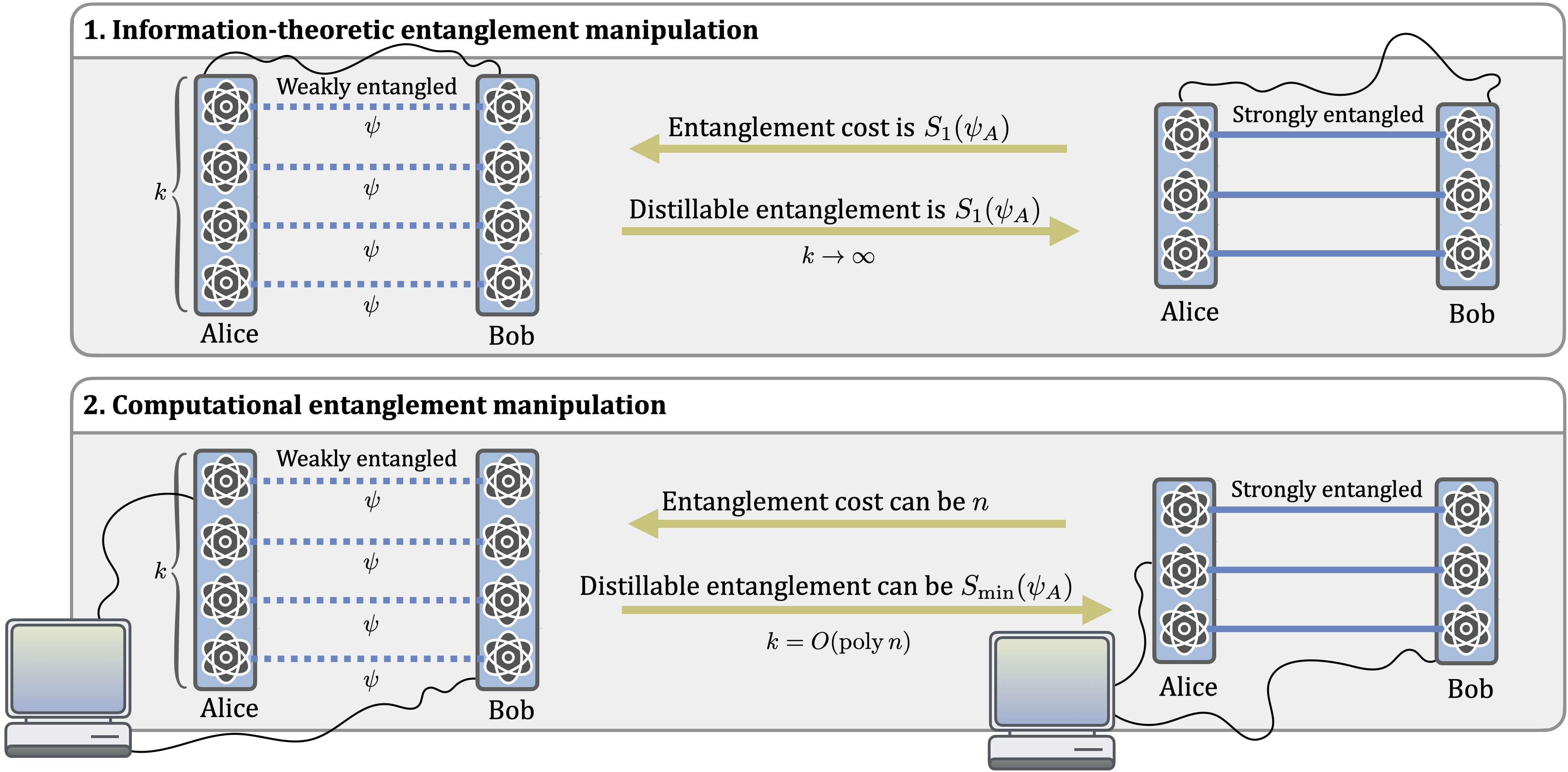}
    \caption{1.~The information-theoretic distillable entanglement and entanglement cost represent the optimal rates at which one can purify or dilute a collection of identical pure bipartite entangled states $\psi_{A,B}$ within the distant parties paradigm. In this setting, Alice and Bob, can naturally manipulate the quantum states held at their respective laboratories using only LOCC. A single quantity, the von-Neumann entropy of the reduced density matrix $S_1(\psi_A)$ governs both asymptotic rates when computational restrictions are not taken into account.
    2.~In this work, we consider the computational analogs thereof in which each step is assumed to be sample- and computational-efficient. We find families of states for which the computational distillable entanglement is $S_{\rm min}(\psi_A)$, and other families for which the entanglement cost is (maximally) $\Tilde{\Omega}(n)$. These results apply irrespectively of the actual value of $S_1(\psi_A)$, which therefore looses operational meaning in this computationally limited setting.}
\end{figure*}
Notably, the nature of these \textit{computational entanglement measures} is fundamentally different from their information-theoretic counterparts, which are entirely characterized by the von Neumann entropy. As we demonstrate throughout this work, another entanglement monotone assumes a crucial role here: the \textit{min-entropy} $S_{\min} = -\log\|\psi_A\|$, i.e., the negative logarithm of the largest eigenvalue of the reduced density matrix of the state $\psi_{A,B}$.

Specifically, we rigorously establish the existence of states for which any computationally efficient distillation protocol operating on \( k = O(\poly n) \) input copies cannot achieve a distillation rate greater than \( \min\{\Theta(S_{\min}),\Theta(\log k)\} \), regardless of the actual value of the von Neumann entropy which can be $\Tilde{\Omega}(n)$, i.e., maximal up to logarithmic factors. Then, leveraging standard tools from representation theory, we show that this rate is always attainable via an explicit, computationally efficient LOCC protocol—remarkably, without requiring any prior knowledge of the input state. On a seemingly unrelated note, we also apply this protocol to aid \textit{LOCC state tomography}, showing that any (global) tomography algorithm can be converted into an LOCC one incurring in a very small ($O(n/\varepsilon^2)$) computational overhead.
This first result introduces a fundamentally new computational interpretation of the min-entropy in entanglement theory -- which before was only apparent in a single-shot, information-theoretic setting~\cite{Buscemi_2013} -- as the ultimate achievable rate for computationally bounded entanglement distillation.

In contrast, the computational 
entanglement cost behaves quite differently. We demonstrate the 
existence of states for which no computationally efficient dilution 
protocol can achieve a dilution rate lower than \( \tilde{\Omega}(n) \), irrespective of the value of the von Neumann entropy, which can be as small as \( o(1) \). This implies that, in the worst case, efficient state dilution requires maximal consumption of ebits, making the simplest dilution protocol—quantum teleportation—optimal in this setting (see Fig.\ 1). The same limitations derived for the entanglement cost also apply to the task of \textit{state compression} \cite{PhysRevA.51.2738}, where Alice wants to transmit a quantum source to Bob using as few noiseless quantum channel uses as possible.

Crucially, both computationally efficient protocols mentioned above, for distilling and diluting entanglement, are state-agnostic; that is, they do not require any prior knowledge of the state to be performed. This leads us to the second recurring theme of our work: when computational constraints are imposed, the explicit knowledge of the state, even an efficient classical description in terms of the native gates preparing the state, does not enhance the ability to manipulate entanglement efficiently. Instead, state-agnostic protocols again achieve the optimal worst-case rate.

Taken together, our results reveal that in entanglement theory, quantitatively similar computational limitations can arise from fundamentally distinct mechanisms. These include: (1) the state having a computationally complex Schmidt basis, which inherently limits the feasibility of any efficient LOCC protocol; (2) the state admitting an efficient classical description, yet the derivation of an LOCC protocol from this description being computationally intractable; and (3) the absence of prior information about the state, where even unbounded computational power—constrained only by limited sample complexity—fails to extract sufficient information to significantly enhance LOCC performance beyond the established optimal bounds.

Finally, our findings fit within the broader framework of entanglement reversibility. In standard bipartite pure-state entanglement theory, reversibility holds: asymptotically, the entanglement extracted from a quantum state can be used to reconstruct the state back without loss \cite{Bennett_1996}. However, when computational constraints are introduced, this reversibility breaks down. In this setting, entanglement transformations become inherently \textit{irreversible}, with states exhibiting a reversibility ratio that scales inversely with system size $n$, approaching zero in large many-body systems. Notably, this framework uncovers the existence of \textit{computationally bound entangled states}—displaying maximal computational entanglement cost yet possessing near-zero computational distillable entanglement. This phenomenon, which in the standard information-theoretic setting applies exclusively to mixed states \cite{Horodecki_1998}, is thus revealed to have a counterpart in computationally constrained entanglement theory.

The rest of the paper is organized as follows: in \cref{main:B} we set up the stage and introduce the measures of computational distillable entanglement and computational entanglement cost; in \cref{main:C} we present our main results on computationally constrained state transformations; finally in \cref{main:D} we discuss the implications of our results on testing and measuring the von Neumann entropy, and on LOCC tomography. The formal presentation and proofs are deferred to the appendix.

\subsection{Tasks in entanglement theory and computational entanglement measures}\label{main:B}

In entanglement theory, LOCC entanglement manipulation tasks play a key role in constructing entanglement quantifiers, which consequently acquire direct operational meaning. 
In this context, the most operationally significant conversion tasks are entanglement distillation and entanglement dilution \cite{Bennett_1996}, which, in an i.i.d. and asymptotic setting define the known entanglement measures \textit{distillable entanglement} and \textit{entanglement cost} \cite{Horodecki_2009}. Beyond its foundational importance, entanglement distillation plays a critical role in enabling key quantum information tasks such as super-dense coding \cite{PhysRevLett.69.2881}, quantum teleportation \cite{PhysRevLett.70.1895} and quantum cryptography \cite{Pirandola_2020}.

In a scenario where Alice and Bob have limited computational power, it is natural to consider families of states parameterized by a scaling number of qubits $n$, as computational efficiency can only be assessed in relation to the scaling of $n$. While \cref{app:preliminaries} provides a rigorous introduction to these concepts, in what follows we will implicitly assume that we are working with families of states and omit this specification for simplicity.

Similarly as for the information-theoretic case, we can proceed by introducing entanglement manipulation tasks under computational restrictions, which, in turn, lead to the definition of the corresponding \textit{computational entanglement measures}. Let us consider a system of $n$ qubits, and an extensive bipartition $A|B$. 
In its most general form, the task of \textit{$k$-shot entanglement distillation} consists in converting $k$ copies of a bipartite state $\ket{\psi_{A,B}}$ into an approximate source of purified entanglement in the form of ebits $ \ket{\phi^{+}} =   (\ket{0,0} + \ket{1,1})/\sqrt{2}$. The maximum extractable rate of ebits, when starting with $k$ input copies, defines the $k$-shot distillable entanglement $E_D^{(k)}(\psi_{A,B})$ of the state $\ket{\psi_{A,B}}$. 
Here, by approximate, we mean that the protocol may have a non-zero (but small) failure probability, and may produce ebits with a small (constant) error in trace distance. 

When considering entanglement distillation in a computationally restricted setting, this limitation arises at two levels: first, the LOCC distillation protocol must be \textit{sample-efficient}, that is the number of copies received as input must scale at most polynomially with the system size, $k = O(\poly n)$; second, the LOCC operation must be \textit{efficiently implementable}, with running time $T=O(\poly n)$. With these constraints, we can characterize the measure of computational distillable entanglement as
follows.

\begin{definition}[Computational distillable entanglement. Informal of \cref{def:distillableentanglement}]\label{def:dist}
Let $\ket{\psi_{A,B}}$ be a bipartite pure state defined on a sufficiently large number of qubits $n$. Given a number of input copies $k= O(\poly n)$,
the computational distillable entanglement, $\hat{E}_{D}^{(k)}(\psi_{A,B})$, is defined as the maximum rate at which approximate ebits can be extracted using LOCC protocols with running time $T=O(\poly n)$. 
\end{definition}
We now turn to introducing another fundamental entanglement manipulation task, that is entanglement dilution, which serves as the converse operation to distillation. The task of \textit{$k$-shot entanglement dilution} involves the approximate preparation of $k$ copies of a target bipartite state $\ket{\psi_{A,B}}$ via LOCC, starting from a source of pre-shared ebits. The minimum rate of consumed ebits defines the $k$-shot entanglement cost $E_C^{(k)}(\psi_{A,B})$ of the state $\ket{\psi_{A,B}}$. 

This definition again naturally leads to a computationally restricted version. Here, however, beyond requiring that the number of copies and the running time of the protocol scale efficiently, we explicitly want to focus on the computational resources needed to distribute the state across the bipartition $A|B$, rather than to prepare it from scratch. To account for this, we assume that Alice (or, equivalently, Bob) has local access to $O(\poly n)$ samples of the state to be diluted. With this setup, we define the computational entanglement cost as follows. 

\begin{definition}[Computational entanglement cost. Informal of \cref{def:entanglementcost}]\label{def:cost}
Given local sample-access to poly-many copies of a state $\ket{\psi_{A,B}}$ (defined on a sufficiently large number $n$ of qubits), and a target of $k= O(\poly n)$ output copies, the computational entanglement cost, $\hat{E}_{C}^{(k)}(\psi_{A,B})$, is defined as the minimum rate of ebits necessary to approximately dilute the state across the bipartition $A|B$, using LOCC 
protocols with running time $T=O(\poly n)$.
\end{definition}

Note that our \cref{def:dist,def:cost} differ significantly from the ones considered in Ref.~\cite{arnonfriedman2023computationalentanglementtheory} in that we consider the more generic $k$-shot setting rather than single-shot case. Furthermore, contrary to our case, the definition of entanglement cost in Ref.~\cite{arnonfriedman2023computationalentanglementtheory} does not allow Alice and Bob to have local sample access to the target state. However, we think this is necessary to avoid counting computational resources that are only due to local preparation and are independent of the entanglement dilution task itself.

As previously discussed, entanglement manipulation tasks are closely tied to entanglement quantifiers. For bipartite pure states, a family of entanglement monotones is provided by the \textit{R\'enyi entropies of entanglement} $S_{\alpha}(\psi_A) \coloneqq (1-\alpha)^{-1} \log\tr\psi_A^{\alpha}$ with $\alpha \ge0$ \cite{Vidal_2000}. Among these, particularly notable are the von Neumann entropy $S_{1}(\psi_A) = -\tr(\psi_A \log \psi_A)$
($\alpha \to 1^+$), the min-entropy, $S_{\min}(\psi_A) = -\log\|\psi_A\|$ ($\alpha\to\infty$) and the max-entropy $S_{0}(\psi_A) = \log \rank(\psi_A)$ ($\alpha \to 0$). The R\'enyi entropies are monotonically decreasing in $\alpha$, with the min-entropy being the smallest. In the absence of computational constraints, the von Neumann entropy determines the optimal asymptotic rates for both pure-state distillation and dilution tasks when the number of copies approaches infinity: $\lim_{k \to \infty} E_{D}^{(k)}(\psi_{A,B}) = \lim_{k \to \infty} E_{C}^{(k)}(\psi_{A,B}) = S_{1}(\psi_A)$ \cite{Bennett_1996}, which implies the asymptotic reversibility of the theory.

Furthermore, more recent results show that the von Neumann entropy rate is actually reached with $k=O(\poly n)$ shots \cite{Buscemi_2013, Fang_2019,theurer2023singleshotentanglementmanipulationstates}, albeit without taking into account any computational limitations (see \cref{app:preliminaries}). Thus, sample complexity alone does not pose any fundamental limitation on entanglement manipulation. Meanwhile, smoothed versions of the min-entropy and the max-entropy govern the single-shot regime, with the smoothed min-entropy quantifying the single-shot distillable entanglement and the smoothed max-entropy quantifying the single-shot entanglement cost \cite{Buscemi_2011,Buscemi_2013}.

Having outlined the settings of interest and the associated computational entanglement measures, we now proceed to present the main results on computationally constrained entanglement manipulation. 

\subsection{Entanglement theory with limited computational resources}\label{main:C}
We now aim to determine the fundamental limits and capabilities of computationally efficient entanglement manipulation. Here the goal is to precisely characterize the computational distillable entanglement and the computational entanglement cost by establishing bounds on the optimal rates and constructing LOCC protocols that achieve them.

However, when revisiting entanglement theory through a computational lens, another crucial aspect becomes evident: the optimal conversion rates we introduced implicitly assume that Alice and Bob have \textit{perfect knowledge} of the state. While the mere existence of a computationally efficient protocol may provide achievability bounds on computational entanglement measures, this assumption contradicts the perspective adopted in this work. Having perfect knowledge of the state means fully characterizing it, but given the exponential growth of the Hilbert space dimension with the number of qubits $n$, this lies fundamentally beyond the realm of computational feasibility. For entanglement manipulation tasks to remain meaningful in a computationally constrained scenario, there are essentially two possibilities:
\begin{enumerate}[label=(\roman*)]
    \item \textit{State-agnostic}: Alice and Bob have access to (at most polynomially many) copies of an \textit{unknown} state they can manipulate and characterize the entanglement of.
    \item \textit{State-aware}: Alice and Bob possess an efficient \textit{classical description} of the state to manipulate and characterize the entanglement of.
\end{enumerate}

For what concerns efficient entanglement distillation, the state-agnostic scenario limits Alice and Bob to gather information on the unknown state $\ket{\psi_{A,B}}$ through efficient local measurements on $O(\poly n)$
samples, which inform their LOCC distillation protocol but reduce the distillation rate rate due to measurement-induced consumption of copies. In the state-aware scenario, instead, they have full circuit-level knowledge of the state but may still face the computational intractability of extracting key properties like the Schmidt eigenbasis and coefficients, crucial for optimal LOCC design.

For entanglement dilution, the settings are analogous. In the state-agnostic case, Alice and Bob have local access to 
$O(\poly n)$ copies of an unknown bipartite state, and use measurements to inform an LOCC protocol that minimizes ebit consumption. Unlike distillation, measured copies here do not count against resources. In the state-aware case, they additionally have the state’s classical description and aim to design an efficient LOCC protocol based on this information.

Having cleared the framework of entanglement manipulation in the computationally constrained regime, we now quantify the limits and possibilities of the state-agnostic and state-aware scenarios, starting with the former. Surprisingly, our analysis reveals a deep connection between the two: in the worst case, state-awareness offers no advantage, as the fully state-agnostic approach is fundamentally optimal.

\subsubsection{State-agnostic scenario}

We begin our investigation of the state-agnostic setting by considering first entanglement distillation. We note that already previous works have highlighted some limitations in the computationally constrained setting. Specifically, the discovery of \textit{pseudoentangled quantum states} \cite{aaronson2023quantumpseudoentanglement} reveals computationally indistinguishable families of bipartite pure states with an entanglement gap of \( \Omega(n) \) vs. \( O(\text{polylog } n) \), establishing that no state-agnostic protocol can distill more than \( O(\log n) \) ebits in the worst case. Our work significantly goes beyond this result, showing that this bound is overly optimistic and cannot be achieved by any state-agnostic protocol in many scenarios.

The crux of our no-go results is the construction of families of pseudoentangled states with a \textit{maximal entanglement gap} of \( \tilde{\Omega}(n) \) vs. \( o(1) \) (\cref{eq:pseudoentangledconstruction}) across a fixed extensive bipartition \( A|B \). This has a profound implication: state-agnostic protocols cannot distill any ebits from general quantum states, even when the von Neumann entropy is maximal, thus strengthening existing bounds on such protocols' limitations. At the same time, these pseudoentangled families exhibit \( o(1) \) entanglement when quantified by the min-entropy \( S_{\min} \). This suggests that the min-entropy is a key quantity determining the fundamental limitations of distillation protocols in computationally constrained, state-agnostic scenarios. Building on this, we refine our construction to create pseudoentangled families with tunable min-entropy \( S_{\min} \), leading to the formulation of our main state-agnostic no-go result.

\begin{theorem}[No-go on state-agnostic distillation. Informal version of \cref{th:upperbounddistillableentanglement}]\label{th:informal1}
Any sample-efficient state-agnostic approximate distillation protocol cannot distill more than $\min\{S_{\min}\!+\!o(1), 2\log k\!+\!O(1)\}$ ebits from $k$ copies of a bipartite input state, regardless the value of the von Neumann entropy $S_1$. 

\begin{proof}[Proof sketch]
    The proof reduces to the construction of two statistically indistinguishable families of states with the same min-entropy $S_{\min}$ and a tunable von Neumann entropy $S_{1}$ gap, which can be maximal ($\tilde{\Omega}(n)$ vs.\  $O(S_{\min})$). The construction is relatively straightforward as it considers states of the form
\be
\ket{\Psi_{A,B}} = \sqrt{1-\eta} \ket{0}_{A} \ket{0}_{B} \ket*{0}^{\otimes n-2S_{\min}} |\phi^{+}_{A,B}\rangle^{\otimes S_{\min}} 
&+ \sqrt{\eta} \ket{1}_{A} \ket{1}_{B} \ket*{\psi_{A,B}},\label{eq:pseudoentangledconstruction}
\ee
where $\ket{\psi_{A,B}}$ may encode either a Haar random state or a suitable pseudoentangled quantum state. The parameter $\eta\in [0,1]$ is appropriately chosen to tune the entanglement gap. Since the two families of states are statistically indistinguishable and have a tunable entanglement gap, any state-agnostic distillation rate cannot exceed the lowest value of $S_1$ among the two, which coincides with $S_{\min}$ or $\log k$ for appropriate choices of the pseudoentangled families. For the complete rigorous proof, we refer to \cref{App:stateagnostic}.\end{proof}
\end{theorem}

Somewhat unexpectedly, even in the regime of low entanglement $S_{1} = O(\log n)$, the von Neumann entropy is completely inaccessible to any sample-efficient state-agnostic protocol and the distillable entanglement is instead captured by the min-entropy. In the particular case where $S_{\min} = \Omega(\text{polylog } n)$ we recover, as a special case, the results of Ref.\ \cite{aaronson2023quantumpseudoentanglement}, as the bound in \cref{th:informal1} simplifies to $O(\log k)$, which coincides with $O(\log n)$ when $\poly n$ input copies are available

Now, a crucial question becomes evident, does an efficient state-agnostic protocol exist that achieves the limitations imposed by \cref{th:informal1}? We hereby provide a positive answer by presenting a finite-shot regime analysis of a protocol originally introduced in the asymptotic setting in Ref.\ \cite{PhysRevA.75.062338}. This protocol is state-agnostic and relies on the Schur transform, which leverages permutation symmetry to decompose a $k$-fold tensor product state onto irreducible subspaces of the symmetric and unitary groups \cite{weyl1946classical}. Crucially, this decomposition features maximally entangled states in orthogonal (irreducible) subspaces, that can thus be extracted via local operations without perturbation. Remarkably, the Schur transform can be implemented efficiently on a quantum computer when $k = O(\poly n)$ \cite{Fang_2019}. While the complete description of the distillation protocol requires several technical preliminaries and is deferred to \cref{App:stateagnostic}, here we present the main takeaway result in the following informal theorem.

\begin{theorem}[State-agnostic distillation performance. Informal version of \cref{th:distillation}]\label{th:informal2}
There exists an approximate, state-agnostic and computationally efficient distillation protocol that can distill at a rate $\min\{\Omega(S_{\min}(\psi_A)), \Omega(\log k)\}$ from any given input state $\ket{\psi_{A,B}}$. In other words, for any state we have
\be
    \hat{E}_{D}^{(k)}(\psi_{A,B})\ge \min\{\Omega(S_{\min}(\psi_A)),\Omega(\log k)\}.
\ee
\end{theorem}

Taken together, \cref{th:informal1,th:informal2} establish, for the first time, the ultimate rates for state-agnostic entanglement distillation when given access to $k = O(\poly n)$ copies of the input state. Furthermore, since the optimal protocol is computationally efficient, its optimality holds even under (more strict) computational constraints.
These results provide a definitive benchmark for what is achievable in the state-agnostic regime, shedding light on its fundamental limitations and possibilities. Let us note that, when $k = \Omega(\exp n)$, the protocol considered in \cref{th:informal2} achieves the von Neumann entropy rate~ \cite{PhysRevA.75.062338}, therefore state-agnostic protocols are information-theoretically optimal when the sample-complexity scales exponentially.

Clearly, the no-go result in \cref{th:informal1} leaves open the existence of agnostic protocols that may achieve higher rates on some restricted classes of states. For example, as we show in \cref{App:classesofstates}, this is the case for \textit{low-rank states}, for which the protocol analysed in \cref{th:informal2} distills at the (maximal) von Neumann entropy rate.  Notably, these states may not be efficiently simulable. This reveals a gap between the simulability of quantum states and the efficient manipulation of their entanglement. This has already been noted in 
Ref.\ \cite{gu2024magicinducedcomputationalseparationentanglement}, where optimal and computationally efficient entanglement manipulation of \textit{entanglement-dominated t-doped stabilizer states} has been shown beyond the classical simulability barrier.

Since the results in \cref{th:informal1,th:informal2} only cover the state-agnostic scenario, an immediate question is: are there states for which the state-agnostic rate above is optimal among all computationally efficient protocols? 
Surprisingly, the answer is positive, as stated by the following theorem.

\begin{theorem}[Upper bound on the computational distillable entanglement. Informal version of \cref{appth:upperbounddistillableent}]\label{th:upperbounddistillableent}
Let $n$ be sufficiently large, then there exists at least a state $\ket{\psi_{A,B;n}^{*}}$ with von Neumann entropy $\tilde{\Omega}(n)$ and computational distillable entanglement $\hat{E}_{D}^{(k)}\le \min\{S_{\min},\log k\} + o(1)$, for any $k=O(\poly n)$
    \begin{proof}[Proof sketch]
        The proof relies on the key fact that the pseudoentangled state families in \cref{eq:pseudoentangledconstruction}, while not being pseudorandom, exhibit strong pseudorandom components. This leads to a \textit{concentration of measure} phenomenon, ensuring that any LOCC application yields results close to the Haar average. Since computationally efficient LOCCs are  ``only" $2^{O(\poly n)}$ in number, while the concentration is doubly exponential, any efficient LOCC must produce an output near the average. Finally, as the Haar average is statistically close to the pseudorandom one and the two families have a tunable entanglement gap, the distillation rate is upper bounded by the family with the lowest entanglement content. See \cref{App:tightdistillable} for details.
    \end{proof}
\end{theorem}
Together, \cref{th:informal2,th:upperbounddistillableent} provide a complete worst-case characterization of the computational distillable entanglement. Notably, they demonstrate that there exist states for which the computational distillable entanglement precisely satisfies 
\be
\hat{E}_{D}^{(k)} = \min\{\Theta(S_{\min}),\Theta(\log k)\},
\ee
regardless the value of the von Neumann entropy $S_1$. 
This result reveals that in a computationally constrained setting, the min-entropy takes on a pivotal role, replacing the von Neumann entropy as the key quantity traditionally associated with asymptotic entanglement distillation. Moreover, it reaffirms the optimality of the state-agnostic protocol discussed in \cref{th:informal2}.

We now turn to describing our results for the converse task: state-agnostic entanglement dilution. Here, the simplest state-agnostic dilution procedure is the \textit{quantum teleportation} protocol \cite{PhysRevLett.70.1895}, which, starting from local copies of the state to be diluted, consumes $n$ ebits per copy to distribute the state non-locally across the bipartition. While quantum teleportation seems highly inefficient in terms of the entanglement consumed, our novel construction of pseudoentangled states (see \cref{eq:pseudoentangledconstruction}) implies that quantum teleportation is in fact optimal in terms of ebits consumption, even when the von Neumann entropy is $o(1)$. This is summarised in the following theorem.

\begin{theorem}[No-go on state-agnostic entanglement dilution. Informal version of \cref{th:entanglementcostnogo}]\label{th:informal4} 
Any sample-efficient state-agnostic approximate dilution protocol cannot dilute at a rate less than $\tilde{\Omega}(n)$, irrespective of the value of the von Neumann entropy, which can be $o(1)$.
\end{theorem}
\cref{th:informal4} establishes that quantum teleportation, being state-agnostic and computationally efficient, is optimal in terms of ebits consumption even for states with vanishingly small entanglement, i.e., $S_{1}=o(1)$.

Similar to the case of entanglement dilution, one might ask whether there exist states for which quantum teleportation remains the optimal computationally feasible approach to entanglement dilution. In the following theorem, we address this question by demonstrating the existence of states with a computational entanglement cost lower bounded by $\Tilde{\Omega}(n)$, regardless of their von Neumann entropy.

\begin{theorem}[Tight lower bound on computational entanglement cost. Informal version of \cref{appth:lowerboundentcost}]\label{th:lowerboundentcost}
Let $n$ be sufficiently large, then there exists at least a state $\ket{\psi_{A,B;n}^{*}}$ with computational entanglement cost $\hat{E}_{C}^{(k)}\ge \tilde{\Omega}(n)$, irrespective of the (arbitrarily small) value of the von Neumann entropy $S_{1}$
\end{theorem}
As a direct consequence of \cref{th:lowerboundentcost}, and of quantum teleportation, we achieve a full worst-case characterization of the computational entanglement cost. Specifically, we show the existence of states for which $\hat{E}^{(k)}_{C} = \tilde{\Theta}(n)$, despite having an arbitrarily low information-theoretic entanglement cost, as quantified by the von Neumann entropy. 

As a corollary of \cref{th:lowerboundentcost}, the same worst-case bounds apply to the task of \textit{state compression} \cite{PhysRevA.51.2738}, where Alice compresses many i.i.d. copies of a mixed quantum source $ \rho_A$ to minimize the use of a noiseless quantum channel for transmitting the source to Bob. Asymptotically, the optimal compression rate is again given by \( S_1(\rho_A) \) \cite{PhysRevA.51.2738}. Since state compression can serve as a dilution protocol—with Alice first encoding the state and then teleporting its compressed version-- the no-go result of \cref{th:lowerboundentcost} (as well as any other no-go on the entanglement cost) applies directly (for more details, see \cref{sec:statecompression}).

\subsubsection{State-aware scenario}\label{sec:state-awarescenario}
Having examined the state-agnostic scenario, we now consider the state-aware setting, where the agents have an efficient classical description of the state they wish to manipulate. However, as we anticipated, this knowledge does not grant efficient access to the optimal state conversion protocol, as computing key figures of merit may remain unfeasible. The next theorem formalizes this intuition, showing that even with a classical description, designing a better performing LOCC protocol than in the state-agnostic case remains computationally intractable in the worst case. This result holds under the assumption that the \emph{learning with errors} (LWE) decision problem is computationally hard to solve for a quantum computer (see \cref{app:lwe}), a widely adopted assumption in post-quantum cryptography \cite{TCS-074}.

\begin{theorem}[No-go on state-aware entanglement manipulation. Informal version of \cref{thapp:stateawaredistillation,thapp:stateawaredilution}]\label{th:nogostateawaredistillation}
  Let $\delta>0$. Assuming that the LWE problem is superpolynomially hard to solve on a quantum computer, the following two statements are true.
  \begin{itemize}[leftmargin=5mm]
      \item No computationally efficient protocol can be efficiently designed to distill at a rate larger than $S_{\min}(\psi_A) + o(1)$ from $k$ copies of a bipartite pure quantum state, even when provided with an efficient circuit description of the state. This limitation holds regardless of the value of $S_{1}(\psi_A)~=~\tilde{O}(n^{1-\delta})$, and for $S_{\min}(\psi_A)~=~O(n^{\delta})$.
      \item Any computationally efficient protocol efficiently designed to dilute $k$ copies of a bipartite state, even with access to its efficient circuit description, must consume more than $\Omega(n^{1-\delta})$ ebits. This result holds regardless of the value of $S_{1}(\psi_A) = O(n^{1-\delta})$.
  \end{itemize}
\end{theorem}

This reveals a surprising insight: while the state-aware scenario seems advantageous, it does not necessarily improve conversion rates over the fully state-agnostic case. For certain state classes, the information needed to design superior LOCC protocols remains computationally inaccessible, even with a classical description.  

Notably, \cref{th:nogostateawaredistillation} does not rule out the \textit{existence} of efficient LOCC protocols for entanglement manipulation. Unlike \cref{th:upperbounddistillableent} and \cref{th:lowerboundentcost}, it imposes no bounds on computational entanglement measures but prevents Alice and Bob from efficiently identifying such a protocol, even when starting with a classical description.  

Despite its worst-case limitations, \cref{th:nogostateawaredistillation} does not preclude improvements in some specific scenarios. For instance, Ref.\ \cite{gu2024magicinducedcomputationalseparationentanglement} shows that for t-doped stabilizer states with sufficiently high entanglement, knowing the associated stabilizer group enables computationally efficient distillation and dilution at the von Neumann entropy rate.

\subsection{Learning tasks in entanglement theory}\label{main:D}
In this final section we leverage on our previous findings to analyze the ultimate computational feasibility of other fundamental tasks in entanglement theory. In particular, we focus on the \textit{learning tasks} of \textit{estimating} and \textit{testing} the von Neumann entropy, and on LOCC \textit{tomography}. In the entropy estimation task, an agent is given access to $k$ copies of an unknown bipartite pure state $\ket{\psi_{A,B}}$, and the goal is to measure the von Neumann entropy $S_1(\psi_{A})$ with additive precision $\varepsilon$. In the entropy testing task, on the other hand, one is provided with states belonging to two distinct classes: one with von Neumann entropy $\leq\alpha$, and the other with von Neumann entropy $\geq\beta$, where $\alpha<\beta$. The task is to distinguish between these two classes with high success probability. Our primary focus will be to determine the necessary number of input state copies, $k$, required to successfully accomplish either of these tasks for any given range of the von Neumann entropy. Instead, for what concerns LOCC tomography, the objective is for Alice and Bob to construct a classical description of an unknown bipartite pure state up to a certain error in trace distance, given access to $k$ copies of the state and being restricted to performing only LOCC measurements.
First, by leveraging our new construction of pseudoentangled states (\cref{eq:pseudoentangledconstruction}), we derive an explicit and nearly tight lower bound for the tasks of entropy estimation and testing, irrespective of the (fixed) range of the von Neumann entropy.

\begin{corollary}[Measuring and testing von Neumann entropy. Informal version of \cref{cor:lowerboundsestimatingvnentropy}]\label{cor:informal4}
    Given sample access to an unknown bipartite pure state $\ket{\psi_{A,B}}$, $\Omega(2^{n/2})$ samples are necessary to estimate the von Neumann entropy $S_{1}$ of the reduced density matrix up to a constant additive error, even when $S_{1}(\psi_A) = O(1)$. Similarly, testing whether $S_{1}(\psi_A)\le\alpha$ or $S_{1}(\psi_A)\ge \beta$ requires $\Omega(2^{\frac{\alpha}{2\beta}n})$ many copies.
\end{corollary}

Our result on entanglement detection generalizes known results \cite{9163139} by demonstrating that estimating the von Neumann entropy requires exponential sample complexity, even for states with very low $O(1)$ entanglement. Regarding entanglement testing, to the best our knowledge our result is new and extends the known classical result (holding for the Shannon entropy) shown by Valiant \cite{doi:10.1137/080734066}. \cref{cor:informal4} solidifies even more a key theme emerging throughout our work, which is that the von Neumann entropy, in a computationally limited scenario, loses its operational, asymptotic meaning and cannot be efficiently detected, even when the entanglement is very low. In contrast, the entanglement measure that a computationally bounded agent perceives is the min-entropy, or more generally, any R\'enyi entropy with \(\alpha \geq 2\). Notably, the operator norm $\|\psi_A\|$ is efficiently measurable to additive precision 
\(\varepsilon\) using a quantum algorithm that scales as \(O(1/\varepsilon^2)\), as outlined in Ref.~\cite{odonnell2015efficientquantumtomography}.

Finally, as our concluding result, we demonstrate that LOCC state tomography does not impose a significantly higher cost compared to standard tomography. Specifically, thanks to the existence of a computationally efficient agnostic distillation protocol (\cref{th:informal2}), Alice and Bob can transform any global tomography protocol into one based solely on LOCC, incurring only in a linear overhead in the system dimension.

\begin{theorem}[Efficient LOCC tomography. Informal version of \cref{th:LOCCtomography}]\label{th:LOCCtomography-informal} Given an unknown bipartite pure state $\ket{\psi_{AB}}$ of $n$ qubits, any (global) tomography algorithm with sample complexity of $K_n$ and precision $\varepsilon$ can be converted into an $\varepsilon$-precise LOCC tomography protocol with sample complexity $O(nK_n/\varepsilon^2)$. 
\end{theorem}

In particular, this result shows that the optimal sample-complexity bounds for pure-state tomography of $\Omega(d/\varepsilon^2)$ are achieved with regard to the dependency in $d$, up to a logarithmic factor, even by LOCC protocols. Furthermore, any efficiently learnable class of states can be learned efficiently via LOCC. This last result also rigorously shows that for efficiently learnable classes of states, computationally efficient state-agnostic and state-aware protocols are essentially equivalent, as the agents can allocate a small fraction of the input copies to gain precise information on the unkown input, without negatively impacting the protocol's performance.

\subsection{Discussion and open questions}
In this work, we have explored the foundations of pure-state entanglement theory within the computationally efficient realm, establishing the ultimate rates for efficient entanglement manipulation with matching achievability and optimality results. Our bounds highlight a (possibly) maximal gap between information-theoretic and computational transformation rates, showing that, independently of the entanglement content, the von Neumann entropy loses its operational value and becomes inaccessible to computationally bounded agents. Instead, the min-entropy appears as the fundamental limit governing optimal distillation rates. This work underscores the critical importance of considering computational resources in quantum information tasks, which are often assumed to be unlimited.

Our results significantly go beyond any previous findings. First, while Ref.\ \cite{aaronson2023quantumpseudoentanglement} has established a state-agnostic no-go for entanglement distillation beyond a rate of \(O(\log n)\), here, we demonstrate that this is overly optimistic by showing a maximal separation of $\tilde{\Omega}(n)$ vs.~$o(1)$ instead, leveraging a fundamentally new construction of pseudoentangled states with a maximal entanglement gap. We emphasize that this finding does not contradict the conclusions of Refs.~\cite{aaronson2023quantumpseudoentanglement,bouland2023publickeypseudoentanglementhardnesslearning}, where the authors state that the largest pseudoentangled gap should be \(\Omega(n)\) vs.~\(O(\text{polylog } n)\). A closer look at Ref.~\cite{aaronson2023quantumpseudoentanglement} reveals that this no-go result applies only when pseudoentangled states are also required to be \textit{pseudorandom} (i.e., statistically indistinguishable from Haar random states). However, in our case we find this requirement to be unnecessary to our results. Second, while the authors of Ref.~\cite{arnonfriedman2023computationalentanglementtheory} have derived no-go results in the single-shot setting, here we go strictly beyond by considering the more general many-shot regime. In particular, the single-shot (worst-case) no-go results of \cite{arnonfriedman2023computationalentanglementtheory} are easily derived from our theorems as a corollary. This extension allows for a direct comparison to the asymptotic rates given by the von Neumann entropy (actually achievable with poly-many copies), which in the single-shot case is inaccessible even information-theoretically \cite{Buscemi_2011,Buscemi_2013,theurer2023singleshotentanglementmanipulationstates}.  Furthermore, while Refs.~\cite{aaronson2023quantumpseudoentanglement,arnonfriedman2023computationalentanglementtheory} focused solely on no-go results, here we present, for the first time, non-trivial achievable rates (given by $S_{\min}$) for computationally efficient entanglement distillation, by providing an explicit, computationally efficient, and state-agnostic protocol. Our construction also shows that improvements upon the \textit{public-key pseudoentanglement} one of Ref.~\cite{bouland2023publickeypseudoentanglementhardnesslearning} can be very easily obtained, as shown in \cref{th:nogostateawaredistillation}. In particular, we introduced two quantum circuit classes with an entanglement gap of \(\Omega(n^{1-\delta})\) vs.~\(o(1)\) (for \(\delta > 0\)), surpassing the previous \(\Omega(n)\) vs.~\(O(n^{\delta})\) gap under the assumption that LWE is superpolynomially hard for quantum computers. This gap strengthens to \(\tilde{\Omega}(n)\) vs.~\(o(1)\) under the stronger assumption that no subexponential quantum algorithm can solve LWE. This has direct implications for determining the ground state entanglement of local Hamiltonians, reinforcing that even if the ground state is nearly a product state, a classical Hamiltonian description cannot determine its entanglement. This strengthens the prior results of Ref.\ \cite{bouland2023publickeypseudoentanglementhardnesslearning}. However, as with previous constructions, our pseudoentanglement gap does not extend to most cuts. Whether such a large entanglement gap can be achieved across most cuts remains an open question.

For what concerns learning tasks, to the best of our knowledge, our results on entropy estimation and testing also represent a fundamentally new contribution. While the sample-complexity hardness of these tasks for large entropy values is well-established \cite{9163139, wang_et_al:LIPIcs.ESA.2024.101,Lin_2018}, no prior result extends to $O(1)$ values of the entropy. Finally, while our results on efficient LOCC tomography may be easily derived when looking at some specific tomography protocols, to our knowledge it is the first time that such a direct and fundamental connection is shown.

To conclude, while the problem of determining optimal information-theoretic conversion rates has traditionally been studied within the framework of \textit{quantum resource theories} \cite{Chitambar_2019}, it often ignores computational constraints. Our work advocates for a complementary approach based on \textit{computational resource theories}, where efficient computation plays a fundamental role. This shift could open new directions in quantum communication \cite{Gisin_2007}, quantum thermodynamics \cite{Vinjanampathy_2016}, quantum coherence \cite{RevModPhys.89.041003}, non-Gaussianity \cite{PhysRevA.97.062337}, and magic \cite{Veitch_2014}. For what concerns entanglement under computational constraints, we provide a comprehensive characterization of the bipartite pure-state scenario. While our no-go results trivially extend to mixed states, a key open question is whether entanglement can be efficiently manipulated in states that are not excessively mixed. Ref.~\cite{bansal2024pseudorandomdensitymatrices} shows infeasibility for states with superpolynomially small purity, but the possibility remains for those with at most polynomially small purity. Whether an entropy-like measure, akin to min-entropy, plays a fundamental role in this regime remains an open question. Multipartite entanglement presents further challenges. While its manipulation is difficult even information-theoretically~\cite{Sauerwein_2018}, certain structured states—such as stabilizer states and their t-doped extensions—allow computationally efficient transformations~\cite{fattal2004entanglementstabilizerformalism,gu2024magicinducedcomputationalseparationentanglement}. This suggests computational efficiency could offer new insights into multipartite entanglement theory, unveiling novel avenues for exploration.

\section*{Acknowledgments}
We thank Ryuji Takagi, Philippe Faist, 
Nathan Walk, Salvatore F.E. Oliviero, Johannes Jakob Meyer, Sumeet Khatri, Ludovico Lami, Antonio Anna Mele, 
Lennart Bittel, Andi Gu, Francesco Anna Mele for inspiring and thought-provoking 
discussions. 
This project has been funded by the 
BMBF (QR.N, DAQC, MuniQC-Atoms, Hybrid++), the BMWK (EniQmA),
the Munich Quantum Valley (K-8),
Berlin Quantum, the QuantERA (HQCC),
and the European Research Council 
(DebuQC).

\appendix
\section*{Supplemental material}

\resumetoc

\tableofcontents
\section{Preliminaries}\label{app:preliminaries}
\subsection{General definitions}\label{app:general}
In this work, we will exclusively deal with finite-dimensional Hilbert spaces of $n$ qubits $\mathcal{H} \cong \mathbb{C}^{\otimes n}$, equipped with the standard scalar product. The respective dimension will be denoted as $\dim \mathcal{H}$. We indicate the set of linear maps between two Hilbert spaces $\mathcal{H}$ and $\mathcal{H}'$ as $\mathcal{B}(\mathcal{H} \to \mathcal{H}')$, and more coincisely $\mathcal{B}(\mathcal{H})$ whenever $\mathcal{H} = \mathcal{H}'$. We say that $\rho \in \mathcal{B}(\mathcal{H})$ is a \textit{quantum state} if $\rho$ is Hermitian, $\rho \geq 0$ and $\tr\rho =1$. In the manuscript, lowercase Greek letters will always denote quantum states. We say that a quantum state $\psi$ is pure if $\rank \psi = 1$. Pure states will also be denoted as $\ketbra{\psi}$, or with the corresponding vector $\ket{\psi}$. In a mild but common abuse of notation, we will 
also call state vectors $\ket{\psi}$ pure states. A bipartition $A|B \coloneqq (n_A,n_B)$ with $n_A + n_B = n$ equips $\mathcal{H}$ with a tensor product structure $\mathcal{H} = \mathcal{H}_A \otimes \mathcal{H}_B$ with $\dim\mathcal{H}_A = 2^{n_A}$, and $\dim\mathcal{H}_B = 2^{n_B}$. We will mostly consider \textit{extensive bipartitions} with $n_A,n_B = \Theta(n)$. We denote a bipartite state as $\rho_{A,B}$ and its corresponding reduced density matrix on system A (respectively B) as $\rho_A = \tr_B \rho_{A,B}$ (respectively $\rho_B = \tr_A \rho_{A,B}$). When the state is pure, we will often write $\psi_A$ instead of $\rho_A$. Other common objects in quantum information theory are measurements and channels. A \textit{quantum measurement} is specified by a set of operators $\{M_i\}_i$, for some $M_i \in \mathcal{B}(\mathcal{H} \to \mathcal{H}')$ such that $\sum_i M^\dag_i M_i = \id_{\mathcal{H}}$. The un-normalized output state when measurement $M_i$ is successful can be written as $M_i\rho M_i^\dag$, while the probability of outcome $i$ reads $p(i) = \tr M_i^\dag M_i\rho$. If, furthermore, the measurement operators form a set of orthogonal projectors, satisfying $M_i = M_i^\dag$ and $M_iM_j = \delta_{ij} M_i$, we then call it a projective measurement. When the output state after the measurement is not relevant, one can equivalently consider positive-operator-valued measurements (POVM), which are specified by a set of operators $\{\Lambda_i\}_i$, with $\Lambda_i \in \mathcal{B}(\mathcal{H})$ such that $\Lambda_i \geq 0$ and $\sum_i \Lambda_i = \id_{\mathcal{H}}$, with $p(i) = \tr \Lambda_i \rho$. \textit{Quantum channels}, instead, are completely positive trace-preserving (CPTP) maps acting on operators in the Hilbert space.
In this manuscript, we will mostly deal with \textit{local operations and classical communication} (LOCC) channels defined on a bipartition $A|B$. These are composed by sequential applications of local quantum channels on systems $A$ and $B$ alternated by rounds of classical communication. Other fundamental quantities in quantum information are \textit{distance measures} between quantum states. Given an Hermitian operator $M\in \mathcal{B}(\mathcal{H})$, $\|M\|_p$ with $p\in(0,\infty)$ denotes the  $p$-norm $\|M\|_p \coloneqq \left(\tr|M|^p\right)^{1/p}$. In particular, $\|M\| \coloneqq \lim_{p\to\infty}\|M\|_p = \max_i |\lambda_{i}|$ is the operator norm, where $\lambda_i$ are the eigenvalues of $M$. Given two quantum states $\rho$ and $\sigma$, their trace distance is defined as 
\begin{equation}
D(\rho,\sigma) \coloneqq \frac{1}{2}\|\rho-\sigma\|_1.
\end{equation}
Another, equivalent (variational) characterization of the trace distance is in terms of the optimal POVM distinguishing between the two quantum states $\rho$ and $\sigma$ as $D(\rho,\sigma)=\max_{0\leq\Lambda\leq\id}\tr\Lambda(\rho-\sigma)$. As a consequence, the optimal success probability in distinguishing is given by $p_\text{succ}^\text{opt} = (1 + D(\rho,\sigma) )/2$. This is known as Helstrom bound (see, for example, 
Ref.\ \cite{1976quantum}), and implies that if $D(\rho,\sigma)< (2p-1)$ then it is not possible to distinguish the two states with success probability $\geq p$. Another common distance measure is the fidelity $F(\rho,\sigma)\coloneqq \| \sqrt{\rho}\sqrt{\sigma}\|_1^2$. The 
trace distance and the fidelity are equivalent distance measures as stated by Fuchs–van de Graaf inequalities \cite{fuchs1998cryptographicdistinguishabilitymeasuresquantum} 
\begin{equation}
1 - \sqrt{F(\rho,\sigma)} \leq D(\rho,\sigma) \leq \sqrt{1 - F(\rho,\sigma)}. 
\end{equation}
Other key quantities are measures of randomness in quantum systems. Given a bipartite state $\rho_{A,B}$, the corresponding $\alpha$-R\'enyi entropies for $\alpha \geq 0$ are defined as $S_{\alpha}(\rho_A) \coloneqq (1-\alpha)^{-1} \log \|\rho_A\|_{\alpha}^{\alpha}$. The limit $\alpha\to1^+$ gives the \textit{von Neumann entropy} $S_1(\psi_A) = -\tr(\rho_A\log\rho_A)$. The limit $\alpha\to\infty$ gives instead the \textit{min-entropy} $S_{\min} = -\log\|\rho_A\|$. The R\'enyi entropies are non-negative, upper bounded by $n_A$, and monotonically non-increasing in $\alpha$, so that in particular $S_1 \geq S_{\min}$. Other inequalities hold in the other direction for $\alpha>1$, for example $S_2 \leq 2 S_{\min}$. Note that the von Neumann entropy can be arbitrarily larger than the min-entropy.

We use Landau notation to describe the asymptotic behavior of functions, providing a way to compare their growth rates. Big-O notation, \( O(g(n)) \), represents an upper bound, meaning \( f(n) \leq c \cdot g(n) \) for some constant \( c \), for all \( n \geq n_0 \). In contrast, Big-Omega, \( \Omega(g(n)) \), denotes a lower bound, ensuring \( f(n) \geq c \cdot g(n) \) for sufficiently large \( n \), i.e., for all \( n \geq n_0 \). When a function is both \( O(g(n)) \) and \( \Omega(g(n)) \), we use Big-Theta, \( \Theta(g(n)) \), to indicate tight bounds. Finally, Little-Omega, \( \omega(g(n)) \), describes functions that grow strictly faster than \( g(n) \), as for any constant \( c \), there exists \( n_0 \) such that \( f(n) > c \cdot g(n) \) for all \( n \geq n_0 \) and $\lim\frac{g(n)}{f(n)}=0$, analogously Little-o, \( o(g(n)) \), denotes strict lower bounds on $g(n)$.

\subsection{Computational and information-theoretic entanglement measures}\label{app:compinfomeasures}
In this section, we briefly introduce the reader to two well-known tasks in quantum information processing, namely \textit{entanglement distillation} and \textit{entanglement dilution} of bipartite pure quantum states. These tasks naturally arise in the context of two distant parties, say Alice ($A$) and Bob ($B$), wanting to manipulate their pre-shared quantum entanglement using only LOCC. While entanglement distillation deals with purifying the correlations already present in a pre-shared quantum state, the task of dilution consists in locally engineering a bipartite target state starting from pre-shared pure entanglement. While entanglement dilution has a more theoretical flavour, entanglement distillation is necessary to actually harness the power of entanglement in tasks like super-dense coding \cite{PhysRevLett.69.2881}, quantum teleportation \cite{PhysRevLett.70.1895} and quantum cryptography \cite{Pirandola_2020}.
The imposed limitation of relying only on LOCC boils down to an experimentally motivated setting in which the two agents, being spatially separated, cannot rely on non-local quantum operations without employing (harder to implement) quantum communication. 

The optimal efficiency in performing entanglement distillation and dilution tasks allows one to introduce two corresponding \textit{information-theoretic entanglement measures} (i.e., quantifiers of entanglement), namely the \textit{distillable entanglement} and the \textit{entanglement cost}. After recalling these quantities, we will then move to consider these tasks in a scenario in which the two agents have the additional restriction of \textit{limited computational power}: this will lead us to introduce the corresponding \textit{computational entanglement measures}. These will quantify entanglement from the point of view of an observer having limited computational power to probe the state. 
For a more detailed but still succint overview of entanglement theory of pure bipartite states we refer the reader to \cref{sec:reviewent}, while we point to Ref.\ \cite{Horodecki_2009} for an extensive review of general entanglement theory.

Before diving into entanglement manipulation tasks, let us recall some basic facts about entanglement and its quantification. Given a (possibly mixed) bipartite state $\rho_{A,B}$ we say that $\rho_{A,B}$ is \textit{separable} if it can be expressed as a convex combination of tensor product states on $A$ and $B$, that is, 
\begin{equation}
\rho_{A,B} = \sum_i p_i \rho_{A;i}\otimes \rho_{B;i} 
\end{equation}
with $\sum_i p_i = 1, p_i \geq 0$. We say $\rho_{A,B}$ is \textit{entangled} if such a decomposition does not exist. For pure states this reduces to $\psi_{A,B}$ being different than a tensor product state. 

An important aspect of entanglement theory involves its quantification in terms of entanglement monotones under LOCC. An LOCC \textit{entanglement monotone} is a function of a bipartite state $f(\rho_{A,B})$ that does not increase under LOCC. If the function also does not increase on average under LOCC, we say that $f(\cdot)$ is a \textit{strong monotone}, namely if, for any LOCC $\Lambda$ implementing the mapping $\rho \xrightarrow[]{\Lambda} \{p_i,\rho_i\}$ we have
\be
f(\rho) \geq \sum_i p_i f(\rho_i).
\ee
The theory of entanglement monotones is particularly simple for bipartite pure states: as Vidal has shown  \cite{Vidal_2000} that the set of all pure-states strong entanglement monotones coincides with all local-unitary-invariant concave functions of the reduced density matrix. In particular, this implies that the von Neumann entropy $S_1(\psi_A) = -\tr(\psi_A\log\psi_A) $ is a strong monotone, as well as any $\alpha$-R\'enyi entropy $S_\alpha(\psi_A) = (1-\alpha)^{-1} \log\tr\psi_A^{\alpha}$ for $\alpha < 1$. On the contrary, $\alpha$-R\'enyi entropies for $\alpha>1$ are only (weak) monotones. The von Neumann entropy plays a special role for pure quantum states, since it is the entanglement monotone satisfying the stronger form of continuity and it is essentially the unique measure of pure-state entanglement in the asymptotic sense \cite{Bennett_1996}. More precisely, the von Neumann entropy fully governs the asymptotic conversion rates between bipartite pure states, and, as a consequence, the optimal asymptotic rates of dilution and distillation. However, we will see that when Alice and Bob have limited computing power, the von Neumann entropy fails to capture the accessible entanglement, and another entanglement monotone takes a fundamental role instead, that is the min-entropy $S_{\min}(\psi_A) = -\log\|\psi_A\|$.

We are now ready to discuss the tasks of entanglement distillation and dilution of pure quantum states. The purpose of entanglement distillation is to turn some (entangled) bipartite state $\ket{\psi_{A,B}}$ shared between Alice and Bob into a canonical entangled state, which consists of the tensor product of ebits $\ket{\phi^+}_{A,B}^{\otimes m}$ with $\ket{\phi^+}_{A,B} = \frac{1}{\sqrt{2}}(\ket{0,0}_{A,B}+\ket{1,1}_{A,B})$. As announced before, we are interested in manipulating the state acting only with a restricted set of quantum operations, which in the usual and experimentally motivated setting, are chosen to be LOCC. The most basic setting for entanglement distillation is the \textit{$k$-shot} case, in which Alice and Bob have access to $k$ copies of a state $\ket{\psi_{A,B}}$, and they wish to obtain as many approximate ebits per copy of the input as possible with high probability of success using only LOCC. Let us start defining formally what constitutes an entanglement distillation protocol.

\begin{definition}[Entanglement distillation protocol]\label{def:distillationprotocol} Let $k\in\mathbb{N}^+$, $R\in\mathbb{R}^{\geq0}$, $\varepsilon \in [0,1)$ and $p \in (0,1]$. Starting from $k$ copies of a bipartite pure state $\psi_{A,B}$ on $n$ qubits, a valid $(k,R,\varepsilon,p)$ entanglement distillation protocol consists of an LOCC $\Gamma$ that outputs $kR$ ebits with precision $\varepsilon$ and success probability $p$. In formulae, the output of $\Gamma$ takes the form
\be
\Gamma(\psi_{A,B}^{\otimes k})=p\ketbra{0}_{X}\otimes\omega_{A,B}+(1-p)\ketbra{1}_{X}\otimes\sigma_{A,B},
\ee
with $\frac{1}{2}\|\omega_{A,B}-\phi^{+\otimes kR}_{A,B}\|_1\le \varepsilon$, and the flag (classical) bits $\ket{0}$ (respectively, $\ket{1}$), which are accessible to both parties, indicate that the protocol has succeeded (respectively, failed).
\end{definition}

The definition of a protocol achieving a specific rate allows us to define the measure of $k$-shot distillable entanglement as follows.

\begin{definition}[$k$-shot distillable entanglement]\label{def:kshotdist} Let $\psi_{A,B}$ be a bipartite pure state on $n$ qubits. For a given $k\in\mathbb{N}^+$, and any success probability $1 \ge p\ge \sqrt{64/65}$, and error
$0 \leq \varepsilon \leq 10^{-4}$, the $k$-shot distillable entanglement, denoted as $E_{D}^{(k)}(\psi_{A,B})$ is the maximum achievable rate $R$ among all $(k,R,\varepsilon,p)$ valid entanglement distillation protocols.
\end{definition}

\begin{remark}
    In what follows we will always consider protocols achieving a certain success probability and error in trace distance below a given constant threshold, thus allowing for imperfect protocols. The constants that we consider are necessary to prove precise statements regarding optimal rates in the computationally restricted setting, but we believe they are due to proof artifacts and can possibly be improved.
\end{remark}

In the limit of infinitely many copies we recover the usual, asymptotic definition of distillable entanglement (with finite success probability and error).

\begin{definition}[Asymptotic distillable entanglement] Let $\psi_{A,B}$ be a bipartite pure state on $n$ qubits. The asymptotic distillable entanglement is defined as the following limit
\be
E_D(\psi_{A,B}) =\lim_{k\to\infty} E_{D}^{(k)}(\psi_{A,B}).
\ee
\end{definition}

While the distillable entanglement captures the conversion rate of mixed entanglement into ebits, it is also customary to ask the converse question, that is to consider the optimal (minimal) rates of turning a source of ebits into a specified bipartite target state using only LOCC. Let us again start by defining the corresponding dilution protocol achieving a certain rate.

\begin{definition}[Entanglement dilution protocol]\label{def:costprotocol} Let $k,l\in\mathbb{N}^+$, $R\in\mathbb{R}^{\geq0}$, $\varepsilon \in [0,1)$ and $p \in (0,1]$. Starting from $kR$ ebits $\phi^{+\otimes kR}_{A , B}$ and local sample access to $l$ copies of a target bipartite  pure state $\psi_{A,B}$ on $n$ qubits, a valid $(k,l,R,\varepsilon,p)$ entanglement dilution protocol consists of an LOCC $\Gamma$ that outputs $\psi_{A,B}^{\otimes k}$ with precision $\varepsilon$ and success probability $p$. In formulae, the output of $\Gamma$ takes the form
\be
\Gamma(\psi_{A,A'}^{\otimes l} \otimes \phi^{+\otimes kR}_{A , B})=p\ketbra{0}_{X}\otimes\omega_{A,B}+(1-p)\ketbra{1}_{X}\otimes\sigma_{A,B},
\ee
with $\frac{1}{2}\|\omega_{A,B}-\psi_{A,B}^{\otimes k}\|_1\le\varepsilon$, $\psi_{A,A'} \coloneqq \mathcal{I}^{B\to A'}(\psi_{A,B})$, and the flag (classical) bits $\ket{0}$ (respectively, $\ket{1}$), which are accessible to both parties, indicate that the protocol has succeeded (respectively, failed).
\end{definition}

\begin{remark} In the above definition of the dilution protocol we explicitly inserted local sample access to the state to be diluted. While this modification with respect to the usual setting (where there is no local sample access) \cite{Bennett_1996} plays no role when the LOCC $\Gamma$ can have arbitrarily high computational power, it makes a crucial difference when, as we will consider in the following, the LOCC $\Gamma$ must also be computational-efficient. Indeed, since the computational model of LOCC supplied with ebits is computationally equivalent to any other universal circuit model, if the state is hard to prepare even locally, and no sample access is available, clearly no efficient LOCC exists that dilutes the state non-locally. In our definition, we want to capture exclusively the computational hardness of distributing entanglement, and not that of creating the state locally to begin with. For this reason, we assume that the state has already been prepared locally, and it must only be distributed to the other party via LOCC. This setting differs from the one adopted in 
Ref.\ \cite{arnonfriedman2023computationalentanglementtheory}, where the circuit complexity of local state preparation plays a fundamental role instead.

Moreover, if Alice and Bob have an efficient classical description of the state to be diluted, they can always prepare it locally, thus falling under the same access model considered in \cref{def:costprotocol}.  

\end{remark}

The previous definition allows us to define the $k$-shot and asymptotic entanglement cost similarly as for the distillable entanglement.

\begin{definition}[$k$-shot entanglement cost]\label{def:kshotcost}
Let $\psi_{A,B}$ be a bipartite pure state on $n$ qubits. For a given $k\in\mathbb{N}^+$, and any success probability $1 \ge p\ge \sqrt{64/65}$, and error
$0 \leq \varepsilon \leq 10^{-4}$, the $k$-shot entanglement cost, denoted as $E_{C}^{(k)}(\psi_{A,B})$ is the minimum achievable rate $R$ among all $(k,l,R,\varepsilon,p)$ entanglement dilution protocols. 
\end{definition}

\begin{definition}[Asymptotic entanglement cost] Let $\psi_{A,B}$ be a bipartite pure state on $n$ qubits. The asymptotic entanglement cost is defined as the following limit
\be
E_C(\psi_{A,B}) =\lim_{k\to\infty} E_{C}^{(k)}(\psi_{A,B}).
\ee
\end{definition}

\begin{remark}
In the literature, both the asymptotic distillable entanglement and the entanglement cost are usually defined for vanishingly small error $\varepsilon$ in the limit $k\to\infty$. Indeed, for general mixed states, the definition adopted here, with finite $\varepsilon\in(0,1)$ (known as strong converse rate), usually differs from the former \cite{APE03}. However, for pure states, this distinction is irrelevant when $k\to\infty$, as both rates are equal to the von Neumann entropy of the state \cite{gour2024resourcesquantumworld}.
\end{remark}

While the previous measures capture what is possible to obtain information-theoretically, we now introduce the corresponding quantities in the case where the two parties, Alice and Bob, can only rely on limited computational resources to actually implement the desired protocol. This restricts both the number of shots and the circuit complexity of the LOCC. 

Similarly to the approach in Ref.~\cite{arnonfriedman2023computationalentanglementtheory}, computational entanglement measures can be properly defined only for families of states indexed by the number of qubits \(n\), as the concept of LOCC efficiency is meaningful only in relation to the scaling of $n$. Below, we first introduce efficient LOCC protocols and then use this definition to define the computational distillable entanglement properly (informally defined in \cref{def:dist}) via its lower and upper bounds, and analogously the computational entanglement cost (informally defined in \cref{def:cost}).

\begin{definition}[Efficient LOCC protocols] Let $\mathcal{S}=\{\psi_{A,B;n}\}_n$ be a family of bipartite pure quantum states (indexed by $n$) and $k(n), l(n)$ be functions of $n$. We say that $\mathcal{G}=\{(k_n, l_n,R_n,\varepsilon_n,p_n)\}_n$ is a family of entanglement LOCC protocols on the family of states $\mathcal{S}$ if there exists $n_0\in\mathbb{N}^{+}$ such that for $n\ge n_0$, for each $\psi_{A,B;n}\in\mathcal{S}$ there exists a LOCC $\Gamma_n\in\mathcal{G}$ that, with probability $p_n$, distills $\varepsilon_n$-approximate ebits from $k_n\le k(n)$ copies of $\psi_{A,B;n}$ (resp. dilutes $k_n\ge k(n)$ $\varepsilon_n$-approximate state copies with local sample-access to $l_n \leq l(n)$ copies) at a rate $R_n$. The family $\mathcal{G}$ is an efficient family of LOCC protocols if $k(n)=O(\poly n)$ (additionally, $l(n) = O(\poly n)$) and the LOCC $\Gamma_n$ runs in time $T=O(\poly n)$ for any $\Gamma_n\in\mathcal{G}$. 
\end{definition}

\begin{definition}[Computational distillable entanglement]\label{def:distillableentanglement} Let $\mathcal{S}=\{\psi_{A,B;n}\}_n$ be a family of bipartite pure states on $n$ qubits and $k(n), l(n)$ be functions of $n$. We define the computational distillable entanglement as any function of $n$ obeying the following upper and lower bounds. Given a family $\{(k_n, l_n,R_n,\varepsilon_n,p_n)\}_n$ of efficient  LOCC entanglement distillation protocols acting on the family of states $\mathcal{S}$, with $1\ge p_n\ge \sqrt{64/65}$ and $0\le\varepsilon_n\le10^{-4}$, the computational distillable entanglement satisfies $\hat{E}_{D}^{(k)}(\psi_{A,B;n})\ge R_n$. On the other hand, $\hat{E}_{D}^{(k)}(\psi_{A,B;n})\le \tilde{R}_n$ if there is no efficient family $\{(k_n, l_n,\tilde{R}_n,\varepsilon_n,p_n)\}_n$ of LOCC entanglement distillation protocols with $1\ge p_n\ge \sqrt{64/65}$ and $0\le\varepsilon_n\le10^{-4}$. 
\end{definition}

When it will not lead to inconsistencies, we will drop the state dependence from the notation of $E_{D}^{(k)}(\psi_{A,B})$ ($\hat{E}_{D}^{(k)}(\psi_{A,B})$) and write just $E_{D}^{(k)}$ ($\hat{E}_{D}^{(k)}$). The same applies to other definitions, like the entanglement cost (see \cref{def:entanglementcost}). Notice also that we have dropped the dependency on $n$, as it the rest of the paper we will always deal with sequences of states parametrized by the number of qubits $n$.

Finally, we now introduce the corresponding measure in the case where the two agents, Alice and Bob, have limited computational power to implement the desired protocol.

\begin{definition}[Computational entanglement cost]\label{def:entanglementcost} Let $\mathcal{S}=\{\psi_{A,B;n}\}_n$ be a family of bipartite pure states on $n$ qubits and $k(n)$ be a function of $n$. We define the computational entanglement cost as any function of $n$ obeying the following upper and lower bounds. Given a efficient family $\{(k_n,l_n,R_n,\varepsilon_n,p_n)\}_n$ of LOCC entanglement dilution protocols acting on the family $\mathcal{S}$, with $1\ge p_n\ge \sqrt{64/65}$ and $0\le\varepsilon_n\le10^{-4}$, the computational entanglement cost satisfies $\hat{E}_{C}^{(k)}(\psi_{A,B;n})\le R_n$. On the other hand, $\hat{E}_{C}^{(k)}(\psi_{A,B;n})\ge \tilde{R}_n$ if there is no efficient family of $\{(k_n, l_n,\tilde{R}_n,\varepsilon_n,p_n)\}_n$ LOCC entanglement dilution protocols with $1\ge p_n\ge \sqrt{64/65}$ and $0\le\varepsilon_n\le10^{-4}$. 
\end{definition}

It is immediate to verify that the following relations hold among different measures.

\begin{lemma}[Order relations among entanglement measures]\label{lem:orderrelation} For any number of input copies $k = O(\poly n)$, the following order relations hold
\be\label{eq:relationkshot}
\hat{E}_{D}^{(k)}&\le E_{D}^{(k)}, \quad E_{C}^{(k)}\le \hat{E}_{C}^{(k)} .
\ee
\begin{proof} The proof follows directly from the definition. We only prove the first inequality.  Let \(\mathcal{S} = \{\psi_{A,B;n}\}\) be a family of bipartite pure states indexed by \(n\), and let \(k(n)\) be a function of \(n\). By the definition of \(k\)-shot distillable entanglement in \cref{def:kshotdist}, for every \(n\), there exists an optimal \(k(n)\)-shot LOCC distillation protocol achieving the rate \(E_{D}^{(k)}\). Consequently, we can construct a family \(\mathcal{G}\) of LOCC distillation protocols acting on the family \(\mathcal{S}\).  Now, we consider two cases. If no family of efficient LOCC protocols achieves the same rate as \(\mathcal{G}\), i.e., \(E_{D}^{(k)}\), then by \cref{def:distillableentanglement}, we have \(\hat{E}_{D}^{(k)} \leq E_{D}^{(k)}\).  On the other hand, if there exists a family \(\mathcal{G}'\) of efficient LOCC protocols that achieves the rate \(E_{D}^{(k)}\), then we obtain \(\hat{E}_{D}^{(k)} \geq E_{D}^{(k)}\). Since, by definition of \(k\)-shot distillable entanglement, no LOCC protocol can distill more than \(E_{D}^{(k)}\) for each \(n\), it follows that \(\hat{E}_{D}^{(k)} \leq E_{D}^{(k)}\) as well.  

\end{proof}
\end{lemma}

As mentioned before, the von Neumann entropy completely characterizes the optimal rates for asymptotic entanglement distillation and dilution. The first entanglement distillation protocol designed for pure states and shown to reach the von Neumann entropy bound asymptotically was due to Bennett et al. \cite{Bennett_1996}, while the first distortion-free (distilling exact ebits) and state-agnostic protocol reaching the same, optimal, rate is due to Matsumoto and Hayashi \cite{PhysRevA.75.062338}. Notably, in the infinitely many copies limit there is basically no advantage in having a perfect description of the state: that is the best agnostic protocol performs as well as the best state-aware protocol apart from sub-leading terms in the rate \cite{PhysRevA.75.062338}. This picture however dramatically changes when we take only $k = \poly (n)$ input copies.

For what concerns the finite-copy setting, the optimal entanglement distillation and cost rates for mixed states in a $k$-shot scenario where considered in 
Refs.\ \cite{Fang_2019,theurer2023singleshotentanglementmanipulationstates}. The characterization becomes explicit for some restricted families of states. For the one-shot case instead, it is known that the smoothed min-entropy (respectively max-entropy) gives the optimal rates for what concerns entanglement distillation (respectively, dilution) single-shot protocols \cite{Buscemi_2011,Buscemi_2013}. 
For what concerns us, the next Lemma will be fundamental in the rest of the manuscript.

\begin{lemma}[Information-theoretic bounds on $k$-shot rates]\label{lemma:infothbounds}
Given a bipartite pure state $\psi_{A,B}$, any rates of $k$-shot distillation $R_D^{(k,\varepsilon,p)}(\psi_{A,B})$ and dilution $R_C^{(k,\varepsilon,p)}(\psi_{A,B})$ (see \cref{def:distillationprotocol,def:costprotocol}) achieved with success probability $p \in (0,1)$ and trace distance error $\varepsilon \in [0,1)$ with $1-p+\varepsilon < 1/8$ must satisfy the following relations
\be
R_D^{(k,\varepsilon,p)}(\psi_{A,B}) \leq S_1(\psi_A) + O\left(\frac{n_A}{\sqrt{k}}\right),\label{eq:bounddist}
\ee
\be
R_C^{(k,\varepsilon,p)}(\psi_{A,B}) \geq S_1(\psi_A) - O\left(\frac{n_A}{\sqrt{k}}\right)\label{eq:boundcost},
\ee
where the implicit constants only depend on $\varepsilon$ and $p$. Similarly, it holds that
\be
R_D^{(k,\varepsilon,p)}(\psi_{A,B}) &\leq \frac{S_1(\psi_A) }{p-\varepsilon}\label{eq:bounddist2}\,,
\ee
\be
R_C^{(k,\varepsilon,p)}(\psi_{A,B})&\ge pS_{1}(\psi_{A})-n_Ap\varepsilon\label{eq:boundcost22}\,.
\ee
In particular, these bounds apply to $E_D^{(k)}(\psi_{A,B})$ and $E_C^{(k)}(\psi_{A,B})$ (see \cref{def:kshotdist,def:kshotcost}).
\begin{proof}
Let us start proving the upper bound on the $k$-shot distillable entanglement. If we define the $\varepsilon$-Ball around a given state $\rho$ as
\be
\mathcal{B}^{\varepsilon}(\rho) \coloneqq \{ \tau \in \mathcal{D}(\mathcal{H}): F(\rho,\tau) \geq 1 - \varepsilon^2  \},
\ee
then, the smoothed min- and max-entropies of $\rho$ are defined as \cite{tomamichel2013frameworknonasymptoticquantuminformation}
\be
S_{\min}^{\varepsilon}(\rho) \coloneqq \max_{\Tilde{\rho} \in \mathcal{B}^{\varepsilon}(\rho) } S_{\min}(\Tilde{\rho}), \quad\quad S_{\max}^{\varepsilon}(\rho) \coloneqq \min_{\Tilde{\rho} \in \mathcal{B}^{\varepsilon}(\rho) } S_{\max}(\Tilde{\rho}).
\ee
Now, if we consider distillation protocols with infidelity $ 1 - F \leq \varepsilon$ for any $\varepsilon\in[0,1/4)$ we have
\be
R_D^{(k,\varepsilon)}(\psi_{A,B}) &\overset{\text{(i)}}\leq  \frac{1}{k} \left( S_{\min}^{\varepsilon'}(\psi_A^{\otimes k}) - \log(1-2\sqrt{\varepsilon}) \right) \\
&\overset{\text{(ii)}}\leq S_1(\psi_A) + \frac{\delta(\varepsilon'',\gamma)}{\sqrt{k}} + \frac{1}{k}\log \frac{1}{1 - (\varepsilon'+\varepsilon'')^2}  - \frac{\log(1-2\sqrt{\varepsilon})}{k} \\
&\overset{\text{(iii)}}\leq S_1(\psi_A) + O\left(\frac{n_A}{\sqrt{k}}\right).
\ee
Here, in (i) we used Theorem 1 in Ref.\ \cite{Buscemi_2013} and we set $\varepsilon' = 2^{5/4}\varepsilon^{1/8}$, and in (ii) we used Eq. 6.26 in Ref.\ \cite{tomamichel2013frameworknonasymptoticquantuminformation}, picked some $0<\varepsilon''<1-\varepsilon'$ and introduced
\be
\delta(\varepsilon'',\gamma) = 4 \log \gamma \sqrt{g(\varepsilon'')}, \quad g(\varepsilon'') = - \log (1 - \sqrt{1-(\varepsilon'')^2}),
\ee
and
\be
\gamma &= 2^{-1/2 S_{3/2}(\rho_A\|\pi_A)} + 2^{1/2 S_{1/2}(\rho_a\|\pi_A)} + 1,
\ee
where $\pi_A = \id_A/d_A$ and we denoted the $\alpha$-relative R\'enyi entropies as $S_\alpha(\rho\|\sigma) \coloneqq (1-\alpha)^{-1}\log\tr(\rho^{\alpha}\sigma^{1-\alpha})$. Finally, in (iii) we used the trivial bound
\be
\gamma &= 2^{-1/2 (S_{3/2}(\psi_A)-n_A)} + 2^{1/2 ((S_{1/2}(\psi_A)-n_A))} + 1 \\
&\leq 2^{1/2 (n_A - S_{3/2}(\psi_A))} + 2 \\
&\leq 2^{n_A/2} + 2. \label{eq:trivialbound}
\ee
For what concerns the entanglement cost instead, for any protocol with infidelity $ 1 - F \leq \varepsilon$ and any $\varepsilon\in[0,1/4)$ we have
\be
R_C^{(k,\varepsilon)}(\psi_{A,B}) &\overset{\text{(i)}}\geq \frac{1}{k} S_{\max}^{\varepsilon}(\psi_A^{\otimes k}) \\
&\overset{\text{(ii)}}\geq S_1(\psi_A) - \frac{\delta(\varepsilon',\gamma)}{\sqrt{k}} - \frac{1}{k}\log \frac{1}{1 - (\varepsilon+\varepsilon')^2} \\
&\geq S_1(\psi_A) - O\left(\frac{n_A}{\sqrt{k}}\right),
\ee
where in (i) we used Theorem 1 in Ref.\ \cite{Buscemi_2011}, and in (ii) we used the dual version of Eq. 6.26 in Ref.\ \cite{tomamichel2013frameworknonasymptoticquantuminformation} and picked $0<\varepsilon' < 1 - \varepsilon$. Finally in the last line we applied again the bound in \cref{eq:trivialbound}. To translate these bound in our framework, a protocol succeeding with probability $p$ and trace distance error $\varepsilon$ can be seen as a protocol succeeding with unit probability and trace distance error $\tilde{\varepsilon}= 1-p+\varepsilon$, imposing that $F > 3/4$ we get $(1-\Tilde{\varepsilon})^2 > 3/4$, and we can pick any $\Tilde{\varepsilon} < 1/8$. This concludes the proof. 

Let us note, for completeness, that, for the distillation rate, a similar, but asymptotically worse bound could have been obtained in a much simpler way as follows:
\be
kR_D^{(k,\tilde{\varepsilon})}(\psi_{A,B}) &= S_1((\phi^+_{A})^{\otimes k R}) \\
&\leq S_1(\Lambda(\psi_{A}^{\otimes k})) + kR_D^{(k,\varepsilon)}(\psi_{A,B})\tilde{\varepsilon} \\
&\leq k S_{1}(\psi_A) + kR_D^{(k,\varepsilon)}(\psi_{A,B})\tilde{\varepsilon}.
\ee
Where we used Fannes inequality \cite{Wilde_2013} and the monotonicity of the von Neumann entropy under the LOCC $\Lambda$. This gives
\be
R_D^{(k,\tilde{\varepsilon})}(\psi_{A,B}) \leq \frac{S_{1}(\psi_A)}{1-\tilde{\varepsilon}}.
\ee
thus proving \cref{eq:bounddist2}. For the cost we have
\be
kR_{C}^{(k,{\varepsilon},p)}(\psi_{A,B})\ge pS_{1}(\Lambda(\phi_{A}^{+\otimes k}))\ge pkS_{1}(\psi_{A})-kn_Ap{\varepsilon}
\ee
implying that $R_{C}^{(k,{\varepsilon},p)}(\psi_{A,B})\ge pS_{1}(\psi_{A})-n_A{p\varepsilon}$, which is meaningful only if ${\varepsilon}\le \frac{S_{1}(\psi_{A,B})}{n_A}$.
\end{proof}
\end{lemma}

Let us conclude this section with a final remark on computational entanglement measures. In the remainder of this work, we will distinguish further between \textit{state-agnostic} and \textit{state-aware} protocols. We say that a protocol is \textit{state-agnostic} if it does not rely on any prior knowledge of the input state. Conversely, in a computationally limited setting we say that a protocol is \textit{state-aware} if a classical (efficient) circuit description of the state object of the protocol is available to the agents performing the protocol. We will also consider cases in between, where only partial information on the state is available. Let us note that in the asymptotic case these distinctions essentially play no role, as the optimal achievable rates for pure states can be obtained via state-agnostic protocols \cite{PhysRevA.75.062338}.
In the following, we will show several no-go results in these settings. To this regard, it is important to stress that any no-go result on state-agnostic or state-aware protocols does not provide any bound on computational entanglement measures, as it does not rule out the existence of a protocol achieving a better rate. For this reason, one needs \textit{existential} no-go  theorems as the one we will prove in \cref{App:tightdistillable,app:lowerboundcost}.

\subsection{Review of other known results in entanglement theory of pure quantum states}\label{sec:reviewent}
Here, we will briefly expand on what we already discussed in the previous section, to give the reader a slightly bigger overview on entanglement theory of bipartite pure states. The study of entanglement quantification and manipulation for pure states dates back to seminal work by Nielsen \cite{Nielsen_1999} and Vidal \cite{Vidal_2000}. In Ref.\  \cite{Vidal_2000}, a complete characterization of all possible pure-state entanglement measures satisfying strong monotonicity was given in terms of concave and unitarily invariant functions of the reduced density matrix. In Ref.\ \cite{Nielsen_1999}, it was shown that necessary and sufficient conditions for exact LOCC transformations between two pure states $\psi \to \phi$ with success probability one are given by the majorization condition on the Schmidt coefficients of the states $\lambda(\psi)\prec \lambda(\phi)$. This condition implies that in any exact (with probability of success one) entanglement distillation protocol, the maximum achievable rate is given by $S_{\min}(\rho_A)$. Lo and Popescu \cite{PhysRevA.63.022301} later showed that the symmetry imposed by the Schmidt decomposition implies that any LOCC protocol for pure states can always be recasted in a local (generalized) measurement performed by one party, Alice, who then communicates the outcome to Bob at the end of the procedure. This in particular implies that only one-way classical communication is needed when dealing with pure states, which greatly simplifies the picture compared to general LOCC. Subsequently, Nielsen condition was generalized by Vidal to probabilistic, but exact, LOCC transformations to $\lambda(\psi)\prec^{\omega} p \lambda(\phi)$ \cite{PhysRevLett.83.1046}, while the most general LOCC transformation mapping $\psi \to \{p_i, \phi_i\}$ is characterized by the condition \cite{PhysRevLett.83.1455} $\lambda(\psi)\prec \sum_i p_i \lambda(\phi_i)$. All these statements are constructive, in that they show which LOCC one has to implement to perform the desired transformation. Let us note that while the theory of entanglement manipulation is well-known, these protocols assume full knowledge of the state and are usually computationally inefficient \cite{Nielsen_1999}. For these reasons, any $k$-shot protocol obtained by these sufficient conditions, even when $k=O(\poly n)$, can possibly reach the von Neumann entropy rate \cite{theurer2023singleshotentanglementmanipulationstates}, which, as we show in the following, cannot be obtained in a computationally restricted setting.

\subsection{Schur-Weyl duality and the Schur transform}
Schur-Weyl duality is a prominent toolbox in quantum information theory, with applications that include random quantum circuits \cite{Brand_o_2016}, quantum state tomography \cite{Haah_2017,odonnell2015efficientquantumtomography}, spectrum estimation \cite{PhysRevA.64.052311}, entanglement distillation \cite{PhysRevA.75.062338} and state compression \cite{PhysRevA.66.022311}. In this section, we review the fundamentals of Schur-Weyl duality. These will be useful in the next section, where these techniques will be applied to design a state-agnostic entanglement distillation protocol.

Before diving into Schur-Weyl duality let us recap some basic facts about representation theory.
Given a group $G$, a representation of $G$ on a Hilbert space $\mathcal{H}$ consists of a set of unitary operators $R_g \in \mathcal{B(\mathcal{H})}$ such that $R_{gh}=R_g R_h$ for arbitrary $g,h\in G$. 
An important notion is that of \textit{invariant subspace}. These are the subspaces of $\mathcal{H}$ that are left invariant by the representation $R_g$, i.e.,  such that $R_g$ has a block-diagonal form where each block is associated to one of the invariant subspaces. We call such a representation irreducible if it cannot be decomposed further into non-trivial (i.e.,  different from zero or themselves) blocks. A fundamental task in representation theory is to identify the irreducible representation of a group $G$, up to isomorphisms. 

It can be shown that any compact group admits a direct-sum decomposition into irreducible representations, or \textit{irreps} \cite{weyl1946classical}. Another very important notion is that of an \textit{intertwiner}. An intertwiner between two representations $\mathcal{H}$ and $\mathcal{H'}$ is a linear map $K \in \mathcal{B(\mathcal{H}\to\mathcal{H}')}$ that satisfies $KR_g = R_g'K$. The following basic result, known as Schur's Lemma, then holds (see for example \cite{weyl1946classical}).
\begin{lemma}\label{lemma:schur}(Schur~\cite{fulton2004rep}) Let $\mathcal{H}$ and $\mathcal{H'}$ be two irreducible representations of a group $G$, and let $K:\mathcal{H}\to\mathcal{H'}$ be an intertwiner between them. Then the following holds
\begin{itemize}
    \item If the irreducible representations are not isomorphic, then $K=0$.
    \item If the irreducible representations are isomorphic, and $R_g=R_g'$ then $K\propto I_{\mathcal{H}}$.
\end{itemize}
\end{lemma}
We will make direct use of this in \cref{lemma:state}.
Consider now the \textit{symmetric group} of permutations of $k$ elements, denoted as $S_k$, and the $d$-dimensional \textit{unitary group}, denoted as $\mathcal{U}(d)$. Schur-Weyl duality characterizes the irreducible representations of the product group $S_k\times\mathcal{U}(d)$ over the tensor-product Hilbert space $(\mathbb{C}^d)^{\otimes k}$. Specifically, for every permutation 
of $k$ elements $\pi \in S_k$, we can naturally define its representation $R_{\pi}$ acting on the Hilbert space $\mathcal{H} = (\mathbb{C}^d)^{\otimes k}$ as
\be
R_{\pi} \ket{\psi_1}\otimes \dots \otimes\ket{\psi_k} \coloneqq \ket{\psi_{\pi^{-1}(1)}}\otimes \dots \otimes\ket{\psi_{\pi^{-1}(k)}} .
\ee
Analogously, given a unitary $U\in\mathcal{U}(d)$ we define its representation $T_{U}$ on $\mathcal{H}$ as
\be
T_{U} \ket{\psi_1}\otimes \dots \otimes\ket{\psi_k} \coloneqq U^{\otimes k} \ket{\psi_1}\otimes \dots \otimes\ket{\psi_k} .
\ee
It is immediate to verify that both $R_{\pi}$ and $T_{U}$ are valid group representations, in that they preserve the group structure. In particular, they define a group action for the product group $S_k\times\mathcal{U}(d)$ over $\mathcal{H}$. Since $S_k$ is a finite group and $\mathcal{U}(d)$ is a compact group, it is possible to show that an decomposition into a direct sum of irreducible invariant subspaces exists for both $S_k$ and $\mathcal{U}(d)$~\cite{fulton2004rep}. Furthermore, since $[R_{\pi},T_{U}]=0$, it is possible to simultaneously decompose $(\mathbb{C}^d)^{\otimes k}$ into irreducible representations of both $S_k$ and $\mathcal{U}(d)$ as
\be
(\mathbb{C}^d)^{\otimes k} \cong \bigoplus_{\lambda \in \mathcal{I}_{(k,d)}} \mathcal{U}_{\lambda}^{(d)}\otimes \mathcal{V_{\lambda}}\label{schur-weyl},
\ee
where $\mathcal{U_\lambda}$ is an irreducible subspace for the group $\mathcal{U}(d)$, while $\mathcal{V_\lambda}$ is an irreducible subspace for the group $S_k$. This result is what is known as Schur-Weyl duality. The set $\mathcal{I}_{(k,d)}$ spans the so-called Young coefficients, that is to say all the partitions of $d$ positive integers summing to $k$
as
\begin{equation}
    {\mathcal{I}_{(k,d)} \coloneqq \left\{ \lambda \in \mathbb{Z}^d: \sum_{i=1}^d \lambda_i = k, \quad \lambda_i \geq \lambda_{i+1} \geq 0  \right\} }.
    \end{equation}In the remainder of this work, we will only deal with the case $k < d$, in this case the partition will involve at most the first $k$ integers, the rest being zero, and we will denote it by $\mathcal{I}_{(k,k)}$.

The decomposition in \cref{schur-weyl} means that under a proper choice of basis that we will denote as $\ket{\lambda,u_i} \otimes \ket{\lambda, v_j}$ the respective group actions of $S_k$ and $\mathcal{U}(d)$ will act in the block diagonal form
\be
{R_{\pi} \cong \bigoplus_{\lambda \in \mathcal{I}_{(k,d)}} I_{\dim \mathcal{U}_{\lambda}^{(d)}} \otimes R_{\pi}^{(\lambda)}  }\label{block-pi},
\ee
\be
{T_{U} \cong  \bigoplus_{\lambda \in \mathcal{I}_{(k,d)}} T_{U}^{(\lambda)} \otimes I_{\dim \mathcal{V}_\lambda} }\label{block-u}.
\ee
We will call this \textit{Schur basis}, and the unitary transformation sending the standard basis to the Schur one will be called \textit{Schur transform},
\be
U_{\text{Sch}}: \{\ket{k}\otimes\ket{l}\} \rightarrow \{ \ket{\lambda,u_i} \otimes \ket{\lambda, v_j} \}.
\ee
Let us note that the different subspaces $\mathcal{U}_{\lambda}^{(d)},\mathcal{V}_{\lambda}$ will have a dimension depending on $\lambda$. In the following we will pad some registers with zeroes so that the dimensions match for every $\lambda$ and we can assume the indices $i,j$ run over the same set for all $\lambda$'s.
The projector onto the $\lambda$ irrep can instead be written as \cite{weyl1946classical}
\be
\Pi_{\lambda} = \frac{\dim\mathcal{V}_\lambda}{k!}\sum_\pi \chi^{\lambda}(\pi) R_{\pi}\label{proj},
\ee
where $\chi^{\lambda}(\pi) \coloneqq \tr R_{\pi}^{(\lambda)}$ are the characters of the symmetric group $S_k$. The dimensions of the irreducible representations of $S_k$ and $\mathcal{U}(d)$ can be readily expressed as \cite{weyl1946classical}
\be
\dim \mathcal{U}_{\lambda}^{(d)} = \frac{\prod_{1\leq i<j\leq d}(\lambda_i - \lambda_j - i + j)}{\prod_{m=1}^d m!},
\ee
\be
\dim \mathcal{V}_{\lambda} &= \frac{k!\prod_{1\leq i<j\leq d}(\lambda_i - \lambda_j - i + j)}{(\lambda_1 + d-1)!(\lambda_1 + d-2)!...\lambda_d !}.
\ee
Importantly, while $\dim \mathcal{U}_{\lambda}^{(d)}$ has an explicit dependence on the dimension of the Hilbert space $d$, $\dim \mathcal{V}_{\lambda}$ depends on $d$ only implicitly through $\lambda$. The quantum measurement that describes the projection onto the $\lambda$ irreducible representations is known as \textit{weak Schur sampling}~\cite{Childs}. The following result will be fundamental in our work. This states that if the number of copies is $ k = \poly n$, then the Schur transform and weak Schur sampling can be implemented efficiently (in time $\poly n$ in the number of qubits if $d=2^n$) on a quantum computer:
\begin{theorem}[\cite{Krovi2019efficienthigh,Childs}] The Schur transform, and therefore also weak Schur sampling, can be implemented on a quantum computer with precision $\delta$ in running time $O(\poly(k, n, \log(1/\delta)))$.
\end{theorem}

\subsection{Entanglement distillation protocols based on the Schur transform}\label{section:protocol}
In this section we will describe an entanglement distillation protocol originally proposed by Matsumoto and Hayashi \cite{PhysRevA.75.062338}. The analysis of this protocol in the previously unexplored, sample-efficient regime $k = \poly n$, will be the focus of part of our results. Differently from the protocol proposed by Bennett et al. \cite{Bennett_1996}, the protocol of Ref.\ \cite{PhysRevA.75.062338} achieves the entropy rate asymptotically in the number of copies without relying on any prior knowledge of the input state. This state-agnostic property will be of crucial importance for our purposes. Prior to describing the algorithm, it will be useful to consider the following lemma due to \cite{PhysRevA.75.062338}, for which we present a self-contained proof for the sake of clarity.

\begin{lemma}[\cite{PhysRevA.75.062338}]\label{lemma:state} The $k$-fold tensor product of a bipartite pure state $\ket{\psi_{A,B}}$ can be expressed as a direct sum of states belonging to different (orthogonal) irreducible representations $\lambda$ of $S_k\times\mathcal{U}(d)$  as
\be
\ket{\psi_{A,B}}^{\otimes k} = \sum_{\lambda} \sqrt{\Pr(\lambda)} \ket{\Phi_{A,B}^\lambda} \otimes (U_{\mathcal{V}_{\lambda}}^A \otimes U_{\mathcal{V}_{\lambda}}^B) \ket*{\phi^+_{A,B}}^{\otimes \log \dim\mathcal{V}_\lambda}\label{state-dec},
\ee
where $\Pr(\lambda) = \tr\Pi_{\lambda}^A\rho_A^{\otimes k}$, $\ket*{\Phi_{A,B}^\lambda} \in \mathcal{U}_\lambda^A \otimes \mathcal{U}_\lambda^B$, $\ket*{\phi^+_{A,B}}^{\otimes \log \dim\mathcal{V}_\lambda} \in \mathcal{V}_\lambda^A \otimes \mathcal{V}_\lambda^B$ is the tensor product of $\log \dim\mathcal{V}_\lambda$ ebits, and $U_{\mathcal{V}_{\lambda}}^{(A)}$ ($U_{\mathcal{V}_{\lambda}}^{(B)}$) is the Schur transform restricted on the subspace $\mathcal{V}_{\lambda}$ on system A (B).
\end{lemma}
\begin{proof} The proof is a direct consequence of Schur-Weyl duality (\cref{schur-weyl}) and of the fact that the $k$-fold tensor product state is invariant under the permutation $R_{\pi}^A\otimes R_{\pi}^B$. 

In particular, let us express the $k$-copy state $\ket{\psi_{A,B}}^{\otimes k}$ in terms of the basis in which the actions of $T_U$ and $R_{\pi}$ are both block-diagonal (see \cref{block-pi,block-u}). For simplicity, let us denote this basis as $\ket{\lambda,u_i}\otimes\ket{\lambda, v_j}$. Upon padding with zeros, we can assume the indices $i,j$ to run over a set independent from $\lambda$. The state can then be written as
\be
\ket{\psi_{A,B}}^{\otimes k} = \sum_{\lambda,\lambda',i,j,k,l} c_{\lambda,\lambda',i,j,k,l} \ket{\lambda,u_i}_A \otimes \ket{\lambda, v_j}_A \otimes \ket{\lambda',u_k}_B \otimes \ket{\lambda', v_l}_B.\label{state1}
\ee
Now, this state is in one-to-one correspondence with the mixed state on system $A^k$ obtained from \cref{state1} by sending $\ket{\lambda,u_i}_A \otimes  \ket{\lambda',u_k}_B \to \ketbra{\lambda,u_i}{\lambda',u_k}_A$ and $\ket{\lambda,v_j}_A \otimes  \ket{\lambda',v_l}_B \to \ketbra{\lambda,v_j}{\lambda',v_l}_A$. Let us denote this operator as acting on $A^k$ as $\Psi^k$
\be
\Psi^k \coloneqq \sum_{\lambda,\lambda',i,j,k,l} c_{\lambda,\lambda',i,j,k,l} \ketbra{\lambda,u_i}{\lambda',u_k}_A \otimes \ketbra{\lambda, v_j}{\lambda', v_l}_A. \label{state2}
\ee
Now, since the overall state is left invariant by any permutation we also have the corresponding invariance for $\Psi^k$
\be
R_{\pi}^A\otimes R_{\pi}^B \ket{\psi}_{A,B}^{\otimes k} = \ket{\psi}_{A,B}^{\otimes k} \implies R_{\pi} \Psi^k \left(R_{\pi}\right)^{\dagger} = \Psi^k .\label{state3}
\ee
Let us now apply the projectors onto any two irreducible representations of $S_k$, say $\Pi_{\mathcal{V}_\lambda}$ and $\Pi_{\mathcal{V}_{\lambda'}}$, on the left and right hand side of \cref{state3} respectively. Since $\mathcal{V}_\lambda$ and $\mathcal{V}_{\lambda'}$ are invariant subspaces with respect to the action of $R_{\pi}$, we have $[R_{\pi},\Pi_{\mathcal{V}_\lambda}]=0$ and we get 
\be
R_\pi^\lambda \Pi_{\mathcal{V}_\lambda} \Psi^k \Pi_{\mathcal{V}_{\lambda'}} (R_\pi^{\lambda'})^\dagger = \Pi_{\mathcal{V}_\lambda} \Psi^k \Pi_{\mathcal{V}_{\lambda'}}.
\ee
This shows that the operator $\Pi_{\mathcal{V}_\lambda} \Psi^k \Pi_{\mathcal{V}_{\lambda'}}$ is an intertwiner between the two (different) irreducible representations $R_{\pi}^{\lambda}$ and $R_{\pi}^{\lambda'}$. Now, only two things can happen: either the two irreducible representations are inequivalent, so $\lambda\neq\lambda'$, and therefore by \cref{lemma:schur} we must have $\Pi_{\mathcal{V}_\lambda} \Psi^k \Pi_{\mathcal{V}_{\lambda'}} = 0$, or the two irreducible representations are isomorphic, that is $\lambda=\lambda'$ and by \cref{lemma:schur} we must have $\Pi_{\mathcal{V}_\lambda} \Psi^k \Pi_{\mathcal{V}_{\lambda}} \propto I_{\mathcal{V}_\lambda}$. In particular this implies that $c_{\lambda,\lambda',i,j,k,l} = 0$ unless $\lambda = \lambda'$ and that $c_{\lambda,\lambda,i,j,k,l} = \Tilde{c}_{\lambda,i,k}  \delta_{jl}$. This finally allows us to write the state in the following fashion
\be
\Psi^k = \sum_{\lambda,i,k} \Tilde{c}_{\lambda,i,k} \ketbra{\lambda,u_i}{\lambda,u_k}_A \otimes \sqrt{\frac{1}{\dim\mathcal{V}_{\lambda}}} \id_{\dim\mathcal{V}_{\lambda}}^A, \label{state4}
\ee
for some complex coefficients $\Tilde{c}_{\lambda,i,k}$. Finally, applying the inverse of the transformation that we applied in \cref{state2} we go back to $\ket{\psi_{A,B}}^{\otimes k}$ and get
\be
\ket{\psi_{A,B}}^{\otimes k} = \sum_{\lambda} \sqrt{a_\lambda} \ket{\Phi_{A,B}^\lambda} \otimes (U_{\mathcal{V}_{\lambda}}^A\otimes U_{\mathcal{V}_{\lambda}}^B) \ket*{\phi^+_{A,B}}^{\otimes \log \dim\mathcal{V}_\lambda},
\ee
for some $a_\lambda\geq0$ and normalized $\braket{\Phi^\lambda}_{A,B}=1$. Projecting onto $\Pi^\lambda$ and tracing out system $B^n$ we also find
\be
\Pr(\lambda) = \tr\Pi_{\lambda}^A\rho_A^{\otimes k} = a_{\lambda},
\ee
so the overall claim is proven.
\end{proof}

After having discussed \cref{lemma:state} we are ready to describe the distillation protocol. The procedure works as follows. Alice and Bob start having access to $k$ copies of $\ket{\psi_{A,B}}$. As shown in \cref{lemma:state}, the resulting tensor product state can be decomposed into a direct sum over the $\lambda$ irreps. Alice and Bob then independently apply the projective measurement $\{\Pi_{\lambda}^{A}\otimes \Pi_{\lambda}^{B}\}_{\lambda}$ onto their respective (local) subspaces. The decomposition in \cref{lemma:state} (valid only for pure states) guarantees that they sample the same value of $\lambda$ with probability $\Pr(\lambda)$, and they are left with the tensor product state $\ket*{\Phi_{A,B}^\lambda} \otimes (U_{\mathcal{V}_{\lambda}}^{A}\otimes U_{\mathcal{V}_{\lambda}}^{B}) \ket*{\phi^+_{A,B}}^{\otimes \log\dim\mathcal{V}_\lambda}$. Alice and Bob then discard the state $\ket{\Phi_{A,B}^\lambda}$, that is the register associated to the $\mathcal{U}_{\lambda}$ basis and they are left with a set of ebits in a rotated basis $(U_{\mathcal{V}_{\lambda}}^{A}\otimes U_{\mathcal{V}_{\lambda}}^{B}) \ket*{\phi^+_{A,B}}^{\otimes \log\dim\mathcal{V}_\lambda}$. They then locally apply the inverse Schur transform to go back to the computational basis. At the end of this procedure they are left with perfect shared ebits in the computational basis. 

Now, as we will discuss in the next sections, due to \cref{th:distillation}, with probability greater than $2/3$ the number of ebits per copy will be at least 
$\min\{\frac{1}{4}S_{\min}(\rho_A), \frac{1}{4}\log_2 k  \}$ even when $k = O(\poly n)$. The success probability can eventually be boosted by repeating the protocol.
The protocol is summarized in \cref{algo1}. As a remark, if Alice and Bob choose a larger value of $k$ than what is needed by \cref{th:distillation}, they are still guaranteed to reach at least the same rate, as it should be intuitively the case, and as is it formalized by the following result of 
Ref.\ \cite{PhysRevA.75.062338} for which we show a self-contained proof.

\begin{theorem}[\cite{PhysRevA.75.062338}]\label{thm:optimalprotocol} The distillation protocol described above, hereby denoted as $\Gamma_*^k$, satisfies the following optimality properties on any fixed input state $\psi_{A,B}$,
\begin{enumerate}[label=(\roman*)]
    \item The probability of achieving a given rate $R$, $\Pr_{\Gamma_*^k}(x\geq R| \psi_{A,B}^{\otimes k})$ of $\Gamma_*^k$ for any fixed $k$ is optimal among all unitarily invariant, distortion-free $k$-shot distillation protocols.
    \item The distillation probability $\Pr_{\Gamma_*^{(k+1)}}(x\geq R| \psi_{A,B}^{\otimes (k+1)})$ achieved by the protocol acting on $k+1$ copies is at least as high as that of the protocol acting on $k$ input copies $\Pr_{\Gamma_*^{k}}(x\geq R| \psi_{A,B}^{\otimes k})$. 
\end{enumerate}
\begin{proof} We will first prove (i). Let us denote a generic, unitarily invariant and distortion-free LOCC distillation protocol by $\Gamma^k$. The proof scheme is as follows: first, we will show that the probability of success (i.e.,  achieving a rate larger than $R$) of any $\Gamma^k$ is the same as that of applying $\Gamma_*^k$ followed by another LOCC protocol $\Tilde{\Gamma}^k$. Then, we will show that any choice of $\Tilde{\Gamma}^k$ cannot increase the success probability already achieved by $\Gamma_*^k$. 

Let us start noting that for any unitarily invariant LOCC distillation protocol $\Gamma^k$, the following holds
\be
\Pr_{\Gamma^k}(x\geq R| \psi_{A,B}^{\otimes k}) &= \Pr_{\Gamma^k}(x\geq R| (U_A^{\otimes k}\otimes U_B^{\otimes k}) \psi_{A,B}^{\otimes k}(U_A^{\otimes k}\otimes U_B^{\otimes k})^\dagger) \\
&= \Pr_{\Gamma^k}(x\geq R|\int_{U_A,U_B\in\text{Haar}} (U_A^{\otimes k}\otimes U_B^{\otimes k}) \psi_{A,B}^{\otimes k}(U_A^{\otimes k}\otimes U_B^{\otimes k})^\dagger dU_A dU_B) \\
&= \Pr_{{\Tilde{\Gamma}}^k \circ \Gamma_*^k}(y\geq R| \psi_{A,B}^{\otimes k}) .
\ee 
In this expression, we have used that the probabilities are unitarily invariant when local unitaries are concerned and that they are linear in the input. Then we noted that applying $\Gamma_*^k$ is the same as applying the Haar average and then applying some other LOCC operation $\Gamma^k$. Indeed one has, following a similar reasoning as in the proof of \cref{lemma:state}, that
\be
\int_{U_A,U_B\in\text{Haar}} (U_A^{\otimes k}\otimes U_B^{\otimes k}) \psi_{A,B}^{\otimes k}(U_A^{\otimes k}\otimes U_B^{\otimes k})^\dagger dU_A dU_B = \sum_{\lambda} \Pr_{\Gamma_*^k}(\lambda) \frac{\Pi_{\mathcal{U}_{\lambda}}^A \otimes \Pi_{\mathcal{U}_{\lambda}}^B\otimes\ketbra*{\phi_{A,B}^+}^{\otimes \log \dim\mathcal{V}_\lambda}}{(\dim\mathcal{U}_{\lambda})^2},
\ee
therefore, the output of the Haar average is the same that one gets when applying $\Gamma_*^k$ followed by the local map 
\begin{equation}\ketbra{\lambda} \otimes \ketbra*{\phi_{A,B}^+}^{\otimes \log \dim\mathcal{V}_\lambda} \to \Pi_{\mathcal{U}_{\lambda}}^A \otimes \Pi_{\mathcal{U}_{\lambda}}^B\otimes\ketbra*{\phi_{A,B}^+}^{\otimes \log \dim\mathcal{V}_\lambda}.
\end{equation}
Now, we can discuss the optimal choice of $\Tilde{\Gamma}^k$. Here the distortion-free hypothesis is crucial. Assume that the output of $\Gamma_*^k$ (with some probability) is $\ket*{\phi^+_{A,B}}^{\otimes kx}$, then any LOCC distortion-free protocol $\Gamma^k$ applied to the state must output some other state of the form $\ket*{\phi^+_{A,B}}^{\otimes ky}$ with probability $\Pr_{\Gamma^k}(y|x)$. Now, necessary and sufficient conditions for the existence of LOCC protocols sending a given input pure state into some desired mixture of pure states are completely characterized by Nielsen condition \cite{Nielsen_1999}, which in particular implies that it is not possible to increase the Schmidt rank of pure states even probabilistically 
(see also Ref.\ \cite{PhysRevA.63.022301}). This implies 
that the optimal choice of $\Tilde{\Gamma}^k$ corresponds to leaving the output ebits unchanged. This proves (i).

Let us now prove (ii). By (i), the distillation protocol achieves the optimal probability of success among all distortion-free agnostic protocols for $k+1$ input copies. 
Now, in particular, the same protocol acting on $k$ copies is a specific type of protocol acting on $k+1$ copies that discards one copy. The $k$ copies protocol must therefore achieve a smaller success probability. This concludes the proof.
\end{proof}
\end{theorem}

\cref{thm:optimalprotocol} in particular implies that any unitarily invariant protocol, that therefore can exploit knowledge only of the Schmidt coefficients, but not of the Schmidt basis, cannot perform better (for any number of input copies) than the protocol of Ref.\ \cite{PhysRevA.75.062338}, which uses no knowledge about the state. Of course this optimality is not guaranteed to hold when either the Schmidt basis is known (in this case one can distill $S_{\min}$ even with just one copy of the state, while this protocol distills zero), or if it is unknown but the protocol is not unitarily invariant.

\begin{algorithm}[H]
\caption{Entanglement distillation protocol}\label{algo1}
\begin{algorithmic}[1]
\Require{$\ket{\psi}_{A,B}^{\otimes k}$, $k\in\mathbb{N}^+$}
\Ensure{$\ket*{\phi^+_{A,B}}^{\otimes k R}$, $R\ge\min\frac{1}{4}\{S_{\min},\log k\}$ with probability $\ge \frac{2}{3}$}
\Function{SchurSampling}{$\ket{\psi}_{A,B}^{\otimes k}$}
\State $\ket{\Phi_{A,B}^\lambda} \otimes (U_{\mathcal{V}_{\lambda}}^{A}\otimes U_{\mathcal{V}_{\lambda}}^{B})\ket*{\phi^+_{A,B}}^{\otimes \log\dim\mathcal{V}_\lambda} \gets$ Alice and Bob 
locally apply the projective measurement $\{\Pi_{\lambda}^{A}\otimes \Pi_{\lambda}^{B}\}_{\lambda}$
\State $(U_{\mathcal{V}_{\lambda}}^{A}\otimes U_{\mathcal{V}_{\lambda}}^{B})\ket*{\phi^+_{A,B}}^{\otimes \log\dim\mathcal{V}_\lambda} \gets$ Alice and Bob discard the subsystem $\mathcal{U}_{\lambda}$
\EndFunction
\Function{InverseSchurTransform}{}
\State $\ket*{\phi^+_{A,B}}^{\otimes \log\dim\mathcal{V}_\lambda} \gets$ Alice and Bob apply the inverse Schur transform $(U_{\mathcal{V}_{\lambda}}^{A}\otimes U_{\mathcal{V}_{\lambda}}^{B})^\dagger$
\State \Return{$\ket*{\phi_{A,B}^{+}}^{\otimes kR}$, for $R=\frac{\log\dim\mathcal{V}_{\lambda}}{k}$}
\EndFunction
\end{algorithmic}
\end{algorithm}

\subsection{Learning with errors}\label{app:lwe}
In this section we give a brief overview of the \emph{learning with errors} (LWE) decision problem. The conjectured  computational hardness of LWE constitutes the basis of post-quantum cryptography and allows us to derive our no-go results on computationally efficient entanglement manipulation when a classical description of the state is available to Alice and Bob.
Let us consider the set of $\mathbb{Z}_q^{m\times n}$ matrices. We will consider two basic probability distributions over this set:
the uniform distribution, denoted as $\mathcal{U}_q^{m\times n}$, and the i.i.d.\  Gaussian distribution, denoted by $\mathcal{N}_{q,\sigma}^{m\times n}$ and obtained by sampling each element in the matrix independently from a discretized Gaussian distribution with variance $\sigma^2$. The goal of the LWE decision problem is to distinguish between noisy inner products and uniformly random samples:

\begin{definition}[Learning with errors decision problem] The $\text{LWE}_{n,m,q,\sigma}$ problem is the computational problem of distinguishing between the following two distributions over $\mathbb{Z}_q^{m\times n}\times\mathbb{Z}_q^{m}$
\begin{enumerate}
    \item $(A,u) \sim \mathcal{U}_q^{m\times n}\times \mathcal{U}_q^{m} $,
    \item $(A,As + e)$, where $ A \sim \mathcal{U}_q^{m\times n}$, $s \sim \mathcal{U}_q^{n}$ and $e \sim \mathcal{N}_{q,\sigma}^m$.
\end{enumerate}
\end{definition}
The hardness of LWE is then assumed as follows, following standard post-quantum cryptographic assumptions \cite{TCS-074}.

\begin{assmt}[Hardness of LWE] 
For every constant $\beta >0$, every $\poly(m)$-time quantum algorithm has \text{negl}$(m)$-advantage in solving the $LWE_{n,m,q,\sigma}$ problem for $n = m^{\beta}$, $q = 2^{\poly(n)}$ and 
\begin{equation}
\sigma \geq \frac{q}{\poly(m)}\geq 2 \sqrt{m}.
\end{equation}
\end{assmt}

\subsection{Concentration bounds}

In this section, we derive a useful concentration bound that allows us to restrict the size of state ensembles without incurring a significant error in trace distance of the mixture. For this, we make use of a vector-version of the Bernstein inequality.

\begin{lemma}[Vector Bernstein Inequality, Lemma 18 in \cite{kohler2017sub}, Theorem 6 in \cite{gross2011recovering}]\label{lem:vector-bernstein}
Let $X_1,\ldots,X_N$ be independent vector-valued random variables with common dimension $d$ and assume that each one is centered, uniformly bounded and also the variance is bounded above:
\begin{equation*}
\mathbb{E}[X_i]=0 \text{ and } \norm{X_i}_2 \leq \mu \text{ as well as } \mathbb{E}[\norm{X_i}_2^2]\leq \sigma^2
\end{equation*}
Then we have for $0<\varepsilon' <\sigma^2/\mu$
\begin{equation} \label{eq:final_bernstein_vector}
P\left(\norm{\frac{1}{N}\sum_{i=1}^n X_i}_2\geq \varepsilon'\right) \leq \exp\left(- \frac{N\varepsilon'^2}{8\sigma^2}+\frac{1}{4}\right).
\end{equation}
\end{lemma}
This almost immediately gives us our result, as a corollary:
\begin{lemma}\label{lem:concentration}
Let $\mathcal{E}$ an ensemble of $d$-dimensional states. Suppose we sample $N$ states $\{\rho_i\}_{i=1}^N$ uniformly and independently from $\mathcal{E}$. Then, for any $\varepsilon \in (0,1)$, $k \geq 1$, we have
\begin{equation}
    \Pr(\left\|\frac{1}{N}\sum_{i=1}^{N}\rho_i^{\otimes k}-\mathbb{E}_{\rho\in\mathcal{E}}[\rho^{\otimes k}]\right\|_1 \geq \varepsilon) \leq \exp(-\frac{N\varepsilon^2}{16d^{2k}} + \frac{1}{4}).
\end{equation}
\end{lemma}
\begin{proof}
We take $X_i = \rho_i^{\otimes k} - \mathbb{E}_{\rho\in\mathcal{E}}[\rho^{\otimes k}]$ in \cref{lem:vector-bernstein}, which are $d^{2k}$ dimensional vectors. Now note that $\mathbb{E}[X_i] = \mathbb{E}_{\rho_i\in\mathcal{E}}[\rho_i^{\otimes k}] - \mathbb{E}_{\rho\in\mathcal{E}}[\rho^{\otimes k}] = 0$ by the i.i.d.\ sampling of $\rho_i$. We also trivially have $\norm{X_i}_2 = \norm{\rho_i^{\otimes k} - \mathbb{E}_{\rho\in\mathcal{E}}[\rho^{\otimes k}]}_2 \leq \sqrt{2}$, i.e., the maximal distance in 2-norm between states, and therefore we can set $\mu=\sqrt{2}$, $\sigma^2=2$. Finally, from the inequality $\norm{X - Y}_1 \leq d^k \norm{X-Y}_2$ for $d^{2k}$-dimensional vectors, we get the desired result by setting $\varepsilon' = \varepsilon/d^k$.
\end{proof}

\subsection{Pseudoentanglement}\label{app:pseudo}
In this section, we briefly recap the basic notions of pseudoentanglement, and introduce a new class of pseudoentangled states, that we name Haar-subsystem states, which will be of interest to our purposes.
The notion of \textit{pseudoentangled states} \cite{aaronson2023quantumpseudoentanglement} builds up on earlier works on quantum \textit{pseudorandomness} \cite{Ji_2018}. Pseudorandom states are a family of pure quantum states which are efficiently preparable, but still are statistically indistinguishable from Haar random states to any distinguisher using efficient (running in $O(\poly n)$ time) quantum circuits and $k = O(\poly n)$ input copies. Computational indistinguishability is a crucial aspect in the definition of pseudorandom states, since information-theoretic indistinguishability cannot be achieved with efficiently preparable states \cite{Brand_o_2016}. 
The notion of pseudoentanglement \cite{aaronson2023quantumpseudoentanglement} follows similar steps as that of pseudorandomness: pseudoentangled states are a family of pure quantum states which are statistically indistinguishable (for any distinguisher using $k = O(\poly n)$ input copies) from another family of states displaying a gap in the von Neumann entropy. As much as pseudorandom states mimic Haar random states, pseudoentangled states mimic states having (possibly) much larger entanglement. While the definition of pseudorandom states relies on efficient preparation and computational indistinguishability, we think that for the purposes of this paper these conditions are not necessary. In fact, we will employ the stronger condition of statistical indistinguishability to prove a no-go theorem on sample-efficient agnostic protocols (see \cref{app:agnosticnogo}). Then, we will settle the case of efficiently preparable states in a separate section, using slighlty different techniques (see \cref{app:stateaware}). We also think that adopting families of pseudoentangled states which are not efficiently preparable may actually be necessary to prove the existence no-go results \cref{appth:upperbounddistillableent,appth:lowerboundentcost}. Furthermore, since we are dealing with bipartite entanglement theory, we limit our definition of pseudoentangled families to a fixed (extensive) bipartition chosen beforehand. This differs from other definitions \cite{aaronson2023quantumpseudoentanglement,arnonfriedman2023computationalentanglementtheory}, which require the states to be pseudoentangled across every extensive bipartition.

\begin{definition}[Pseudoentangled states]\label{def:pseudo} A set of pseudoentangled states with gap $f(n)$ vs.\  $g(n)$ across a fixed extensive bipartition $A|B$ consists of two families of pure states $\mathcal{E} = \{ \ket{\Phi_k}\}_k$ and $\mathcal{F} = \{ \ket{\Psi_l}\}_l$ such that with very high probability the von Neumann entropy of the reduced density matrix on system A (or, equivalently, B) of $\ket{\Phi_k}$ ($\ket{\Psi_l}$, respectively) is $\Omega(f(n))$ ($O(g(n))$, respectively), and the two ensembles $\mathcal{E}$ and $\mathcal{F}$ cannot be distinguished by any quantum algorithm having access to at most $k = O(\poly n)$ copies of the states.   
\end{definition}
In Ref.\ \cite{arnonfriedman2023computationalentanglementtheory}, it has been shown that subset-phase-states are both pseudorandom and pseudoentangled, displaying extremely low von Neumann entropy $S_1 = O(\poly\log n)$ while being computationally indistinguishable from Haar random states which have $S_1 = \Omega(n)$. Subset-phase-states saturate the lower bound on the minimal entanglement necessary to computationally mimic Haar random states \cite{aaronson2023quantumpseudoentanglement}. However, as shown in \cref{sec:indstinguishablefamilies}, a maximal gap of $O(1)$ vs.\  $\tilde{\Omega}(n)$ can be obtained if one allows for our less strict definition of pseudoentanglement. 
In the following, we introduce a new family of pseudoentangled states which we find of fundamental importance to our purpose
\begin{definition}[Haar-subsystem states]\label{def:haarsubsystemstates} Given a subsystem of $m$ qubits in the $n$-qubit Hilbert space, and the set of unitaries $P = \{U_{\pi}\}_{\pi}$ implementing a random permutation on $n$-bit strings $\pi :  \{0,1\}^n \to \{0,1\}^n $, $U_{\pi}\ket{x} = \ket{\pi(x)}$, the Haar-subsystem states are defined as the following ensemble
\be
\mathcal{E}_m=\{U_{\pi} \ket{0}^{\otimes n-m}\ket{\phi_m}\,:\, U_{\pi} \in P \,, \ket{\phi_m}\sim \haar(m)\}.
\ee
\end{definition}

While it is not necessary in this work, this family can be slightly modified so that it is also efficiently preparable. In particular, it suffices to insert a pseudorandom subsystem and then apply a pseudorandom permutation. It is immediate to see that this class of states has very low entanglement if the Haar subsystem is small.

\begin{corollary}[Low entanglement states]\label{cor:haarsub} Let $\psi_{A,B}\in \mathcal{E}_m$ be a Haar-subsystem state, then its von Neumann entropy across any extensive bipartition $A|B$ satisfies
\be
S_1(\psi_A)\le m.
\ee
\begin{proof}
    The proof is immediate by definition. Indeed, any Haar-subsystem state can be expressed in the computational basis as a combination of at most $2^m$ computational basis states.
\end{proof}
\end{corollary}
The next corollary shows that the Haar-subsystem states are statistically pseudorandom for any $m=\omega(\log n)$, and therefore are pseudoentangled.

\begin{lemma}[Haar-subsystem states are statistically pseudorandom]\label{lemma:haarsub} Let $m\subset[n]$ and let $\mathcal{E}_{m}$ be the set of Haar-subsystem states. Then
\be
\|\mathbb{E}_{\ket{\psi}\sim \mathcal{E}_m}\left[\ketbra{\psi}^{\otimes k}\right]-\mathbb{E}_{\ket{\psi}\sim \haar}\left[\ketbra{\psi}^{\otimes k}\right]\|_1\le C\frac{k^2}{2^m},
\ee
with $C=\Theta(1)$. Hence for $m=\omega(\log n)$ and for any $k=O(\poly n)$, the set $\mathcal{E}_m$ is statistically pseudorandom. 
\begin{proof} In the following, we will commit a slight abuse of notation and indicate with $\Pi_{\text{sym}}^{(n,k)}$ the projector over the symmetric subspace of $n$ qubits and with $\Pi_{\text{sym}}^{(S,k)}$ the truncated projector on the symmetric subspace spanned by the set of strings $S\subseteq \{0,1\}^n$. To prove the statement, it is sufficient to note the following three facts.
\begin{enumerate}[label=(\roman*)]
    \item It is well known that the Haar average of $\psi^{\otimes k}$ gives the normalized projector over the symmetric subspace of $k$ copies of $n$ qubits (see, for example, Ref.\ \cite{Mele_2024}):
    \be
    \mathbb{E}_{\ket{\psi}\sim \haar}\left[\ketbra{\psi}^{\otimes k}\right] = \frac{\Pi_{\text{sym}}^{(n,k)}}{\tr \left( \Pi_{\text{sym}}^{(n,k)}\right)}.
    \ee
    \item Averaging over the ensemble of Haar-subsystem states appears in two places. First, one must average over Haar random states on \(n\) qubits, resulting in a state proportional to the symmetric projector $\Pi_{\text{sym}}^{(m,k)}$, which is then permuted over all possible bitstrings, as described in \cref{def:haarsubsystemstates}. Thus, we can express the average over the Haar-subsystem states as

    \be
    \mathbb{E}_{\ket{\psi}\sim\mathcal{E}_m}\left[\ketbra{\psi}^{\otimes k}\right]=\mathbb{E}_{S:|S|=2^m} \frac{\Pi_{\text{sym}}^{(S,k)}}{\tr \left( \Pi_{\text{sym}}^{(S,k)}\right)}.
    \ee

    \item Due to Lemma 2.2 in Ref.\ \cite{aaronson2023quantumpseudoentanglement}, we know that the average over the subsets of the truncated projector is close to the projector over the entire symmetric subspace as
    \be
    \left\| \mathbb{E}_{S:|S|=2^m} \left[ \frac{\Pi_{\text{sym}}^{(S,k)}}{\tr \left( \Pi_{\text{sym}}^{(S,k)}\right)} \right] - \frac{\Pi_{\text{sym}}^{(n,k)}}{\tr \left( \Pi_{\text{sym}}^{(n,k)}\right)} \right\|_1 \leq C \frac{k^2}{2^m},
    \ee
with $C = \Theta(1)$.
\end{enumerate}
Putting all items together, we get the desired claim.
\end{proof}
\end{lemma}

For one of our proofs (see \cref{app:lowerboundcost}), it will be useful to restrict the number of permutations that are used to define the Haar-subsystem states from the $n!$ that are currently used in \cref{def:haarsubsystemstates} to $2^{O(\poly n)}$. To do this, we make use of the concentration bound derived in \cref{lem:concentration}.

\begin{definition}[Restricted Haar-subsystem states]\label{def:restrictedhaarsubsystemstates} Given a subsystem of $m$ qubits in the $n$-qubit Hilbert space, and the set of unitaries $P = \{U_{\pi}\}_{\pi}$ implementing a random permutation on $n$-bit strings $\pi : \{0,1\}^n \to \{0,1\}^n $, $U_{\pi}\ket{x} = \ket{\pi(x)}$, let $P_N = \{U_{\pi_i}\}_{i=1}^{N}$ be an ensemble of $N$ uniformly sampled permutations from $P$. The \emph{restricted} Haar-subsystem states are defined as the following ensemble
\be
\mathcal{E}^{(N)}_{m}=\{U_{\pi_i} \ket{0}^{\otimes n-m}\ket{\phi_m}\,:\, U_{\pi_i} \in P_N \,, \ket{\phi_m}\sim \haar(m)\}.
\ee
\end{definition}

\begin{lemma}\label{lem:error-restricted}
    Given a Haar-subsystem state ensemble $\mathcal{E}_m$ in an $n$-qubit Hilbert space, then a randomly defined restricted Haar-subsystem state ensemble $\mathcal{E}^{(N)}_{m}$ satisfies,
    \begin{equation}
        \left\| \mathbb{E}_{\psi\sim\mathcal{E}^{(N)}_{m}} [\psi^{\otimes k}] - \mathbb{E}_{\psi\sim\mathcal{E}_{m}} [\psi^{\otimes k}] \right\|_1 \leq \frac{k}{2^n},
    \end{equation}
    with probability $1-\exp\left(-O\left(\frac{Nk^2}{2^{2(k+1)n}}\right)\right)$ over the random choice of permutations to construct $\mathcal{E}^{(N)}_{m}$. Whenever $k \in O(\poly n)$ and $N = 2^{\omega(\poly n)}$, then the trace distance error is in $\text{negl}(n)$ with probability $1-o(1)$.
\end{lemma}
\begin{proof}
    We apply the concentration bound of \cref{lem:concentration} with $\varepsilon = k/2^n$, which gives the desired trace distance error with probability $1-\exp(-\frac{Nk^2}{16d^{2(k+1)}} + \frac{1}{4})$ over the choice of permutations. The result for $k \in O(\poly n)$ and $N = 2^{\omega(\poly n)}$ follows immediately.
\end{proof}

\section{State-agnostic entanglement theory}\label{App:stateagnostic}
In this section, we prove the main results concerning optimal rates for state-agnostic distillation and dilution protocols for bipartite pure quantum states when the computational resources are constrained to be $O(\poly n)$. 
In \cref{App:agnosticgo}, we show an achievability result (\cref{th:distillation}), namely that the rate $\min\{\Theta(S_{\min}),\Theta(\log k)\}$ can be achieved via the sample-efficient and state-agnostic protocol based on the Schur transform and introduced in \cref{section:protocol}. In this regime, this protocol is also computationally efficient. We show that this protocol is robust to errors and can therefore be applied also to almost-i.i.d.\  quantum states. 

In \cref{sec:indstinguishablefamilies} we introduce new families of pseudoentangled states. Differently from any previously known result, these states display a maximal, tunable, gap (up to log factors) in the von Neumann entropy of $o(1)$ vs $\tilde{\Omega}(n)$ (\cref{lemma:indistinguishability,lemma:entanglement}).
In \cref{app:agnosticnogo}, we leverage on the previously introduced classes of pseudoentangled states to show that any sample-efficient and state-agnostic protocol cannot improve on the previous bound $\min\{\Theta(S_{\min}),\Theta(\log k)\}$. We then employ the same classes of states to show sample-complexity lower bounds for estimating and testing the von Neumann entropy (\cref{cor:lowerboundsestimatingvnentropy,cor:lowerboundstestingvnentropy}). These improve on previously known results since they are valid also when fixing a specific range for the entropy, notably also $S_1 = \Theta(1)$. In \cref{th:entanglementcostnogo} we show similar optimality results for state-agnostic and sample-efficient entanglement dilution protocols. Here we show that in the worst case, the optimal protocol is simply quantum teleportation, and $\tilde{\Omega}(n)$ ebits are needed to dilute the state, even when $S_1(\psi_A) = o(1)$.

Finally, in \cref{App:tightdistillable,app:lowerboundcost} we show that the previous worst-case bounds are not solely due from state-agnosticity, but are ultimately a consequence of limiting the LOCC to be computationally efficient. In \cref{appth:upperbounddistillableent,appth:lowerboundentcost} we show that the previous bound holds for the computational distillable entanglement and the computational entanglement cost, respectively. That is, there exists states for which any computationally efficient LOCC cannot improve on the previous bound, that is saturated by state-agnostic protocols.
All together, these results settle the question of what kind of entanglement manipulation is possible to achieve using computationally efficient LOCC protocols in the worst case.

\subsection{Achievability result for state-agnostic entanglement distillation}\label{App:agnosticgo}
We start with proving the achievability result for state-agnostic and computationally efficient entanglement distillation protocols.
\begin{theorem}[Achievability result for state-agnostic entanglement distillation. Formal version of \cref{th:informal2}]\label{th:distillation}
Given an extensive bipartition of an $n$-qubit system $A|B$, an unknown bipartite pure state $\ket{\psi_{A,B}}$ with min-entropy $S_{\min}(\psi_A)$, and a suitable number of input copies $k = \poly(n)$, the computational distillable entanglement of $\ket{\psi_{A,B}}$  obeys
\begin{equation}
\hat{E}_D^{(k)}(\psi_{A,B}) \geq \min\{\frac{1}{20}S_{\min}(\psi_A), \frac{1}{20}\log k  \}.
\end{equation}
That is, there exists a $k$-shot, computationally efficient and state-agnostic protocol that achieves the bound on the right-hand side.
\end{theorem}
\begin{proof} Let us consider the protocol described in \cref{section:protocol} applied to $k$ copies of the input state $\ket{\psi_{A,B}}$. As already discussed, the protocol is state-agnostic and computational-efficient when $k = \poly(n)$. We now want to show that for a (sufficiently high) polynomial number of input copies $k$ we can distill at a rate $\Omega(S_{\min}(\psi_A))$ if $S_{\min}(\psi_A) \le\log k$, while we distill at a rate $ \Omega(\log k)$ if $S_{\min}(\psi_A) > \log k$ instead.
Let us start from the probability of distilling $\log\dim\mathcal{V}_\lambda$ ebits when projecting onto an irrep $\lambda$, where we recall that $\mathcal{V}_\lambda$ is the $\lambda$-th irrep of the symmetric group $S_k$. This reads (see \cref{proj})
\begin{equation}
\Pr(\lambda) = \frac{\dim\mathcal{V}_\lambda}{k!}\sum_{\pi\in S_k} \chi^{\lambda}(\pi) \tr (R_{\pi} \rho_A^{\otimes k}).
\end{equation}
We can bound the probability noting that
\begin{equation}
\tr R_{\pi}\rho_A^{\otimes k} \leq \|\rho_A\|^{k-|\pi|} \equiv \gamma^{|\pi|-k}\label{bound-tr},
\end{equation}
where $\|\cdot\|$ is the operator norm, $|\pi|$ is the number of cycles of the permutation $\pi$ and we introduced $\gamma \coloneqq \|\rho_A\|^{-1}$. \cref{bound-tr} follows since we can decompose the permutation $\pi$ into $m\le k-1$ disjoint cycles $\pi = p_1 p_2 ... p_m$, so that we have
\begin{align}
\tr R_{\pi}\rho_A^{\otimes k} \overset{\text{(i)}}= \prod_{i=1}^m \tr (R_{p_i} \rho_A^{\otimes k}) \overset{\text{(ii)}}= \prod_{i=1}^m \tr(\rho_A^{k-|p_i|+1}) \overset{\text{(iii)}}\leq \prod_{i=1}^m \| \rho_A \|^{k-|p_i|}  \overset{\text{(iv)}}= \|\rho_A\|^{k-|\pi|},
\end{align}
where in (i) we have used  the fact that the $p_i$'s are disjoint permutations; in (ii) we have used that each $p_i$ can be decomposed into $k-|p_i|$ swaps
and we repeatedly applied the partial swap-trick $\tr_2 R_{(12)}A\otimes B = AB $; in (iii) we applied the inequality $\tr(\rho_A^{k+1})\leq \|\rho_A\|^k \|\rho_A \|_1 \leq \|\rho_A\|^k$ and in (iv) we have used that $\sum_i(k-|p_i|)$ is the sum of the lengths of the $p_i$'s (the number of elements in each cycle minus one), which is exactly $k - |\pi|$. Note that the previous bound is tight, as can be checked inserting a maximally mixed state with any choice of support. \cref{bound-tr} allows us to bound the probabilities as follows
\begin{align}
\Pr(\lambda) &\overset{\text{(i)}}\leq \frac{\dim\mathcal{V}_\lambda}{k!}\sum_{\pi\in S_k} |\chi^{\lambda}(\pi)| \tr (R_{\pi} \rho_A^{\otimes k})\\ &\overset{\text{(ii)}}\leq \frac{\dim\mathcal{V}_\lambda}{k!}\sum_{\pi\in S_k} |\chi^{\lambda}(\pi)| \gamma^{|\pi|-k} \\
&\overset{\text{(iii)}}\leq \mu(\lambda) \sum_{\pi\in S_k} \gamma^{|\pi|-k} = \mu(\lambda) \frac{k!}{\gamma^k} \frac{1}{k!} \sum_{\pi\in S_k} \gamma^{|\pi|} \\
&\overset{\text{(iv)}}= \mu(\lambda) \frac{k!}{\gamma^k} \tr\Pi_{\text{sym}}^{(\gamma,k)} \overset{\text{(v)}}= \mu(\lambda) \frac{k!}{\gamma^k} \binom{k+\gamma-1}{\gamma-1} \\
&= \mu(\lambda)\frac{1}{\gamma^k}\gamma(\gamma+1)...(\gamma+k-1) \\
&\leq \mu(\lambda) \left( \frac{\gamma+k}{\gamma} \right)^k \\
&\overset{\text{(vi)}}\leq \mu(\lambda) e^{k^2/\gamma}\label{bound-prob}.
\end{align}
Here, in (i) we applied triangle inequality; in (ii) we inserted \cref{bound-tr}; in (iii) we have used that $|\chi^{\lambda}(\pi)| \leq \dim\mathcal{V}_\lambda$ and we introduced the Plancherel measure over the symmetric group $\mu(\lambda) \coloneqq (\dim\mathcal{V}_\lambda)^2/k!$~\cite{weyl1946classical}; in (iv) we have used that $\frac{1}{k!} \sum_\pi \gamma^{|\pi|}$ is the trace of the projector $\Pi_{\text{sym}}^{(\gamma,k)}$ over the symmetric subspace $\text{Sym}^{k}(\mathbb{C}^\gamma)$ of $k$ elements of local dimension $\gamma$ \cite{weyl1946classical}, the trace of which is the dimension $\dim \text{Sym}^{(\gamma,k)}(\mathbb{C}^\gamma) = \binom{k+\gamma-1}{\gamma-1}$, as used in (v) . Finally in (vi) we have used  the simple fact $(1+x)^k\leq e^{xk}$ when $k>0$.
Now let us set $k\le \gamma$ and let us evaluate the probability of distilling less than $ck\log k$ ebits with $c=\Theta(1)$
\be
\Pr(\lambda\,:\, \dim\mathcal{V}_{\lambda}\le 2^{ck\log k})&\le \sum_{{\lambda\,:\, \dim V_{\lambda}\le 2^{ck\log k}}}\frac{(\dim V_{\lambda})^2}{k!}\\
&\le e^{\frac{k^2}{\gamma}}\frac{e^{2ck\log k}}{k!}\sum_{{\lambda\,:\, \dim V_{\lambda}\le 2^{ck\log k}}}\\
&\le e^{\frac{k^2}{\gamma}}\frac{e^{2ck\log k}}{k!} \frac{e^{\pi\sqrt{2k/3}}}{k^{3/4}}\\
&=e^{k}\frac{e^{2ck\log k}}{k!} \frac{e^{\pi\sqrt{2k/3}}}{k^{3/4}}\label{eq133232},
\ee
where we have used that the cardinality of integer partition of a number is upper bounded by $\frac{e^{\pi\sqrt{2k/3}}}{n^{3/4}}$~\cite{de2009simple}. Notice that the upper bound on the probability in \cref{eq133232} is arbitrary small for large enough $k$ and any $c<1/2$. However, since we set $k\le \gamma\coloneqq 2^{S_{\min}}$ to derive explicit lower bounds on the validity of the distillation protocol, we impose \cref{eq133232} to be (strictly) less than $1/3$. To do this, we set $c=1/4$ for simplicity, and derive that the probability of distilling at least $\frac{k}{4}\log k$ many ebits is lower bounded by
\be
\Pr(\lambda\,:\, \dim\mathcal{V}_{\lambda}\ge 2^{\frac{k}{4}\log k})\ge\frac{2}{3},\quad\forall S_{\min}\ge \log_2k\ge \frac{13}{2}\label{eq:conditionprobability},
\ee
where the condition on $S_{\min}$ comes from the fact that we set $k\le 2^{S_{\min}}$. From \cref{eq:conditionprobability}, it is easy to derive the rate of the protocol: if $S_{\min}(\psi_A)=O(\log n)$, we proved in \cref{eq:conditionprobability} that, setting $k=2^{S_{\min}}=O(\poly n)$, the protocol distills $\frac{k}{4}S_{\min}$ many ebits and has a rate of $\frac{1}{4}S_{\min}$. On the other hand, $S_{\min}(\psi_A) = \omega(\log n)$, then we cannot set $k=2^{S_{\min}}=\omega(\poly n)$ anymore; thus from \cref{eq:conditionprobability} we see that we distill $\frac{k}{4}\log k$ many ebits and, therefore, the rate is $\frac{1}{4}\log k$, which proves the other lower bound. 

The last point that remains to be addressed is the state-agnosticism of the protocol, since in the last part of the proof we apparently used some information about the state, i.e.,  the min-entropy, to choose a suitable value of $k$. When $S_{\min}(\psi_A) = \omega(\log n)$ the choice is effectively independent of $S_{\min}(\psi_A)$ since we proved the result for any $k = \poly(n)$. If instead $S_{\min}(\psi_A) = O(\log n)$, apparently we must choose $k$ large enough to achieve $S_{\min}(\psi_A)$ but not too large, otherwise the bound \cref{eq133232} will not hold anymore. Hence there are two situations that may arise. If $k\ge 2^{S_{\min}}$, thanks to \cref{thm:optimalprotocol}, we know that for any $k=O(\poly n)$ and $k\ge k'=2^{S_{\min}}$, we have
\be
\Pr(\lambda \in \mathcal{I}_{(k,k)} \,:\, \dim\mathcal{V}_{\lambda}\le 2^{\frac{k}{4}S_{\min}})\le \Pr( \lambda \in \mathcal{I}_{(k',k')} \,:\, \dim\mathcal{V}_{\lambda}\le 2^{k^{\prime}\log k^{\prime}})<\frac{1}{3},
\ee
and the protocol processing $k$ copies distills $S_{\min}/4$ ebits per input state copy as well. If $k\le 2^{S_{\min}}$, the protocol distills $\log k$ ebits per copy.

The above analysis proves the existence of a state agnostic protocol distilling $\min\{\frac{1}{4}S_{\min},\frac{1}{4}\log k\}$ with probability $2/3$. In the definition of distillable entanglement, see \cref{def:distillableentanglement}, the probability must be greater than $\sqrt{64/65}$ due to analytics artifacts. To speed the probability of such a protocol up to, it is sufficient to run the protocol $5$ times to be sure to distill with a rate of $\min\{\frac{1}{4}S_{\min},\frac{1}{4}\log k\}$, giving an overall rate of $\min\{\frac{1}{20}S_{\min},\frac{1}{20}\log k\}$. Hence the theorem is proven.

\end{proof}

\begin{remark}
In \cref{th:distillation}, we have shown that the agnostic protocol described in \cref{section:protocol} distills $1/4\min\{S_{\min},\log k\}$ provided that $S_{\min}\ge\frac{13}{2}$. Let us now consider the case where $S_{\min}<\frac{13}{2}$. 

First, note that, thanks to Ref.~\cite{odonnell2015efficientquantumtomography}, Alice and Bob can determine whether the operator norm $\|\rho_A\|$ is either $<2^{-13/2}$ or $\ge 2^{-13/2}$. When $S_{\min}<\frac{13}{2}$, the key idea is to apply the agnostic protocol to multiple copies of $\ket{\psi_{AB}}$ in a coherent fashion in order to boost the operator norm. Indeed, observe that $\|\rho_{A}^{\otimes l}\|\le\|\rho_A\|^{l}$. Hence, by running the protocol on $\ket{\psi_{AB}}^{\otimes l}$ with $l$ chosen such that $S_{\min}(\rho_A)\ge \frac{13}{2l}$, Alice and Bob can still distill $\frac{1}{k}(\frac{1}{4}\frac{k}{l}(lS_{\min}))=\frac{1}{4}S_{\min}$.  Therefore, we conclude that the restriction in \cref{th:informal2}, requiring $S_{\min}\ge\frac{13}{2}$, is merely artificial. In an agnostic scenario, Alice and Bob need at most $2^{13}$ additional copies to measure $\|\rho_A\|$ and determine the appropriate number of copies $l$ for applying the agnostic distillation protocol coherently. Of course, the efficiency of the protocol breakes whenever $S_{\min}=o(1/\poly(n))$, a regime which is not interesting because, as we show in \cref{sec:LOCCtomography}, states with low entanglement, when (crucially) quantified by $S_{\min}$, are well approximated by the tensor product of the principal components of the reduced density matrices on $A$ and $B$ respectively. 

\end{remark}

\subsubsection{Robustness of the distillation protocol on almost-i.i.d.\  states}
While our paper primarily focuses on the theory of bipartite pure-state entanglement, in this section, we analyze the robustness of the protocol against deviations from pure states, and most importantly, with the absence of the exact i.i.d. assumption. Specifically, we examine the behavior of (potentially mixed) input states that are close in trace distance to $k$ copies of some pure state $\psi_{A,B}$ with min-entropy $S_{\min}(\psi_A)$.

To address this scenario, it is important to note that a potentially mixed state $\rho_{A,B}\in(\mathbb{C}^{d})^{\otimes k}$, when decomposed in the Schur basis, loses the key property of being maximally entangled within the subspaces spanned by the projectors $\Pi_{\mathcal{V}_{\lambda}}$. Therefore, the subspaces in $A$ and $B$ are no longer perfectly correlated. This property is essential for the state agnostic protocol considered in \cref{th:distillation} to work properly. In other words, \cref{lemma:state} critically relies on the assumption that the input state is pure.

For input states that may be mixed and no longer i.i.d., Alice and Bob cannot have any guarantee that their local Schur sampling projects onto the same space $\mathcal{V}_{\lambda}$. This uncertainty can be mitigated by allowing classical communication between Alice and Bob, which is entirely unnecessary when dealing with pure input states.

The modified protocol proceeds as follows: Alice and Bob apply the algorithm described in \cref{algo1} to a potentially mixed input state $\rho_{A,B}$, which should still be close to being i.i.d. After performing the projective measurement $\{\Pi_{\lambda_A}^{A}\otimes \Pi_{\lambda_B}^{B}\}_{\lambda_A,\lambda_B}$, they communicate their measurement outcomes $\lambda$ to each other. If the outcomes agree, i.e., $\lambda = \lambda_A = \lambda_B$, they can be confident that they have distilled $\log \dim \mathcal{V}_{\lambda}$ ebits (up to some error). Conversely, if the outcomes disagree, i.e., $\lambda_A \neq \lambda_B$, they abort the process and start anew with fresh copies of $\rho_{A,B}$. 
As is evident, this protocol inherently requires two-way classical communication. In contrast, with the promise of pure input states $\ket{\psi_{A,B}}^{\otimes k}$, no communication is needed, as \cref{lemma:state} guarantees that Alice and Bob obtain the same outcome $\lambda$ and thus distill $\log \dim \mathcal{V}_{\lambda}$ ebits.

Let us analyze the performance of this protocol on potentially mixed states $\rho_{A,B}\in(\mathbb{C}^{d})^{\otimes k}$. As announced before, we assume that $\rho_{A,B}$ is close to some i.i.d. source of pure states. We then define the optimal choice of this pure i.i.d. state as the one that maximises the min-entropy
\be
\ket{\psi_{\rho}}\coloneqq\operatorname{argmax}\{S_{\min}(\psi_A)\,:\, \|\rho_{A,B}-\ketbra{\psi}^{\otimes k}\|_1\le \delta,\,\psi_A\coloneqq\tr_B\ketbra{\psi}\}.
\ee
Now, first we need to estimate the probability that the protocol fails because Alice and Bob get different outcomes $\lambda\neq\lambda'$. Setting $\Pi_{(\lambda,\lambda')}^{A,B} \coloneqq \Pi_\lambda^A \otimes \Pi_{\lambda'}^B$, we have
\be
\Pr\left( \lambda \neq \lambda' | \rho_{A,B}  \right) = \sum_{\lambda\neq\lambda'} \tr\left( \Pi_{(\lambda,\lambda')}^{A,B} \rho_{A,B}\right) \leq \delta. \label{eq:neq}
\ee
Indeed, the state $\rho_{A,B}$ is guaranteed to be $\delta$-close in trace distance to a pure i.i.d.~state $\psi_{\rho}^{\otimes k}$ with $\Pr\left( \lambda \neq \lambda' | \psi_{\rho}^{\otimes k}  \right) = 0$, and all the projectors corresponding to different tuples $(\lambda,\lambda')$ are orthogonal among each other, so by the variational definition of trace distance we must have \cref{eq:neq}. From now on, we will assume that $\lambda=\lambda'$.
Adapting the analysis of \cref{th:distillation}, the probability of distilling at least $\frac{k}{4}\log k$ ebits from $\rho_{A,B}$, with $ k\le 2^{S_{\min}(\psi_{\rho,A})}$, can be bounded as follows:
\be
\Pr(\lambda\,:\, \log\dim\mathcal{V}_{\lambda}\ge \frac{k}{4}\log k)&=\sum_{\lambda\,:\, \log\dim\mathcal{V}_{\lambda}\ge \frac{k}{4}\log k}\tr(\Pi_{\mathcal{V}_{\lambda}}^{A}\otimes \Pi_{\mathcal{V}_{\lambda}}^{B}\rho_{A,B})\\
&\ge \sum_{\lambda\,:\, \log\dim\mathcal{V}_{\lambda}\ge \frac{k}{4}\log k}\tr(\Pi_{\mathcal{V}_{\lambda}}^{A}\otimes \Pi_{\mathcal{V}_{\lambda}}^{B}\psi_{\rho}^{\otimes k}) -\delta\\
&\ge 1-e^{k}\frac{e^{\frac{k}{2}\log k}}{k!} \frac{e^{\pi\sqrt{2k/3}}}{k^{3/4}}-\delta\\
&\ge \frac{2}{3} - \delta,\quad\quad \forall S_{\min}\ge\log_2k\ge\frac{13}{2} .
\ee
Therefore, whenever $S_{\min}(\psi_{\rho,A}) = O(\log n)$, for $k$ sufficiently large we will distill $\delta$-approximate ebits at a rate $\frac{1}{4}S_{\min}(\psi_{\rho,A})$ with probability of success $(1-\delta)(2/3-\delta)$, while whenever $S_{\min}(\psi_{\rho,A}) = \omega(\log n)$ we will distill $\delta$-approximate ebits distill at a rate $\frac{1}{4}\log k$ with probabilty arbitrarily close to $(1-\delta)^2$.
The above result therefore showcases the robustness of our distillation protocol whenever an iid state exists which is sufficiently close to the input.

\subsection{Pseudoentangled families of states with a tunable gap in the von Neumann entropy}\label{sec:indstinguishablefamilies}
In this section, we introduce new families of pseudoentangled states (see \cref{def:pseudo}) on a fixed extensive bipartition $A|B$. These are statistically indistinguishable to any distinguisher using $O(\poly n)$ samples, yet they display a tunable gap in their von Neumann entropies, which can possibly be maximal ( $\Tilde{\Omega}(n)$ vs. $o(1)$). The states introduced here will be of primary importance for proving all the main no-go results of our paper.
Our construction relies on the Haar-subsystem states defined previously (see \cref{def:haarsubsystemstates}) and is based on the following two Lemmas: first, in \cref{lemma:indistinguishability}, we prove the statistical indistinguishability of the two families, and then, in \cref{lemma:entanglement} we show the tunable gap in their von Neumann entropies.

\begin{lemma}[Statistical indistinguishability of the two state families]\label{lemma:indistinguishability} 
    For any extensive bipartition $A|B$, and any $0\le S_{\min}\le \min\{n_A,n_B\}$, $0 \leq \eta \leq 1$, the following two ensembles of $n$-qubit bipartite pure states are statistically indistinguishable when having access to $ k = O(\poly n)$ samples:
    \be
\mathcal{E}_{\haar,\eta,S_{\min}}&=\{\sqrt{1-\eta}\ket{0}_{A}\ket{0}_{B}\ket{0}^{\otimes n-2S_{\min}} \ket*{\phi^{+}_{A,B}}^{\otimes S_{\min}}+\sqrt{\eta}\ket{1}_{A}\ket{1}_{B}\ket*{\psi_{A,B}},\quad \ket*{\psi_{A,B}}\sim \haar(n)\},\\
\mathcal{E}_{m,\eta,S_{\min}}&=\{\sqrt{1-\eta}\ket{0}_{A}\ket{0}_{B}\ket{0}^{\otimes n-2S_{\min}} \ket*{\phi^{+}_{A,B}}^{\otimes S_{\min}}+\sqrt{\eta}\ket{1}_{A}\ket{1}_{B}\ket*{\psi_{A,B}},\quad \ket*{\psi_{A,B}}\sim \mathcal{E}_{m}\}. \label{eq:ens}
    \ee
Here $\mathcal{E}_{m}$ is the ensemble of Haar-subsystem states (see \cref{def:haarsubsystemstates}) on $m = \omega(\log n)$ qubits.
\begin{proof}
To prove the statement, we can simplify the notation and consider
\be
\ket{\psi_{\eta}} = \sqrt{1-\eta}\ket{\psi_1}+\sqrt{\eta}\ket{\psi_2},
\ee
where $\ket{\psi_1}=\ket{0}_{A}\ket{0}_{B}\ket{0}^{\otimes n - 2S_{\min}}\ket*{\phi^{+}_{A,B}}^{\otimes S_{\min}}$ and $\ket{\psi_2}=\ket{1}_{A} \ket{1}_{B}\ket*{\psi_{A,B}}$. For the ensemble $\mathcal{E}_{\haar,\eta,S_{\min}}$ $\ket*{\psi_{A,B}}\sim\haar(n-2)$, while for $\mathcal{E}_{m,\eta,S_{\min}}$, we have $\ket*{\psi_{A,B}}\sim\mathcal{E}_m$. It is easy to see that, in both cases, the measure over the two ensembles is invariant under multiplication by a phase $e^{i\theta}$ with $\theta\in[0,2\pi)$. This means that we can analogously look at the following class of states
\be
\ket{\psi_{\eta,\theta}}=\sqrt{1-\eta}e^{i\theta}\ket{\psi_1}+\sqrt{\eta}\ket{\psi_2} \quad \theta\sim[0,2\pi),
\ee
and consider the average over $\theta\sim[0,2\pi)$ first. This will simplify the calculation considerably.
Notice also that $\ket{\psi_1}$ and $\ket{\psi_2}$ are orthogonal vectors, this will be crucial in the following. To show statistical indistinguishability with $k=O(\poly n)$ copies, the equality 
\be
\|\mathbb{E}_{\theta}\mathbb{E}_{\ket*{\psi_{A,B}}\sim\operatorname{Haar}}\ketbra{\psi_{\eta,\theta}}^{\otimes k}-\mathbb{E}_{\theta}\mathbb{E}_{\ket*{\psi_{A,B}}\sim\mathcal{E}_{m}}\ketbra{\psi_{\eta,\theta}}^{\otimes k}\|_1=o\left(\frac{1}{\operatorname{poly}n}\right).\label{boundist}
\ee
must be proven. Let us first compute the average over $\theta\in[0,2\pi)$. Borrowing results from Section 5.3 of Ref.~\cite{Anshu_2022}, we can write
\be
\mathbb{E}_{\theta}\ketbra{\psi_{\eta,\theta}}^{\otimes k}=\sum_{t=0}^{k}\binom{k}{t}(1-\eta)^{k-t}\eta^{t}\ketbra{\phi_t},\quad \ket{\phi_t}=\frac{1}{\sqrt{\binom{k}{t}}}\sum_{T\subseteq[k], |T|=t}\ket{\psi_1}_{\bar{T}}\otimes \ket{\psi_2}_{T},
\ee
where $\ket{\psi_2}_{T}$ denotes the state obtained by storing $\ket{\psi_2}$ in only $t$ of the $k$ registers (the definition is analogous for $\ket{\psi_1}_{\bar{T}}$, where $\bar{T}$ is the complement of $T$). Generalizing Lemma 11 of Ref.~\cite{Anshu_2022}, since the two states are orthogonal, we know that there exists a unitary $U_t$ such that $U_t \ket{\psi_1}^{\otimes t}\ket{\psi_2}^{\otimes k-t}=\ket{\phi_t}$, for any $t\in[k]$. We can then bound the trace distance between the two ensembles as
\be
&\|\mathbb{E}_{\theta}\mathbb{E}_{\ket*{\psi_{A,B}}\sim\operatorname{Haar}}\ketbra{\psi_{\eta,\theta}}^{\otimes k}-\mathbb{E}_{\theta}\mathbb{E}_{\ket*{\psi_{A,B}}\sim\mathcal{E}_{m}}\ketbra{\psi_{\eta,\theta}}^{\otimes k}\|_1\\
&\overset{\text{(i)}}\le \sum_{t=0}^{k}\binom{k}{t}(1-\eta)^{k-t}\eta^{t}\|\mathbb{E}_{\ket{\psi}\sim\operatorname{Haar}}\ketbra{\phi_t}-\mathbb{E}_{\ket{\psi}\sim\mathcal{E}_m}\ketbra{\phi_t}\|_1\\
&\overset{\text{(ii)}}=\sum_{t=0}^{k}\binom{k}{t}(1-\eta)^{k-t}\eta^{t}\left\|U_t\left(\mathbb{E}_{\ket*{\psi_{A,B}}\sim\operatorname{Haar}}\ketbra{\psi_1}^{\otimes t}\otimes \ketbra{\psi_2}^{\otimes k-t}-\mathbb{E}_{\ket*{\psi_{A,B}}\sim\mathcal{E}_m}\ketbra{\psi_1}^{\otimes t}\otimes \ketbra{\psi_2}^{\otimes k-t}\right)U_t^{\dag}\right\|_1\\
&=\sum_{t=0}^{k}\binom{k}{t}(1-\eta)^{k-t}\eta^{t}\left\|U_t \ketbra{\psi_1}^{\otimes t}\otimes \left(\mathbb{E}_{\ket*{\psi_{A,B}}\sim\operatorname{Haar}} \ketbra*{\psi_{A,B}}^{\otimes k-t}-\mathbb{E}_{\ket*{\psi_{A,B}}\sim\mathcal{E}_m} \ketbra*{\psi_{A,B}}^{\otimes k-t}\right)U_t^{\dag}\right\|_1\\
&\overset{\text{(iii)}}\le\sum_{t=0}^{k}\binom{k}{t}(1-\eta)^{k-t}\eta^{t}\|\mathbb{E}_{\ket*{\psi_{A,B}}\sim\operatorname{Haar}}\ketbra*{\psi_{A,B}}^{\otimes k-t}-\mathbb{E}_{\ket*{\psi_{A,B}}\sim\mathcal{E}_m}\ketbra*{\psi_{A,B}}^{\otimes k-t}\|_1\\
&\overset{\text{(iv)}}\le \sum_{t=0}^{k}\binom{k}{t}(1-\eta)^{k-t}\eta^{t}\frac{C(k-t)^2}{2^m}\le O\left(\frac{k^2}{2^m}\right)=o\left(\frac{1}{\operatorname{poly}n}\right),
\ee
where in (i) we have used  triangle inequality, in (ii) we have made use of the fact that $U_t \ket{\psi_1}^{\otimes t}\ket{\psi_2}^{\otimes k-t}=\ket{\phi_t}$ (note that $U_t$ is the same for both families), in (iii) we made use of the unitary invariance of the trace distance and the multiplicativity property with respect to tensor products. Finally in (iv), we have used  \cref{lemma:haarsub}:
\be
\|\mathbb{E}_{\ket*{\psi_{A,B}}\sim\operatorname{Haar}}\ketbra{\psi_{A,B}}^{\otimes k-t}-\mathbb{E}_{\ket*{\psi_{A,B}}\sim\mathcal{E}_m}\ketbra*{\psi_{A,B}}^{\otimes k-t}\|_1\le\frac{C(k-t)^2}{2^m},\quad C=\Theta(1),
\ee
which is valid for Haar-subsystem states. Since the number of qubits in the Haar-subsystem states is $m = \omega(\log n)$, this shows \cref{boundist} and proves the desired claim. 
\end{proof}
\end{lemma}

\begin{lemma}[Scaling of von Neumann entropies of the two state families]\label{lemma:entanglement}
    Consider an extensive bipartition of an $n$-qubit system $A|B$, and a bipartite pure state $\ket{\psi_{A,B, \eta}}$ sampled from one of the two sets $\mathcal{E}_{\haar,\eta,S_{\min}}$, $\mathcal{E}_{m,\eta,S_{\min}}$, with $m=\log^{1+c}n_A$ for any $c>0$, and $\eta = S_1/n_A$. Then, if $S_{1}=o(n_A/\log^{1+c}n_A)$ and $S_{\min}=O(\log n_A)$ the von Neumann entropy of the reduced density matrix $\ket{\psi_{A;\eta}}$ obeys
    \be
    S_{1}(\psi_{A;\eta})&=S_{\min}+S_{1}+o(1),\quad \ket{\psi_{\eta}}\sim\mathcal{E}_{\haar,\eta,S_{\min}},\\
    S_{1}(\psi_{A;\eta})&=S_{\min}+o(1),\quad \ket{\psi_{\eta}}\sim\mathcal{E}_{m,\eta,S_{\min}}\label{eqlemma:easy1},
    \ee
    with overwhelmingly high probability. Furthermore, the min-entropy is $S_{\min}(\psi_{A;\eta})=S_{\min}+o(1)$ for both ensembles. If instead $S_{1}=O(n_A/\log^{1+c}n_A)$, bust still $S_{\min}=O(\log n_A)$ then
    \be
S_{1}(\psi_{A;\eta})&=S_{\min}+S_{1}+O(1),\quad \ket{\psi_{\eta}}\sim\mathcal{E}_{\haar,\eta,S_{\min}},\\
    S_{1}(\psi_{A;\eta})&=S_{\min}+ O(1),\quad \ket{\psi_{\eta}}\sim\mathcal{E}_{m,\eta,S_{\min}}\label{eqlemma:hard1},
    \ee
     with overwhelmingly high probability, and the min-entropy being $S_{\min}(\psi_{A;\eta})=S_{\min}+O(1)$ for both ensembles.
    \begin{proof}
        We prove Eq.~\eqref{eqlemma:easy1}. The proof for \eqref{eqlemma:hard1} follows by identical reasoning. For both ensembles, the reduced density matrix on system $A$ takes the form
        \be
    \psi_{A;\eta}=(1-\eta)\ketbra{0}_{A}\otimes \ketbra{{0}}^{\otimes (n_A -S_{\min})}\otimes\frac{I_{2^{S_{\min}}}^A}{2^{S_{\min}}}+\eta\ketbra{1}_{A}\otimes\psi_{A}.
        \ee
    The corresponding von Neumann entropy is then given by
    \be
    S_1(\psi_{A;\eta})=(1-\eta)S_{\min}+\eta S_1(\psi_A)+h_2(\eta),
    \ee
    where $h_2(\eta)$ is the binary entropy $h_2(\eta) \coloneqq -\eta\log\eta-(1-\eta)\log(1-\eta)$. Now, we know that for $\ket*{\psi_{A,B}}\sim \haar(n)$, with high probability, one has $S_1(\psi_{A})= n_A-O(1)$ \cite{Mele_2024}, whereas for $\ket*{\psi_{A,B}}\in\mathcal{E}_m$ it holds that $S_1(\psi_{A})\le m$ (see \cref{cor:haarsub}). Now, since $\eta = S_1/n_A$, $S_{1}=o(n_A/\log^{1+c}n_A)$ and $S_{\min}=O(\log n_A)$, if $\ket{\psi_{\eta}}\sim\mathcal{E}_{\haar,\eta,S_{\min}}$ we have
    \be
    S_{1}(\psi_{A;\eta})&= S_{\min} + S_1 -\eta(S_{\min} + O(1)) + h_2(\eta) \\
    &= S_{\min}+S_{1}+o(1)\label{eq1}.
    \ee
    If instead $\ket{\psi_{\eta}}\sim\mathcal{E}_{m,\eta,S_{\min}}$, since $m=\log^{1+c} n_A$ we have
    \be
    S_{1}(\psi_{A;\eta})&\leq S_{\min} + \eta(m - S_{\min}) + h_2(\eta) \\
    &= S_{\min} + o(1)\label{eq2},
    \ee
    with the $\geq$ inequality holding trivially. Both \cref{eq1,eq2} hold with overwhelmingly high probability. 
    Concerning the min-entropy, for both ensembles we have
    \be
    S_{\min}(\psi_{A;\eta})&= \min\{ -\log\eta + S_{\min}(\psi_A), -\log(1-\eta) + S_{\min}\} \\
    &= \min\{S_{\min}(\psi_A) + \Omega(1), S_{\min} + o(1)\}.
    \ee
    Now, if $\ket*{\psi_{A,B}}\sim \haar(n)$, with high probability, one has $S_{\min}(\psi_{A})= n_A-O(1)$, and since, in particular, $S_{\min}=o(n_A)$ we must have $S_{\min}(\psi_{A;\eta})=S_{\min}+o(1)$. If instead $\ket*{\psi_{A,B}}\sim\mathcal{E}_m$, since the state is also pseudorandom (see \cref{lemma:haarsub}), one has $2S_{\min}(\psi_{A}) \geq S_{2}(\psi_{A}) = \omega(\log n_A)$ \cite{aaronson2023quantumpseudoentanglement}, while $S_{\min}=O(\log n_A)$, this proves the statement also in this last case.
    \end{proof}
\end{lemma}

\subsection{Fundamental limitations on state-agnostic entanglement manipulation, estimation and testing}\label{app:agnosticnogo}
In this section, we show no-go results for sample-efficient and state-agnostic entanglement manipulation, in particular for distillation and dilution tasks. This means that the protocols are assumed not to rely on any prior information on the state to be diluted or distilled from. The following no-go theorems do not rely on any computational assumptions, and they require the protocols to be sample-efficient and not necessarily computational-efficient. For both distillation and dilution tasks, we show that in the worst-case scenario (i.e.,  for some classes of states) any state-agnostic and sample-efficient protocol cannot perform better than a certain rate. Crucially, the no-go bounds that we derive are optimal, since they are saturated by specific choices of existing protocols.

For entanglement distillation (\cref{th:upperbounddistillableentanglement}), the maximum state-agnostic rate corresponds to $\min\{ \Theta(S_{\min}(\psi_A)),\Theta(\log k) \}$, and saturates the one obtained in \cref{th:distillation} with a computationally efficient and state-agnostic protocol. Therefore, in the worst case, giving arbitrary high computational power to a state-agnostic and sample-efficient protocol provides no advantage in terms of the rate. Note that this bound holds irrespective of the value of the von Neumann entropy, which can be maximal up to log factors.

A similar result holds for state-agnostic dilution protocols (\cref{th:entanglementcostnogo}): here we show that any such protocol cannot perform better than the simplest one, that is quantum teleportation. Even for states having small, $o(1)$, values of the von Neumann entropy, any sample-efficient state-agnostic dilution protocol necessarily needs $\tilde{\Omega}(n)$ ebits per copy to dilute the target state.

We stress that the previous no-go results must be interpreted in a worst-case scenario, since for example, it is generally possible to achieve a much better state-agnostic distillation rate (even reaching the von Neumann entropy) on some restricted classes of states, like low-rank ones (see \cref{App:classesofstates}). 

Let us note that similar, but converse, bounds as the one obtained in \cref{lemma:infothbounds} exists, therefore, it is generally possible to distill and dilute at the von Neumann entropy rate in a sample-efficient (albeit not state-agnostic) fashion. In the following, we will also show that the same bounds apply if one restricts the computational complexity of the LOCC instead.

While the previous Theorems give limitations for what concerns state-agnostic protocols, they do not rule out the existence of a computationally efficient protocol (non state-agnostic) achieving a possibly better rate. In 
other words, \cref{th:upperbounddistillableentanglement,th:entanglementcostnogo} do not provide any upper (respectively, lower) bound on the actual computational distillable entanglement (respectively, entanglement cost). In the next sections, \cref{App:tightdistillable,app:lowerboundcost}, we will also show that the same worst-case bound applies to the computational entanglement measures defined in \cref{app:compinfomeasures}.

To conclude the section, as simple corollaries of 
our previous construction, we prove results concerning sample-complexity lower bounds for estimating (\cref{cor:lowerboundsestimatingvnentropy}) and testing (\cref{cor:lowerboundstestingvnentropy}) the von Neumann entropy. Both results hold true even 
when constraining the entropy to be in a specified range, and notably even if $S_1 = O(1)$. To our knowledge, the result on entropy estimation extends on the pre-existing literature, since previously known bounds \cite{wang_et_al:LIPIcs.ESA.2024.101} only hold when one allows the entropy to take high (close to maximal) values. Our results on entropy testing instead are new in the literature and extend on previous work done by Valiant that covers the classical (Shannon entropy) case \cite{doi:10.1137/080734066}.
Let us start with the no-go result on entanglement distillation.

\begin{theorem}[No-go on state-agnostic entanglement distillation. Formal version of \cref{th:informal1}]\label{th:upperbounddistillableentanglement}
    Given an extensive bipartition $A|B$, and large enough $n\in\mathbb{N}^+$, there exists at least one sequence of states $\ket{\psi_{A,B;n}}$, indexed by the number of qubits $n$, such that any $k$-shot with $k=O(\poly n)$ (sample-efficient) but $k = \omega(n_A^2)$ state-agnostic LOCC distillation protocol with error $\varepsilon \leq 10^{-4}$ and success probability $p\geq \sqrt{64/65}$ cannot distill at a rate larger than $\min\{S_{\min}(\psi_{A;n})+o(1),2\log k+O(1)\}$, regardless of the value of the von Neumann entropy of the state $S_{1}(\psi_{A;n})\le O(n_A/\log^3 n_A)=\tilde{O}(n_A)$.  
    \begin{proof}
    We start by noting that, without loss of generality, we can consider $S_{\min}(\psi_A)\!\leq\!O(\log n)$, since we always restrict ourselves to $k\leq\!O(\poly n)$, which means that the regime $S_{\min}(\psi_A)\!=\!\omega(\log n)$ will be superseded by the upper bound $O(\log k)$.
    
    First, let us show that any $k$-shot sample-efficient and state-agnostic protocol cannot distill at a rate higher than $S_{\min}(\psi_A)+o(1)$ necessarily, whenever $\min\{S_{\min}(\psi_A)+o(1),2\log k+O(1)\} = S_{\min}(\psi_A)+o(1)$. Let us assume, towards contradiction, that there exists such a protocol distilling at a rate $S_{\min}(\psi_A)+\Omega(1)$. This protocol could then be used as a black box distinguisher for the two ensembles of states $\mathcal{E}_{\haar,\eta,S_{\min}}$ and $\mathcal{E}_{m,\eta,S_{\min}}$, introduced in \cref{lemma:indistinguishability}, where we take $m = \log^2 n$ (i.e., $c=1$ in \cref{lemma:entanglement}). 

    To construct the distinguisher, consider $\ket{\psi_{\eta}} \sim \mathcal{E}_{\haar,\eta,S_{\min}}$. The hypothetical protocol would then distill at a rate
    \begin{equation}
    S_{\min}(\psi_{A;\eta})+\Omega(1) = S_{\min}+\Omega(1) \le E_{D}^{(k)} \leq S_1(\psi_{A;\eta}) + o(1) = S_{\min} + S_{1} + o(1).
    \end{equation}
    Here, the lower bound comes from the hypothesis towards contradiction and \cref{lemma:entanglement}, while the upper bound comes from \cref{lemma:infothbounds,lemma:entanglement}. In contrast, if $\ket{\psi_{\eta}} \sim \mathcal{E}_{m,\eta,S_{\min}}$, such a protocol cannot distill more than $E_{D}^{(k)} \le S_1(\psi_{A;\eta}) + o(1) = S_{\min}+o(1)$, again due to the same Lemmas. Then, because of the gap in the rate, such protocol would allow Alice and Bob to distinguish the two ensembles in a sample-efficient fashion. This provides a contradiction because the two ensembles are statistically indistinguishable for any $m = \omega(\log n)$, and to any distinguisher using $k = O(\poly n)$ copies as a consequence of \cref{lemma:indistinguishability}. Le us note that we can tune $S_{\min}$ so that $\min\{S_{\min}+o(1),2\log k+O(1)\} = S_{\min}+o(1)$ for $\ket{\psi_{\eta}} \sim \mathcal{E}_{\haar,\eta,S_{\min}}$. This first part therefore addresses one side of the bound, and identifies the family $\mathcal{E}_{\haar,\eta,S_{\min}}$ as the one for which the no-go holds.

Next, let us consider the other side of the bound. Let us show that for any $k = O(\poly n)$, any state-agnostic protocol cannot distill at a rate higher than $2 \log k + \Omega(1)$ from every state. To see this, assume, towards contradiction, that there exists a state-agnostic protocol distilling $2 \log k+c$ for some constant $c>0$. We will show that, upon a specific choice of $c$, this would provide a contradiction. To this aim, consider the ensemble of Haar-subsystem states $\mathcal{E}_{m}$ (see \cref{def:haarsubsystemstates}), where the Haar subsystem is on $m$ qubits. From \cref{lemma:haarsub}, we know that (for any $m\in\mathbb{N}^+$)
\[
\left\| \mathbb{E}_{\ket{\psi} \sim \mathcal{E}_m} \ketbra{\psi}^{\otimes k} - \mathbb{E}_{\ket{\psi} \sim \haar} \ketbra{\psi}^{\otimes k} \right\|_1 \leq C \frac{k^2}{2^m},
\]
for some constant $C > 0$. Now, if we choose $m$ such that $Ck^2/2^{m+1} < 2p^*-1$, due to Helstrom bound (see \cref{app:general}) it is not possible to distinguish the two ensembles with success probability $\geq p^*$. For example, we can set $m = \log (Ck^2/(2p-1))$, and fix $p^* = \sqrt{64/65}$ to match the success probability of the distillation protocol. On the other hand, the hypothetical distillation protocol would be able to distinguish $\mathcal{E}_{m}$ from the Haar random states ensemble with probability $\geq \sqrt{64/65}$.

To construct the distinguisher, note that for $\ket{\psi_{A,B}} \in \mathcal{E}_{m}$, the $k$-shot distillable entanglement cannot exceed $E_{D}^{(k)} \leq S_1(\psi_A) + o(1) \le m + o(1) = 2 \log k + \log (C/(2p^*-1)) + o(1)$ due to \cref{cor:haarsub,lemma:infothbounds}. On the other side, Haar random states have distillable entanglement $E_{D}^{(k)} \leq \Omega(n)$. Therefore, under our assumption, the protocol would distill $2 \log k+c$ allowing Alice and Bob to resolve the difference between the two ensemble with probability $\geq \sqrt{64/65}$, provided that $c > \log (C/(2p^*-1))$. This leads to a contradiction, which establishes the desired bound. Note indeed that for Haar random states we have $\min\{S_{\min}(\psi_A)+o(1),2\log k+O(1)\} = 2\log k+O(1)$. This last reasoning identifies the class of Haar random states as the one for which the second part of the no-go holds.
Finally, depending on the value of $S_{\min}(\psi_A)$ and the number of shots $k$ used, the rate obtainable by any state-agnostic protocol is upper bounded by the minimum of the two bounds derived above.
    \end{proof}
\end{theorem}

\begin{remark}
The assumption $k = \omega(n_A^2)$ is made only for simplicity, as one could derive a similar bound, apart from a multiplicative constant, for any integer $k>1$, by using a different information-theoretic bound, i.e., \cref{eq:bounddist2} in \cref{lemma:infothbounds}.
\end{remark}

Now we move to the no-go result on state-agnostic entanglement dilution.

\begin{theorem}[No-go on state-agnostic entanglement dilution. Formal version of \cref{th:informal4}]\label{th:entanglementcostnogo} 
Given an extensive bipartition $A|B$, and large enough $n\in\mathbb{N}^+$, there exists at least one sequence of states $\ket{\psi_{A,B;n}}$, indexed by the number of qubits $n$, such that any $k$-shot with $k=O(\poly n)$ (sample-efficient) but $k = \omega(n_A^2)$ state-agnostic LOCC dilution protocol with error $\varepsilon \leq 10^{-4}$ and success probability $p\geq \sqrt{64/65}$ cannot dilute at a rate lower than $\tilde{\Omega}(n_A)$, regardless of the value of the von Neumann entropy of the state, including the regime $S_{1}(\psi_{A;n}) = o(1)$. Thus, in the worst case, the optimal state-agnostic dilution protocol is quantum teleportation.
\begin{proof}
Let us consider the two families of states defined in \cref{lemma:indistinguishability}, $\mathcal{E}_{\haar,\eta,S_{\min}}$ and $\mathcal{E}_{m,\eta,S_{\min}}$, which are statistically indistinguishable for any $\eta \in [0,1]$ and $m = \omega(\log n)$ (which we choose here to be $m = \log^2 n$). By \cref{lemma:entanglement}, setting $S_{\min}=0$ and $S_1=\Theta(n_A/\log^3n_A)$, we can tune the von Neumann entropy of the two families to have an arbitrary large gap of $o(1)$ vs. $\Tilde{\Omega}(n_A)$. In particular, a state $\ket{\psi_\eta}\sim \mathcal{E}_{\haar,\eta,0}$ has von Neumann entropy $\Theta(n_A/\log^3n_A)$, while a state in $\mathcal{E}_{m,\eta,0}$ has von Neumann entropy $o(1)$. Now assume, towards contradiction, that there exists a state-agnostic and sample-efficient dilution 
protocol that is able to dilute any state with von Neumann entropy $S_{1}=o( n_A/\log^3n_A)$ using $o(n_A/\log^3n_A)$ many ebits, having local sample access to the state to be diluted. This hypothetical protocol would be able to distinguish between states belonging to either of the two families described above, thus leading to a contradiction.

We can construct a distinguisher as follows. If the dilution protocol is given a state $\ket{\psi_\eta}\sim\mathcal{E}_{m,\eta,0}$ (or polynomially many copies thereof), since $ S_{1}(\psi_{A;\eta})=o(1)$ 
(by \cref{lemma:entanglement}), only $o(n_A/\log^{3}n_A)$ of the ebits should be consumed by the hypothesis towards contradiction. On the other hand, if the protocol is given a state $\ket{\psi_\eta}\sim\mathcal{E}_{\haar,\eta,0}$, $\Omega(n_A/\log^3n_A)$ ebits should be consumed at least. Indeed, by \cref{lemma:entanglement,lemma:infothbounds}, one has that $E_{C}^{(k)}\ge S_{1}(\psi_{A;\eta}) - o(1)=\Omega(n_A/\log^3n_A)$ for states in the ensemble $\mathcal{E}_{m,\eta,0}$, therefore no $k$-shot LOCC dilution protocol can use fewer than $\Omega(n_A/\log^3n_A)$ many ebits per copy to dilute the state with probability greater than $\sqrt{64/65}$. Then, the distinguisher could count the number of ebits used by the protocol. This would then allow to tell whether a $\mathcal{E}_{m,\eta,0}$ state or $\mathcal{E}_{\haar,\eta,0}$ state was available locally, thus leading to a contradiction. 

Therefore, we have shown that for the family of states $\mathcal{E}_{m,\eta,0}$ any sample-efficient state-agnostic dilution protocol must use at least $\Omega(n_A/\log^3n_A)=\tilde{\Omega}(n_A)$ many supplied ebits, even when $S_1 = o(1)$. This shows that in the worst case, the best state-agnostic and sample-efficient dilution protocol is simply quantum teleportation, which is also computationally efficient.
\end{proof}
\end{theorem}

\begin{remark}
    A careful reader might wonder how the results of \cref{th:informal4} reconcile with the fact that states with \( S_{\min} = o(1) \) are, in fact, very close to product states, as they are well approximated by local principal components (see \cref{lem:productstructuresmin}). Indeed, since both Alice and Bob have sample access to the state to be diluted, a trivial dilution protocol could simply be for each of them to output the reduced state of a single copy, provided that \( S_{\min} \) offers a sufficiently good approximation bound relative to the threshold \( \varepsilon \).  The key insight lies in the fact that \cref{th:informal4}, as stated, holds crucially for \( k = \omega(n_A^2) \) and constant \( \varepsilon \). In this regime, the trivial dilution protocol fails because the error is linearly amplified from $\varepsilon$ to \( k\varepsilon = \omega(n_A^2) \). However, just as in the distillation case (\cref{th:informal1}), it is possible to relax the assumption \( k = \omega(n_A^2) \) and instead consider a constant \( k \).  In this regime, however, \cref{eq:boundcost} in \cref{lemma:infothbounds} is no longer applicable; instead, one must use \cref{eq:boundcost22}, which is meaningful only when the approximation error satisfies \( \varepsilon = o(n_A^{-1}) \).
 
\end{remark}

It is worth emphasizing once again the key differences between the usual treatment of entanglement cost and our setting. Without computational constraints, the scenario where Alice and Bob have perfect knowledge of the state to be diluted non-locally is equivalent to assuming that they can sample locally from the state as many times as they wish and postprocessing without computational limits. This allows them, through quantum state tomography, to obtain complete information about the state.
When we move into the realm of computational efficiency, however, we instead assume that Alice and Bob have local access to only 
$\poly(n)$ copies of the state to be diluted. It is then noteworthy that, in the worst case, having access to these copies, even with arbitrarily high computational power, is no better than applying a completely agnostic protocol without any knowledge of the state, such as quantum teleportation.

We can now show our results on entropy estimation and testing.

\begin{corollary}[Sample-complexity lower bounds for estimating the von Neumann entropy]\label{cor:lowerboundsestimatingvnentropy} Given an extensive bipartition $A|B$ and sample access to a bipartite pure state $\ket{\psi_{A,B}}$, then $\Omega(2^{n/{(2+c)}})$ for $c>0$ and $c=O(1)$ many samples are needed to estimate the von Neumann entropy $S_{1}(\psi_A)$ up to additive error $\varepsilon<\frac{c}{2+c}$, for any fixed value of $S_{1}(\psi_A)$, even if $S_{1}(\psi_A)=O(1)$.
\begin{proof}
    For the proof, we make use of the classes of states $\mathcal{E}_{\haar,\eta,m}$ and $\mathcal{E}_{m,\eta,S_{\min}}$ introduced in \cref{lemma:indistinguishability}. Specifically, we set $m=n/(1+c/2)$ for some $c>0$, $S_{\min}=0$ and $\eta=1/n$. The states within the two classes then satisfy the following relations for the von Neumann entropy
    \be
    S_{1}(\psi_{A;\eta})&=1+o(1),\quad \ket{\psi_{\eta}}\sim\mathcal{E}_{\haar,\eta,S_{\min}},\\
    S_{1}(\psi_{A;\eta})&=\frac{1}{1+c/2}+o(1),\quad \ket{\psi_{\eta}}\sim\mathcal{E}_{m,\eta,S_{\min}}.
    \ee
    Notice that the von Neumann entropies are $O(1)$ as promised and they differ by a factor $\frac{c}{2+c}$. Therefore, a measurement of the von Neumann entropy up to additive error $\varepsilon<\frac{c}{2+c}$ starting from $k$ copies would distinguish the two classes of states that, thanks to \cref{lemma:indistinguishability} cannot be distinguished using less than $k=\Omega(\sqrt{2^m})=\Omega(2^{n/(2+c)})$ samples. A similar argument follows in any range of the entropy. Hence the result is proven.
\end{proof}
\end{corollary}

We can now also generalize to the quantum case a previous classical result (holding for the Shannon entropy) due to Valiant \cite{doi:10.1137/080734066}.

\begin{corollary}[Sample-complexity lower bounds for testing the von Neumann entropy]\label{cor:lowerboundstestingvnentropy}
Given an extensive bipartition $A|B$, sample access to a bipartite pure state $\ket{\psi_{A,B}}$ and two arbitrary real numbers $0 < \alpha < \beta < n$, then $\Omega(2^{\frac{\alpha}{2\beta}n})$ samples are needed to test whether $S_{1}(\psi_A)\leq\alpha$ or $S_{1}(\psi_A)\geq\beta$.
\end{corollary}
\begin{proof}
For the proof, we make use of the classes of states $\mathcal{E}_{\haar,\eta,S_{\min}}$ and $\mathcal{E}_{m,\eta,S_{\min}}$ introduced in \cref{lemma:indistinguishability}. Specifically, we set $m=n\alpha/\beta < n$, $S_{\min}=0$ and $\eta=\alpha/n$. The states within the two classes satisfy the following relations for the von Neumann entropy
    \be
    S_{1}(\psi_{A;\eta})&=\alpha+o(1),\quad \ket{\psi_{\eta}}\sim\mathcal{E}_{\haar,\eta,S_{\min}},\\
    S_{1}(\psi_{A;\eta})&=\beta+o(1),\quad \ket{\psi_{\eta}}\sim\mathcal{E}_{m,\eta,S_{\min}}.
    \ee
    Therefore, any measurement able to detect whether $S_{1}(\psi_A)\leq\alpha$ or $S_{1}(\psi_A)\geq\beta$ would be able to distinguish between the two classes of states, which we know are indistinguishable by \cref{lemma:indistinguishability} using less than $k=\Omega(\sqrt{2^m})=\Omega(2^{\frac{\alpha}{2\beta}n})$ samples. This proves the claim.
\end{proof}

\subsection{Tight upper bound on computational distillable entanglement}\label{App:tightdistillable}
In the previous section, we established the optimal rates of entanglement distillation and dilution protocols when they are both state-agnostic and sample-efficient. Having no prior information on the state clearly imposes some stringent limitations on the protocols performance when one is limited also by sample-efficiency. Specifically, any state-agnostic protocol operating in the sample-efficient (not necessarily computationally efficient) regime can distill at most at a rate $ \min\{O(S_{\min}), O(\log k)\}$ (in the worst-case scenario). 

In this section, we show that these limitations are not solely due to state-agnosticity. Instead, they ultimately stem from the restriction on the computational efficiency of the LOCC. In \cref{appth:upperbounddistillableent}, we show that when one allows only for computationally efficient LOCC operations, there are families of states from which any distillation protocol cannot obtain a rate higher than $ \min\{O(S_{\min}), O(\log k)\}$, even when their von Neumann entropy is maximal up to log factors. Together with the achievability results in \cref{th:distillation}, this establishes the existence of states with computational distillable entanglement $\hat{E}_{D}^{(k)} = \min\{\Theta(S_{\min}), \Theta(\log k)\}$.

This result has two significant implications. First, it reveals a substantial (possibly maximal) gap in the rate between sample-efficient LOCC protocols and computationally efficient ones, indeed, by a similar (but converse) argument as of \cref{lemma:infothbounds}, one can distill approximately at the von Neumann entropy rate even for $k=O(\poly n)$ copies, although relying on possibly computationally inefficient LOCC.
Secondly, this result establishes that the state-agnostic distillation protocol discussed in \cref{section:protocol} is optimal in the computationally efficient regime: there are states for which no other computational-efficient protocol can achieve an higher rate of distillation.  

\begin{theorem}[Upper bound on computational distillable entanglement. Formal version of \cref{th:upperbounddistillableent}]\label{appth:upperbounddistillableent} For any extensive bipartition $A|B$, any $n\in\mathbb{N}^+$ large enough, and any $k = O(\poly n)$ such that $k = \omega(n_A^2)$, there exists at least one sequence of $n$-qubit states $\ket*{\psi_{A,B;n}^*}$ (indexed by $n$) with computational distillable entanglement $\hat{E}_{D}^{(k)}(\psi^*_{A,B;n})\le \min\{ S_{\min}(\psi^*_{A;n})+o(1), \log k + o(1) \}$, and asymptotic distillable entanglement $E_{D}(\psi^*_{A,B;n})=\tilde{\Omega}(n)$. In other words, any computationally efficient LOCC protocol starting from $k$ copies of $\ket{\psi_{A,B;n}^*}$ cannot distill at a rate exceeding $\min\{ S_{\min}(\psi^*_{A;n})+o(1), \log k + o(1) \}$, for any success probability $p\ge\sqrt{64/65}$ and error $\varepsilon \le 10^{-4}$.
\begin{proof} 
Let us first show the $S_{\min}(\psi_{A})+o(1)$ bound. As we will discuss later, the proof of the $\log k + o(1)$ bound follows similarly. Let us assume, towards contradiction, that every state $\ket{\psi_{A,B}}$ has computational distillable entanglement $\hat{E}_{D}^{(k)}(\psi_{A,B})\ge S_{\min}(\psi_A)+\Omega(1)$. This implies that for any state $\ket{\psi_{A,B}}$, there exists a computationally efficient LOCC $\Gamma_{\psi}$ that, from $k=\poly n$ copies of $\ket{\psi_{A,B}}$ distills approximate ebits at the given rate with trace distance error $\varepsilon \leq 10^{-4}$ and success probability $p \ge\sqrt{64/65}$. 
By assumption, the output of the LOCC $\Gamma_{\psi}$ is of the form
    \be
    \Gamma_{\psi}(\psi_{A,B}^{\otimes k})=p\ketbra{0}_{X}\otimes\omega_{A,B}+(1-p)\ketbra{1}_{X}\otimes\sigma_{A,B},\label{eq:out}
    \ee 
with $\frac{1}{2}\|\omega_{A,B}-\phi_{A,B}^{+\otimes kR}\|_1\le \varepsilon$. 
Then, there exists another computationally efficient LOCC $\tilde{\Gamma}_{\psi}\coloneqq\tr_{X}\circ\Gamma_{\psi}$ (where the trace is over the classical register) that, with probability one, distills approximate ebits with trace distance error $\le \varepsilon+(1-p)$. This is a trivial consequence of \cref{eq:out} since
\be
\frac{1}{2}\|\tilde{\Gamma}_{\Psi}(\Psi_{A,B}^{\otimes k})-\phi_{A,B}^{+\otimes kR}\|_1&=\frac{1}{2}\|p\omega_{A,B}+(1-p)\sigma_{A,B}-\phi_{A,B}^{+\otimes kR}\|_1 \\
&\le \frac{1}{2}\|(p-1)\omega_{A,B} + (1-p)\sigma_{A,B} +\omega_{A,B}-\phi_{A,B}^{+\otimes kR} \|_1 \\
&\le \varepsilon+(1-p).
\ee
In particular, by the initial assumption, such an LOCC must exist also for states sampled from one of the ensembles introduced in \cref{lemma:indistinguishability}, that is if we consider $\Psi_{A,B}\sim\mathcal{E}_{\haar,\eta,S_{\min}}$.

Now, consider an arbitrary, computationally efficient LOCC $\Lambda$ and define, for every state $\Psi_{A,B}\sim\mathcal{E}_{\haar,\eta,S_{\min}}$, the function $F_{\Lambda}(\Psi_{A,B}^{\otimes k})\coloneqq \tr(\phi_{A,B}^{+\otimes kR}\Lambda(\Psi_{A,B}^{\otimes k}))$, capturing the fidelity of the output with $kR$ ebits. For two states $\Psi_{A,B},\Psi_{A,B}'\sim\mathcal{E}_{\haar,\eta,S_{\min}}$ it holds that
\be
|\tr( \phi_{A,B}^{+\otimes kR}\Lambda(\Psi_{A,B}^{\otimes k}))-\tr(\phi_{A,B}^{+\otimes kR}\Lambda(\Psi_{A,B}^{'\otimes k}))| &\overset{\text{(i)}}\le \frac{k}{2}\|\Psi_{A,B}-\Psi_{A,B}'\|_1\\
&\overset{\text{(ii)}}\le k\|\ket{\Psi_{A,B}}-\ket{\Psi_{A,B}'}\|_2\\
&\overset{\text{(iii)}}\le k \sqrt{\eta}\|\ket{\psi_{A,B}}-\ket{\psi_{A,B}'}\|_2,
\ee
where $\ket{\psi_{A,B}},\ket*{\psi_{A,B}'}\sim\haar(n)$, cfr. \cref{lemma:indistinguishability}. Here, in (i) we have used  the variational definition of trace distance followed by its monotonicity under linear maps and additivity under tensor products; in (ii) we have used  the known fact $\|\Psi_{A,B}-\Psi_{A,B}'\|_1 \leq 2\|\ket{\Psi_{A,B}}-\ket*{\Psi_{A,B}'}\|_2$ and in (iii) we explicitly computed the $2$-norm from the definition of the ensemble in \cref{lemma:indistinguishability}. Therefore, the function $F_{\Lambda}(\cdot)$ has Lipschitz constant $k\sqrt{\eta}$ on Haar random states and as such, we can apply Levy's Lemma (see, for example, Ref.\ \cite{watrous}) and write
\be
\Pr_{\Psi\sim\mathcal{E}_{\haar,\eta,S_{\min}}}[|F_{\Lambda}(\Psi_{A,B}^{\otimes k})-\mathbb{E}(F_{\Lambda}(\Psi_{A,B}^{\otimes k}))|\ge\delta]\le 2^{-\Omega(d\delta^2\eta^{-1}k^{-2})}.
\ee
Using the union bound, we can upper bound the failure probability that, for all computationally efficient LOCC $\Lambda$, the fidelity is close to the average value over the ensemble $\mathcal{E}_{\haar,\eta,S_{\min}}$. Since computationally efficient LOCC are described by at most $O(\poly n)$ many bits, there are ``only'' $2^{O(\poly n)}$ such LOCC. Choosing $\delta=d^{-1/4}$ in the bound above we can then write
\be
\Pr_{\Psi\sim\mathcal{E}_{\haar,\eta,S_{\min}}}[|F_{\Lambda}(\Psi_{A,B}^{\otimes k})-\mathbb{E}(F_{\Lambda}(\Psi_{A,B}^{\otimes k}))|\le d^{-1/4}\,\,,\,\,\forall \, \text{efficient} \, \Lambda]\ge1- 2^{O(\poly(n))}2^{-\Omega(\sqrt{d}k^{-2})}.
\ee
Therefore, for any $k=O(\poly n)$, the exists $n^{*}\in\mathbb{N}^+$ such that for every integer $n\ge n^{*}$ there exists a state $\ket*{\Psi_{A,B}^{*}}\in\mathcal{E}_{\haar,\eta,S_{\min}}$ for which, for every efficient LOCC $\Lambda$, it holds that
\be
F_{\Lambda}((\Psi_{A,B}^{*})^{\otimes k})\le\mathbb{E}_{\Psi\sim\mathcal{E}_{\haar,\eta,S_{\min}}}[F_{\Lambda}(\Psi_{A,B}^{\otimes k})]+d^{-1/4}. 
\ee
Now, by joint application of the linearity of the LOCC $\Lambda$, the monotonicity of the trace distance under linear maps and its variational definition, as well as \cref{lemma:indistinguishability}, notice that the following bounds hold:
\be
&|\mathbb{E}_{\Psi\sim\mathcal{E}_{\haar,\eta,S_{\min}}}[F_{\Lambda}(\Psi_{A,B}^{\otimes k})]-\mathbb{E}_{\Psi\sim\mathcal{E}_{m,\eta,S_{\min}}}[F_{\Lambda}(\Psi_{A,B}^{\otimes k})]| \\
&\leq \frac{1}{2} \|\mathbb{E}_{\Psi\sim\mathcal{E}_{\haar,\eta,S_{\min}}} \Lambda(\Psi_{A,B}^{\otimes k}) - \mathbb{E}_{\Psi\sim\mathcal{E}_{m,\eta,S_{\min}}} \Lambda(\Psi_{A,B}^{\otimes k}) \|_1 \\
&\leq \frac{1}{2} \|\mathbb{E}_{\Psi\sim\mathcal{E}_{\haar,\eta,S_{\min}}} [\Psi_{A,B}^{\otimes k}] - \mathbb{E}_{\Psi\sim\mathcal{E}_{m,\eta,S_{\min}}} [\Psi_{A,B}^{\otimes k}] \|_1 \leq \frac{Ck^2}{2^{m+1}},
\ee
therefore, we can write
\be
\mathbb{E}_{\Psi\sim\mathcal{E}_{\haar,\eta,S_{\min}}}[F_{\Lambda}(\Psi_{A,B}^{\otimes k})]&\le \mathbb{E}_{\Psi\sim\mathcal{E}_{m,\eta,S_{\min}}}[F_{\Lambda}(\Psi_{A,B}^{\otimes k})]+\frac{Ck^2}{2^{m+1}} \\
&\le \max_{\Psi\in\mathcal{E}_{m,\eta,S_{\min}} }F_{\Lambda}(\Psi_{A,B}^{\otimes k})+\frac{Ck^2}{2^{m+1}}.
\ee
Hence, choosing $m=\omega(\log n)$, and for sufficiently large $n$, we have that there exists a state $\Psi_{A,B}^* \in\mathcal{E}_{\haar,\eta,S_{\min}}$ such that, for any computationally efficient LOCC $\Lambda$, the following upper bound holds
\be
F_{\Lambda}((\Psi_{A,B}^{*})^{\otimes k})\le \max_{\Psi\in\mathcal{E}_{m,\eta,S_{\min}} }F_{\Lambda}(\Psi_{A,B}^{\otimes k})+\operatorname{negl}(n)\label{eq12312321}.
\ee
Now, by our hypothesis, the state $\ket*{\Psi_{A,B}^{*}}$ admits a computationally efficient LOCC $\Gamma^* \coloneqq \Tilde{\Gamma}_{\Psi^*}$ that distills from $k$ copies of $\ket*{\Psi_{A,B}^{*}}$ with probability $1$, rate $R \ge S_{\min}(\Psi_A)+\Omega(1) = S_{\min} + \Omega(1)$ (see \cref{lemma:entanglement}) and trace distance error $\le\varepsilon+(1-p)$. Since \cref{eq12312321} holds for every efficient LOCC $\Lambda$, we now set $\Lambda = \Gamma^*$, and applying Fuchs-van de Graaf inequalities two times (see \cref{app:general}), we can write
\be
\min_{\Psi\in\mathcal{E}_{m,\eta,S_{\min}} } \frac{1}{2}\left\| \Gamma^*(\Psi_{A,B}^{\otimes k})- \phi_{A,B}^{+\otimes kR}\right\|_1 &\le \sqrt{1-F_{\Gamma^*}((\Psi_{A,B}^*)^{\otimes k})+\operatorname{negl}(n)}\\
&\le \sqrt{1-(p-\varepsilon)^2+\operatorname{negl}(n)}\\
&< \frac{1}{8}.
\ee
Where we have used  the choices of the bounds on $p$ and $\varepsilon$ as in the hypothesis and in \cref{def:distillableentanglement}. Hence, the LOCC $\Gamma^*$ would also distill at a rate $R\ge S_{\min}+\Omega(1)$ from some state in the ensemble $\mathcal{E}_{m,\eta,S_{\min}}$ with trace distance error $< 1/8$ and unit probability. Now, by \cref{lemma:entanglement}, any state in $\mathcal{E}_{m,\eta,S_{\min}}$ has von Neumann entropy equal to $S_{\min}+o(1)$, and, by \cref{lemma:infothbounds} , we know that any LOCC protocol cannot distill with an higher rate than that. This provides a contradiction and concludes the proof. 

The proof of the $\log k + o(1)$ bound proceeds in the same way, but more simply uses Haar vs Haar-subsystem states, for which all the previous bounds hold analogously provided that $m = \omega(\log n)$. In this case, their von Neumann entropy is at most $m$. Setting $m = \log^{1+c}k$ with arbitrary $c>0$ gives the desired bound.
\end{proof}

\end{theorem}

\subsection{Tight lower bound on computational entanglement cost}\label{app:lowerboundcost}
In this section, we prove a tight lower bound on the computational entanglement cost. We use similar techniques as the one introduced in \cref{App:tightdistillable}. Our results show that in the worst case, quantum teleportation achieves the optimal ebits consumption, and therefore no saving in terms of the consumed entanglement is possible, irrespective of the (arbitrarily small) value of the von Neumann entropy.

\begin{theorem}[Tight lower bound on computational entanglement cost. Formal version of \cref{th:lowerboundentcost}]\label{appth:lowerboundentcost} 

For any extensive bipartition $A|B$, any $n\in\mathbb{N}^+$ large enough, and any $k = O(\poly n)$ such that $k = \omega(n_A^2)$, there exists at least one sequence of $n$-qubit states $\ket*{\psi_{A,B;n}^*}$ (indexed by $n$) with computational entanglement cost $\hat{E}^{(k)}_{C}(\psi_{A,B;n}^*)\ge \Omega(n_A/\log^3n_A)$, and asymptotic entanglement cost $E_{C}(\psi^*_{A,B;n})=o(1)$. In other words, any computationally efficient LOCC protocol cannot dilute $k$ copies of $\ket*{\psi_{A,B;n}^*}$ at a rate lower than $\Tilde{\Omega}(n_A)$, for any success probability $p\ge\sqrt{64/65}$ and error $\varepsilon \le 10^{-4}$.
\begin{proof}
The proof proceeds similarly as the one for the distillable entanglement. However, in this case there are a few nuances to take care of. Namely, at the core of the previous proof was the application of Levy's lemma on the states of the ensemble $\mathcal{E}_{\haar,\eta,S_{\min}}$. However, if one tries to proceed similarly in this proof, they would be faced with applying Levy's lemma to the ensemble $\mathcal{E}_{m,\eta,S_{\min}}$ instead, which only has a Haar random component of effectively $m \in O(\text{polylog }n)$ qubits. This leads to a bound that is too weak to proceed. Indeed, to be able to select a certain state $\Psi^*_{A,B}$ whose fidelity lower bounds the average fidelity of the ensemble, one also needs to incorporate the selection of the random permutation $\pi$ in the probability of this occurring. But unfortunately, the ensemble $\mathcal{E}_{m,\eta,S_{\min}}$ makes use of $2^n!$ permutations (or effectively ${2^n \choose 2^m} \geq 2^{n2^m}/(2^{m2^m})$ when one takes into account the symmetries of the ensemble), which is too large compared to the bound given by Levy's lemma. The solution we adopt here is to consider instead the ensemble $\mathcal{E}^{(N)}_{m,\eta,S_{\min}}$ constructed with the restricted Haar-subsystem states defined in \cref{def:restrictedhaarsubsystemstates} for a carefully chosen number of permutations $N = 2^{\omega(\poly(n))}$ as to combine it properly with the result of Levy's lemma.

Let us assume that for any state $\ket{\psi_{A,B}}$ and number of copies $k = O(\poly n)$ there exists a computationally efficient LOCC $\Gamma_{\psi}$ diluting $\ket{\psi_{A,B}}$ starting from $kR$ ebits, with error $\varepsilon \leq 10^{-4}$ and probability $p\geq \sqrt{64/65}$. This, in particular, implies that such an LOCC exists for every state $\ket{\Psi_{A,B}}\in\mathcal{E}^{(N)}_{m,\eta,S_{\min}}$ with $S_{\min}=0$. Note that such states have von Neumann entropy $S_1(\Psi_A) = o(1)$ (cfr. \cref{lemma:entanglement}). Let us now assume that for any state there exists a computationally efficient LOCC such that one has $R \leq \Tilde{o}(n)$ whenever $S_1 =  \Tilde{o}(n)$. We will show that this leads to a contradiction, as there must be at least one state in $\mathcal{E}^{(N)}_{m,\eta,S_{\min}}$ such that one must have $R \geq \Tilde{\Omega}(n)$ for any efficient LOCC dilution protocol.

As in the proof above, consider the output of the $k$-shot LOCC dilution protocol $\Gamma_{\Psi}$ on any state $\ket{\Psi_{A,B}}\in\mathcal{E}^{(N)}_{m,\eta,0}$. This will be of the form
\be
\Gamma_{\Psi}(\Psi_{A,A'}^{\otimes l}\otimes\phi^{+\otimes kR}_{A , B})=p\ketbra{0}_{X}\otimes \tilde{\Psi}_{A,B}+(1-p)\ketbra{1}_{X}\sigma_{A,B},
\ee
where $\frac{1}{2}\|\tilde{\Psi}_{A,B}-\ketbra{\Psi_{A,B}}^{\otimes k}\|_1\le \varepsilon$. Proceeding as above, this implies the existence of an LOCC protocol $\Tilde{\Gamma}_{\Psi}$ achieving the same rate, but with unit probability and error $\leq (1-p)+\varepsilon$.
For an arbitrary computationally efficient LOCC dilution protocol $\Lambda$, let us now define the function
\be
F_{\Lambda}(\Psi_{A,B}^{\otimes (k+l)})\coloneqq \tr(\Psi_{A,B}^{\otimes k}\Lambda(\Psi_{A,A'}^{\otimes l}\otimes \phi_{A , B}^{+\otimes kR})).
\ee
We first note that the following continuity bound holds for any two (possibly mixed) quantum states $\Psi_{A,B},\Phi_{A,B}$:
\be
&\abs{F_{\Lambda}(\Psi_{A,B}^{\otimes (k+l)}) - F_{\Lambda}(\Phi_{A,B}^{\otimes (k+l)})}\\
&=\abs{\tr(\Psi_{A,B}^{\otimes k}\Lambda(\Psi_{A,A'}^{\otimes l}\otimes \phi_{A , B}^{+\otimes kR})) - \tr(\Phi_{A,B}^{+\otimes k}\Lambda(\Phi_{A,A'}^{\otimes l}\otimes \phi_{A , B}^{+\otimes kR}))} \\
&\le \frac{1}{2} \left\| \Psi_{A,B}^{\otimes k} \otimes \Lambda(\Psi_{A,A'}^{\otimes l}\otimes \phi_{A , B}^{+\otimes kR})  - \Phi_{A,B}^{\otimes k} \otimes \Lambda(\Phi_{A,A'}^{\otimes l}\otimes \phi_{A , B}^{+\otimes kR})  \right\|_1 \\
&\le \frac{1}{2} \left\| \Psi_{A,B}^{\otimes k} \otimes \Psi_{A,A'}^{\otimes l} - \Phi_{A,B}^{\otimes k} \otimes \Phi_{A,A'}^{\otimes l} \right\|_1\\
&= \frac{1}{2} \left\| \Psi_{A,B}^{\otimes (k+l)} - \Phi_{A,B}^{\otimes (k+l)} \right\|_1,\label{eq:Lipschitz}
\ee
where we have expressed the inner products as the result of a SWAP test and used  the variational definition of the trace distance followed by its monotonicity under linear maps. Let us now adopt the following notation for states in $\mathcal{E}^{(N)}_{m,\eta,0}$:
\be
\ket{\Psi_{A,B;\eta,\pi,\phi_m}} = \sqrt{1-\eta}\ket{0}_{A}\ket{0}_{B}\ket{0}^{\otimes n} +\sqrt{\eta}\ket{1}_{A}\ket{1}_{B}\ket{\psi_{A,B;\pi,\phi_m}},
\ee
where $\ket{\psi_{A,B;\pi,\phi_m}} = U_{\pi} \ket{0}^{\otimes n-m}\ket{\phi_m} \in \mathcal{E}^{(N)}_{m}$, with $U_{\pi}$ implementing a permutation $\pi \in P_N$ allowed by the ensemble, and $\ket{\phi_m}\sim\haar(m)$. An application of \cref{eq:Lipschitz} gives us the Lipschitz constant of $F_{\Lambda}(\cdot)$ with respect to the Haar component $\ket{\phi_m}$ of the states $\ket{\Psi_{A,B;\eta,\pi,\phi_m}}$. Indeed, for any $\Psi_{A,B;\eta,\pi,\phi_m},\Psi'_{A,B;\eta,\pi,\phi'_m} \sim \mathcal{E}^{(N)}_{m,\eta,0}$ (for a fixed permutation $\pi$) we have:
\be
\abs{F_{\Lambda}(\Psi_{A,B;\eta,\pi,\phi_m}^{\otimes k+l}) - F_{\Lambda}({\Psi'}_{A,B;\eta,\pi,\phi'_m}^{\otimes k+l})} &\leq \frac{1}{2} \norm{\Psi_{A,B;\eta,\pi,\phi_m}^{\otimes k+l} -{\Psi'}_{A,B;\eta,\pi,\phi'_m}^{\otimes k+l}}_1\\
&\leq (k+l)\sqrt{\eta} \norm{\ket{\phi_m} - \ket{\phi'_m}}_2.
\ee
We can now make use of the Lipschitz continuity of $F_{\Lambda}(\cdot)$ and Levy's lemma to get, for every $\pi \in P_N$:
\be
\Pr_{\phi_m}[|F_{\Lambda}(\Psi_{A,B;\eta,\pi,\phi_m}^{\otimes k+l})-\mathbb{E}_{\phi_m}[F_{\Lambda}(\Psi_{A,B;\eta,\pi,\phi_m}^{\otimes k+l})]|\ge\delta]\le 2^{-\Omega(2^m\delta^2\eta^{-1}(k+l)^{-2})}.
\ee
We then set $\delta = 1/2^{m/4}$, apply the union bound over all efficient LOCCs ($2^{O(\poly n)}$ many), and over all $\pi \in P_N$, for $N=2^{2^{m/4}}$, to get:
\be
\Pr_{\phi_m}[|F_{\Lambda}(\Psi_{A,B;\eta,\pi,\phi_m}^{\otimes k+l})-\mathbb{E}_{\phi_m}[F_{\Lambda}(\Psi_{A,B;\eta,\pi,\phi_m}^{\otimes k+l})]| \le 1/2^{m/4},\ \forall\text{ efficient }\Lambda,\ \forall{\pi \in P_N}] \ge 1 - 2^{O(\poly(n))}2^{2^{m/4}}2^{-\Omega(2^{m/2}/\poly(n))},
\ee
which, for sufficiently large $n$, is arbitrarily close to $1$, given that $m=\omega(\log n)$.

With this, we get that there exists at least one state $\phi^*_m$ such that, for all efficient LOCCs $\Lambda$, and for all permutations $\pi \in P_N$, we have:
\be
F_{\Lambda}(\Psi_{A,B;\eta,\pi,\phi_m^*}^{\otimes k+l}) \leq \mathbb{E}_{\phi_m}[F_{\Lambda}(\Psi_{A,B;\eta,\pi,\phi_m}^{\otimes k+l})] + 1/2^{m/4} , \quad \forall \pi \in P_N.
\ee
Now consider the permutation $\pi^* \in P_N$ that minimizes $\mathbb{E}_{\phi_m}[F_{\Lambda}(\Psi_{A,B;\eta,\pi,\phi_m}^{\otimes k+l})]$:
\be
\pi^* = \text{argmin}_{\pi \in P_N} \mathbb{E}_{\phi_m}[F_{\Lambda}(\Psi_{A,B;\eta,\pi,\phi_m}^{\otimes k+l})].
\ee
We have:
\be
F_{\Lambda}(\Psi_{A,B;\eta,\pi^*,\phi^*_m}^{\otimes k+l}) \leq \min_{\pi \in P_N} \mathbb{E}_{\phi_m}[F_{\Lambda}(\Psi_{A,B;\eta,\pi,\phi_m}^{\otimes k+l})] + 1/2^{m/4} \leq \mathbb{E}_{\pi \in P_N,\phi_m}[F_{\Lambda}(\Psi_{A,B;\eta,\pi,\phi_m}^{\otimes k+l})] + 1/2^{m/4}.
\ee
As a second application of \cref{eq:Lipschitz} to the states $\mathbb{E}_{\Psi \sim\mathcal{E}^{(N)}_{m,\eta,0}}[\Psi_{A,B;\eta,\pi,\phi_m}^{\otimes k+l}]$ and $ \mathbb{E}_{\Psi \sim\mathcal{E}_{m,\eta,0}}[\Psi_{A,B;\eta,\pi,\phi_m}^{\otimes k+l}]$, we get:
\be
&\abs{\mathbb{E}_{\Psi \sim\mathcal{E}^{(N)}_{m,\eta,0}}[F_{\Lambda}(\Psi_{A,B;\eta,\pi,\phi_m}^{\otimes k+l})] - \mathbb{E}_{\Psi \sim\mathcal{E}_{m,\eta,0}}[F_{\Lambda}(\Psi_{A,B;\eta,\pi,\phi_m}^{\otimes k+l})]}\\
&= \abs{F_{\Lambda}(\mathbb{E}_{\Psi\sim\mathcal{E}^{(N)}_{m,\eta,0}}[\Psi_{A,B;\eta,\pi,\phi_m}^{\otimes k+l}]) - F_{\Lambda}(\mathbb{E}_{\Psi\sim\mathcal{E}_{m,\eta,0}}[\Psi_{A,B;\eta,\pi,\phi_m}^{\otimes k+l}])}\\
&\leq \frac{1}{2}\norm{\mathbb{E}_{\Psi\sim\mathcal{E}^{(N)}_{m,\eta,0}}[\Psi_{A,B;\eta,\pi,\phi_m}^{\otimes k+l}] - \mathbb{E}_{\Psi\sim\mathcal{E}_{m,\eta,0}}[\Psi_{A,B;\eta,\pi,\phi_m}^{\otimes k+l}]}_1\\
&\leq \text{negl}(n),
\ee
where the last inequality comes from \cref{lem:error-restricted} by noting that $N = 2^{\omega(\poly n)}$ and $k \in O(\poly n)$.

Similarly, we can replace the average over states in $\mathcal{E}_{m,\eta,0}$ by an average over states in $\mathcal{E}_{\haar,\eta,0}$ up to a negligible error, which in combination with the previous equations, gives us
\be\label{eq:fidelity-bound}
F_{\Lambda}(\Psi_{A,B;\eta,\pi^*,\phi_m^*}^{\otimes k+l}) \leq \mathbb{E}_{\Psi\sim\mathcal{E}_{\haar,\eta,0}} [F_{\Lambda}(\Psi_{A,B}^{\otimes k+l})] + \text{negl}(n)
\ee
Now, since the probability such that $S_{1}(\Psi_A)=\Omega(S_1)$ is overwhelmingly high (cfr. \cref{lemma:entanglement}), we can consider the modified ensemble
\be
\mathcal{F}_{\haar,\eta,0} \coloneqq\{\ket{\Psi_{A,B}}\in\mathcal{E}_{\haar,\eta,0}\,:\, S_{1}(\Psi_A)=\Omega(S_1)\}, 
\ee
and one has that
\be
\frac{|\mathcal{E}_{\haar,\eta,0}|-|\mathcal{F}_{\haar,\eta,0}|}{|\mathcal{E}_{\haar,\eta,0}|} = \operatorname{negl}(n).
\ee
Since the fidelity is bounded by one, one has that
\be
\abs{\mathbb{E}_{\Psi\in\mathcal{E}_{\haar,\eta,0}}F_{\Lambda}(\Psi_{A,B}^{\otimes k+l}) - \mathbb{E}_{\Psi\in\mathcal{F}_{\haar,\eta,0}}F_{\Lambda}(\Psi_{A,B}^{\otimes k+l})} \leq \operatorname{negl}(n).
\ee
Therefore, the bound in \cref{eq:fidelity-bound} holds true exchanging the average over $\mathcal{F}_{\haar,\eta,0}$. Passing to the maximum we then get
\be
F_{\Lambda}(\Psi_{A,B;\eta,\pi^*,\phi_m^*}^{\otimes k+l}) \leq \max_{\Psi\in\mathcal{F}_{\haar,\eta,0}} F_{\Lambda}(\Psi_{A,B}^{\otimes k+l}) + \text{negl}(n).
\ee
Let us now set the efficient LOCC as $ \Lambda = \Gamma^* \coloneqq \Tilde{\Gamma}_{\Psi^*} $. Due to the initial hypothesis, this LOCC achieves the rate $R \leq \Tilde{o}(n)$ with the given error parameters. However, due to the last bound, the same LOCC must achieve the same rate, with a modified error parameter (see the final steps of \cref{th:upperbounddistillableentanglement}) for some state in the ensemble $\mathcal{F}_{\haar,\eta,0}$. Due to \cref{lemma:entanglement}, we know that such states have $S_{1}(\Psi_A)=\Omega(S_{1})$ for every choice of $S_{1}=O(n_A/\log^3n_A)$, then due to \cref{lemma:infothbounds} we must have $R \geq \Tilde{\Omega}(n)$. This provides 
a contradiction and concludes the proof.
\end{proof}
\end{theorem}

\begin{remark}
Let us note that, while the $\mathcal{E}_{\haar,\eta,S_{\min}}$ ensemble provides the no-go result on distillation (see \cref{appth:upperbounddistillableent}), here the states in $\mathcal{E}_{m,\eta,S_{\min}}^{(N)}$ appear instead. Of fundamental importance in our proof is the fact that this last ensemble inherits the Haar-random properties of his $m$-qubit subsystem, and therefore allows us to leverage again strong concentration inequalities like those provided by Levy's lemma. Furthermore, if the total number of permutations $N$ is chosen sufficiently small, we can prove doubly exponential concentration irrespective of the fixed permutation. This last aspect allows us to preserve the concentration bound obtained by Levy's lemma. This cannot be achieved by simply considering the states in $\mathcal{E}_{\haar,\eta,S_{\min}}$. To our knowledge, this reasoning does not apply directly to other pseudoentangled state constructions, like, for example, the one considered in Ref.\ \cite{aaronson2023quantumpseudoentanglement}.  
\end{remark}

\section{State-aware entanglement theory}\label{app:stateaware}
\subsection{Limitations when having a classical description}
In the previous sections, we explored the tasks of computationally efficient and state-agnostic entanglement manipulation by exhibiting an optimal state-agnostic protocol for both distillation and dilution tasks. Moreover, we have shown that if we limit the LOCC protocol to be computationally efficient, then there exists at least one pure state from which one cannot distill or dilute more than the rate predicted by the state-agnostic case. This result showcases the optimality of the agnostic protocol discussed in \cref{section:protocol}, and of the quantum teleportation protocol, which then saturate, in the worst-case scenario, the computational distillable entanglement and the computational entanglement cost, respectively. This last no-go result should not come entirely as a surprise: the states for which \cref{appth:upperbounddistillableent,appth:lowerboundentcost} hold are computationally hard to prepare, and the complexity of the Schmidt basis imposes limits on how well a computationally efficient LOCC protocol can distill or dilute the state.
In this section, we aim at addressing the last open question left behind by our result: what if Alice and Bob have at their disposal a classical description of the state to distill the entanglement from? Can entanglement distillation be enhanced by the perfect, classical knowledge of the state? And what about entanglement dilution? This setting corresponds to the \textit{state-aware} scenario. 

To be more precise, let us first define what constitutes a \textit{state-aware LOCC protocol}.
Such a protocol works on any class of bipartite pure states admitting an efficient classical description, and starting from the classical knowledge of the state computes a set of instructions which are then implemented as (local) quantum gates and measurements by the two parties. Here we show that if the compilation of the instructions is required to be computationally efficient, then no state-aware LOCC protocol exists that can distill at a rate higher than $S_{\min}$ from every quantum state admitting a classical description. An analogue result holds for the entanglement cost: no state-aware LOCC exists that can dilute every $n$-qubit efficiently preparable state at a rate lower than $\Omega(n^{1-\delta})$, for any $\delta>0$. Both these results hold irrespective of the actual value of the von Neumann entropy, which again loses operational meaning in this computationally restricted scenario. This result establishes another time the optimality of the state-agnostic case. 

To prove the result, we need some key lemmas, which are based on the LWE assumption (see \cref{app:lwe}). The first is a consequence of Theorem 3.13 in Ref.\ \cite{bouland2023publickeypseudoentanglementhardnesslearning}.

\begin{lemma}[Public-key pseudoentanglement~\cite{bouland2023publickeypseudoentanglementhardnesslearning}]\label{lem:public} For any extensive bipartition $A|B$ with $n_A,n_B=\Omega(n)$ and $\delta' \in (0,1)$, there exist two families of classical descriptions of efficient quantum circuits $\mathcal{C}=\{C_i\}$ and $\mathcal{C}'=\{C_i'\}$ such that
\begin{itemize}
    \item Assuming that LWE cannot be solved in polynomial time on a quantum computer, $\mathcal{C}$ and $\mathcal{C'}$ are computationally indistinguishable from each other, even when their classical description is provided as input.
    \item $\ket{\psi}=C\ket{0}$ with $C\in\mathcal{C}$ obeys $S_{1}(\psi_A)=\Omega(n)$, while $\ket{\psi'}=C'\ket{0}$ with $C'\in\mathcal{C}'$ obeys $S_{1}(\psi'_A)=O(n^{\delta'})$.
\end{itemize}
\end{lemma}

Modifying the families of states introduced in \cref{sec:indstinguishablefamilies}, we now expand on the above lemma. 

\begin{lemma}\label{lem:circuitfamiliesapp}
    Consider an extensive bipartition $A|B$ with $n_A,n_B=\Omega(n)$. Then, there exist two families of classical descriptions of efficient quantum circuits $\mathcal{C}_{\eta,S_{\min}}$ and $\mathcal{C}'_{\eta,S_{\min}}$ such that the following statements 
    hold.
    \begin{enumerate}[label=(\roman*)]
        \item They are computationally indistinguishable, also when their classical description is provided as input, assuming that LWE is superpolynomially hard on a quantum computer for all $\eta,S_{\min}>0$.
        \item $\ket{\psi_{A,B;\eta}}=C\ket{0}$ with $C\in\mathcal{C}_{\eta,S_{\min}}$ obeys $S_{1}(\psi_{A;\eta})=S_{\min}+o(1)$, while $\ket*{\psi'_{A,B;\eta}}=C'\ket{0}$ with $C'\in\mathcal{C}'_{\eta,S_{\min}}$ obeys $S_{1}(\psi'_{A;\eta})=S_{\min}+S_{1}+o(1)$, for any $S_{1}=o(n^{1-\delta'})$ with $\delta' \in (0,1)$ and $\eta = S_1/n_A$.
        \item States generated by both families of quantum circuits have the same value of the min-entropy $S_{\min}(\psi_{A;\eta}) = S_{\min} + o(1)$ whenever $S_{\min} \leq O(n^{\delta'})$.
    \end{enumerate}
    \begin{proof}
        Similarly to \cref{sec:indstinguishablefamilies}, we construct the two families of circuits as 
        \be
\mathcal{C}_{\eta,S_{\min}}&\coloneqq\{C_{\eta,S_{\min}}\,:\, C_{\eta,S_{\min}}\ket{0}^{\otimes n+2}=\sqrt{1-\eta}\ket{0}_{A}\ket{0}_{B}\ket{0}^{\otimes n-2S_{\min}} \ket*{\phi^{+}_{A,B}}^{\otimes S_{\min}}+\sqrt{\eta}\ket{1}_{A}\ket{1}_{B} C\ket{0},\quad C\in\mathcal{C}\},\\
\mathcal{C}_{\eta,S_{\min}}'&\coloneqq\{C_{\eta,S_{\min}}'\,:\, C_{\eta,S_{\min}}'\ket{0}^{\otimes n+2}=\sqrt{1-\eta}\ket{0}_{A}\ket{0}_{B}\ket{0}^{\otimes n-2S_{\min}} \ket*{\phi^{+}_{A,B}}^{\otimes S_{\min}}+\sqrt{\eta}\ket{1}_{A}\ket{1}_{B} C'\ket{0},\quad C'\in\mathcal{C}'\}.
        \ee
        Here, $\mathcal{C},\mathcal{C}'$ are the families of circuits introduced in \cref{lem:public} for the same parameter $\delta'$ given in the statement of this lemma. Now, two things trivially hold: first, if the two families of circuits $\mathcal{C},\mathcal{C}'$ are computationally indistinguishable in the public-key sense due (up to the LWE assumption) so are the two families $\mathcal{C}_{\eta}$ and $\mathcal{C}_{\eta}'$. Second, the two circuit families $\mathcal{C}_{\eta},\mathcal{C}_{\eta}'$ contain an efficient classical description of their corresponding quantum circuits. We are just left with showing the gap in von Neumann entropy. Similarly to \cref{lemma:entanglement}, we set $\eta=S_{1}/n_A$ and, for any $S_{1}=o(n^{1-\delta'})$, we have that $C_{\eta,S_{\min}}\ket{0}$ with $C_{\eta,S_{\min}}\in\mathcal{C}_{\eta,S_{\min}}$ has von Neumann entropy $S_{1}(\psi_{A;\eta})=S_{\min}+o(1)$. Instead, $C_{\eta,S_{\min}}'\ket{0}$ with $C_{\eta,S_{\min}}'\in\mathcal{C}_{\eta,S_{\min}}'$ has von Neumann entropy $S_{1}(\psi'_{A;\eta})=S_{\min}+S_1+o(1)$ for any $\delta>0$. It is immediate to see that both families of circuits generate states with min-entropy $S_{\min}(\psi_A) = S_{\min}+o(1)$ whenever $S_{\min} \leq O(n^{\delta'})$ from the expression 
        \begin{equation}
        S_{\min}(\psi_{A;\eta}) = \min(-\log(1-\eta) + S_{\min} , -\log(\eta) + S_{\min}(\psi_A)).
        \end{equation}
    \end{proof}
\end{lemma}

\begin{theorem}[No-go on state-aware entanglement distillation. Formal version of \cref{th:nogostateawaredistillation}]\label{thapp:stateawaredistillation} Given an extensive bipartition $A|B$, and large enough $n\in\mathbb{N}^+$, there exists at least one efficiently preparable sequence of pure states $\ket{\psi_{A,B;n}}$ on $n$ qubits, such that any state-aware computationally efficient LOCC distillation protocol acting on $k = \omega(n_A^2)$ copies with error $\varepsilon \leq 10^{-4}$ and success probability $p\geq \sqrt{64/65}$ cannot distill at a rate larger than $S_{\min}(\psi_A) + o(1)$, assuming that the LWE problem is superpolynomially hard to solve on a quantum computer. This result holds regardless of the value of the von Neumann entropy $S_{1}(\psi_A)$, as long as $S_{1}(\psi_A) = O\left(n^{1-\delta}/\log n\right) = \tilde{O}(n^{1-\delta})$ for any $\delta\in(0,1)$, and whenever $S_{\min}(\psi_A) \leq O(n^{\delta})$.
    \begin{proof}
        To proceed with the proof, let us assume, towards contradiction, that there exists a computationally efficient state-aware LOCC protocol, that starting from the circuit description $C$ of any efficiently preparable state $\ket{\psi_{A,B}}=C\ket{0}$ distills at a rate at least equal to $S_{\min}(\psi_A)+\Omega(1)$. Such a protocol could then be used to distinguish between the two families of quantum circuits defined in \cref{lem:circuitfamiliesapp}, for $\delta'=\delta$. Indeed, for any value of $S_{\min} \leq O(n^{\delta})$, we have for both $\ket{\psi_{A,B}}=C_{\eta,S_{\min}}\ket{0_{A,B}}$ with $ C_{\eta,S_{\min}} \in \mathcal{C}_{\eta,S_{\min}}$ and $\ket{\psi_{A,B}}=C_{\eta,S_{\min}}\ket{0_{A,B}}$ with $ C_{\eta,S_{\min}}\in \mathcal{C}'_{\eta,S_{\min}}$ that $S_{\min}(\psi_A) = S_{\min} + o(1)$, while their von Neumann entropies are strictly different and can be tuned by varying $\eta$ (see the results of \cref{lem:circuitfamiliesapp}).

To construct the distinguisher, consider $\ket{\psi_{A,B}}=C_{\eta,S_{\min}}'\ket{0_{A,B}}$ with $ C_{\eta,S_{\min}}'\in \mathcal{C}_{\eta,S_{\min}}'$ and $\eta = S_{1}/n_A$. The hypothetical protocol would distill, by assumption, at a rate
\begin{equation}
S_{\min}(\psi_A)+\Omega(1) \leq \hat{E}_{D}^{(k)} \leq E_{D}^{(k)} \leq (S_{\min} + S_{1} + o(1)),
\end{equation}
for any $S_{1}=O\left(\frac{n^{1-\delta}}{\log n}\right)$.
Here, as before, the lower bound comes from the hypothesis towards contradiction, and the upper bound comes from \cref{lem:orderrelation,lemma:infothbounds,lem:circuitfamiliesapp}. In contrast, for $\ket{\psi_{A,B}}=C_{\eta,S_{\min}}\ket{0_{A,B}}$ with $ C_{\eta,S_{\min}}\in \mathcal{C}_{\eta,S_{\min}}$, such a protocol cannot distill more than $S_{1}(\psi_A) = S_{\min}(\psi_A) + o(1)$ due to the information-theoretic bound \cref{lemma:infothbounds}. This provides a contradiction because the two families of quantum circuits are computationally indistinguishable even when having a classical circuit description, assuming that LWE is computationally hard to solve on a quantum computer (see \cref{lem:circuitfamiliesapp}). This concludes the proof.
    \end{proof}
\end{theorem}

\begin{theorem}[No-go on state-aware entanglement dilution. Formal version of \cref{th:nogostateawaredistillation}]\label{thapp:stateawaredilution}
Given an extensive bipartition $A|B$, and large enough $n\in\mathbb{N}^+$, there exists at least one family of efficiently preparable states $\ket{\psi_{A,B;n}}$ on $n$ qubits, such that any state-aware computationally efficient LOCC dilution protocol with error $\varepsilon \leq 10^{-4}$ and success probability $p\geq \sqrt{64/65}$ cannot dilute the state across $A|B$ at a rate lower than $\Omega(n^{1-\delta})$, assuming that the LWE problem is superpolynomially hard to solve on a quantum computer. This result holds regardless of the value of the von Neumann entropy $S_{1}(\psi_A)$, as long as $S_{1}(\psi_A) = O(n^{1-\delta})$ for any $\delta\in(0,1)$.
\begin{proof}
    Let us consider the two families of quantum circuits defined in \cref{lem:circuitfamiliesapp}, $\mathcal{C}_{\eta,S_{\min}}$ and $\mathcal{C}_{\eta,S_{\min}}'$, which are computationally indistinguishable (up to the LWE assumption) for any $\eta>0,S_{\min}\geq0$ and to any distinguisher having access to their classical description. By \cref{lem:circuitfamiliesapp}, setting $\delta' = \delta/2$, $S_{\min}=0$ and $S_1=O(n^{1-\delta})$, we can tune the entanglement entropy of the two families to have an arbitrary large gap. In particular, a state $\ket{\psi_{A,B}}=C_{\eta,0}\ket{0_{A,B}}$ with $C_{\eta,0}\in  \mathcal{C}_{\eta,0}'$ has entanglement entropy $\Theta(n^{1-\delta})$, while $\ket{\psi_{A,B}}=C_{\eta,0}\ket{0_{A,B}}$ with $C_{\eta,0}\in  \mathcal{C}_{\eta,0}'$ has entanglement entropy $o(1)$. 
    
    Now assume, towards contradiction, that there exists a state-aware LOCC dilution protocol that can be designed from the knowledge of the (efficient) classical description of the quantum circuit preparing the state that is able to dilute the state with von Neumann entropy $S_{1}=o(n^{1-\delta})$ at a rate $o(n^{1-\delta})$. This hypothetical protocol would distinguish between the classical descriptions belonging to either of the two families of quantum circuits described above, leading to a contradiction. We can construct a distinguisher as follows. 

If the dilution protocol is given a state $C_{\eta,0} \in \mathcal{C}_{\eta,0}$, since $ S_{1}=o(1)$ 
(\cref{lem:circuitfamiliesapp}), only $o(n^{1-\delta})$ ebits should be consumed by the hypothesis towards contradiction. On the other hand, if the protocol is given $C_{\eta,0}'\in\mathcal{C}_{\eta,0}'$, $\Omega(n^{1-\delta})$ ebits should be consumed at least. Indeed, by \cref{lem:orderrelation,lemma:infothbounds}, since 
\begin{equation}
\hat{E}_{C}^{(k)}  \ge E_{C}^{(k)}\ge S_{1}=\Omega(n^{1-\delta}) 
\end{equation}
for states sampled from $\mathcal{C}_{\eta,0}'$ then no $k$-shot LOCC protocol can use fewer than $\Omega(n^{1-\delta})$ many ebits to dilute the state with probability greater than $\sqrt{64/65}$. The distinguisher then can count the number of ebits used by the protocol. This would then allow one to distinguish the two circuit families in input, leading to a contradiction. Therefore, we have shown that any state-aware LOCC protocol must use at least $\Omega(n^{1-\delta})$ many supplied ebits in the worst case.
\end{proof}
 \end{theorem}

\subsection{Results for specific classes of states}\label{App:classesofstates}
In this final section, we show how one can circumvent the no-go results showed in the previous sections by examining restricted classes of states. Here, in particular, we consider low-rank, t-doped stabilizer states, for which results concerning state-agnostic entanglement distillation has been derived in Ref.\ \cite{gu2024magicinducedcomputationalseparationentanglement}, as well as some typical classes of states appearing in many-body quantum systems.

\subsubsection{Low-rank states}
If Alice and Bob are guaranteed that the input state has low rank, i.e.,  $r = O(\poly(n))$, then they can distill at the von-Neumann entropy rate using the state-agnostic protocol discussed in \cref{App:agnosticgo}. This is a simple consequence of the structure of the irreducible representations of the symmetric group $S_k$, which do not depend explicitly on the dimension $d$ but only through the Young coefficients $\lambda$. This result is summarised by the following corollary.

\begin{corollary}[Computational distillable entanglement of low-rank states] 
If a pure state $\psi_{A,B}$ is guaranteed to have rank $r = O(\poly(n))$ , then its k-shot computational distillable entanglement is $S_1(\psi_A)$. 
\end{corollary}
\begin{proof}
The proof easily follows from the formula for the probability of projecting onto an irrep $\lambda$ 
\be
\Pr(\lambda) = \frac{\dim\mathcal{V}_\lambda}{k!}\sum_\pi \chi^{\lambda}(\pi) \tr (R_{\pi} \rho^{\otimes k}).
\ee
Indeed, since $r = O(\poly(n))$, we can always choose $k = \poly(n)$ such that $ r < k < d$. Recall that all the valid values of $\lambda$ satisfy $\sum_{i=1}^d\lambda_i=k$ and $\lambda_{i+1}\geq\lambda_i\geq0$, so $\lambda = (\lambda_1, \dots, \lambda_k,0, \dots, 0)$ has at most $k$ non-zero entries. In particular, let us choose a restricted subset of all the possible $\lambda$'s, that is those with at most $r$ non-zero entries $\hat{\lambda} = (\lambda_1, \dots, \lambda_r,0, \dots, 0)$. It is easy to see that if the state has rank $r$, then the probabilities evaluated on $\hat{\lambda}$ are exactly the same that one would have obtained by fixing the dimension to be $r$ from the beginning. In particular, both $\dim\mathcal{V}_\lambda$ and $\chi^{\lambda}(\pi)$ depend on the dimension only through $\lambda$, therefore when we consider integer partitions of this form, since the $d$-dimensional state can be substituted with an effective $r$-dimensional one, everything goes like $d\to r$, therefore by the results of Ref.\ \cite{PhysRevA.75.062338}, since we can choose $k\gg r$ we know that there is a $k = \poly(n)$ such that we distill $S_1(\psi_A)$ with probability arbitrarily close to one.
\end{proof}

\subsubsection{Implications for quantum many-body states}

The setting of many-body states is precisely the one analyzed in this work. Indeed, many-body systems are typically characterized by a large number of particles, meaning that $n$ is often considered very large. Furthermore, key properties of quantum many-body systems, such as quantum phase transitions, emerge in the limit $n \to \infty$, commonly referred to as the \textit{thermodynamic limit}. In this section, we aim to address the following question: how much entanglement can be distilled from physical many-body states?  

For simplicity, we focus on quantum many-body systems described by a quadratic Hamiltonian in terms of fermionic creation and annihilation operators. This also includes spin Hamiltonians that can be mapped to quadratic fermionic Hamiltonians via transformations such as the Jordan-Wigner transformation. In Ref.~\cite{PhysRevA.73.060303}, it has been shown that for a bipartition $A|B$, if the ground state $\ket{\psi}$ of such a model exhibits a logarithmic growth of the von Neumann entropy, i.e., $S_{1}(\psi_A) = \xi \log n_A$, then the min-entropy satisfies $S_{\min}(\psi_A) = \frac{1}{2} S_{1}(\psi_A) + O(1/\log n_A)$. This result is particularly intriguing because a logarithmic scaling of entanglement is a hallmark of one-dimensional chains at the critical point of a phase transition. A notable example, though the result holds more generally, is the transverse field Ising model at criticality, where $\xi \propto c$, the conformal charge of the conformal field theory describing the critical point of the many-body system. Having this in mind, it is clear that applying the protocol discussed in \cref{th:distillation} with $k = O(\poly n_A)$ copies of the many-body state $\ket{\psi}$, one can efficiently distill at a rate $\min\{\Theta(S_{\min}), \Theta(\log n_A)\} = \Theta(S_1(\psi_A))$. Thus, almost all the entanglement present in a many-body system at criticality can be distilled using the agnostic protocol based on the Schur transform discussed in \cref{section:protocol}.

\subsubsection{$t$-doped stabilizer states: entanglement- versus magic-dominated states}
Another class of states that we consider is the one of $t$-doped stabilizer states. These are constructed starting from the computational basis by applying Clifford gates and at most $t'$ non-Clifford gates acting on $l$ qubits, such that $t'l = O(t)$. This definition, given in Ref.~\cite{gu2024magicinducedcomputationalseparationentanglement}, generalizes the conventional notion of $t$-doped stabilizer states, which are typically considered to be ``doped'' only by single-qubit non-Clifford gates. The key distinction is that, in the extreme case where $t' = 1$, we have $l \sim t$, and as long as $t = \omega(\log n)$, the states may not be efficiently preparable. 

In Ref.~\cite{gu2024magicinducedcomputationalseparationentanglement}, the problem of distilling entanglement from states with some degree of magic, quantified by the doping parameter $t$, was explored. In particular, the feasibility of entanglement distillation from a state $\ket{\psi}$ depends on whether $\ket{\psi}$ belongs to the phase of \textit{entanglement-dominated} states, where for a given bipartition $A|B$, we have $S_{1}(\psi_A) = \omega(t)$, or the phase of \textit{magic-dominated} states, where $S_{1}(\psi_A) = O(t)$. In the former case, a computationally efficient distillation protocol exists, achieving a distillation rate of $S_{1}(1 + o(1))$, whereas for certain magic-dominated states, only a vanishing fraction $o(S_1)$ (relative to the von Neumann entropy) of Bell pairs can be efficiently distilled.  

These findings fit naturally into the framework of this paper. We have shown that there exist classes of states for which no computationally efficient protocol can distill more than the min-entropy $S_{\min}$. Thus, it follows that these states must necessarily belong to the magic-dominated phase. Therefore, \cref{th:upperbounddistillableent} improves the findings of Ref.~\cite{gu2024magicinducedcomputationalseparationentanglement} in this regard. On the other hand, for any entanglement-dominated state, it is possible to leverage knowledge of its stabilizer group to implement an efficient protocol that distills all of its entropy. More importantly, the stabilizer group can be efficiently learned using only $O(n)$ copies in an agnostic scenario (see Ref.~\cite{grewal2024efficientlearningquantumstates}) or efficiently estimated from the circuit description of the state in a state-aware scenario (see Ref.~\cite{gu2024magicinducedcomputationalseparationentanglement}). Therefore, entanglement-dominated states are efficiently distillable in both computationally feasible settings analyzed in this paper: the state-agnostic and state-aware scenarios.

To summarize, for any entanglement-dominated state $\ket{\psi_{ED}}$, we find that $\hat{E}_{D}(\psi_{ED,A}) = S_{1}(\psi_{ED,A})$, whereas, according to \cref{th:upperbounddistillableent}, there exist classes of magic-dominated states $\psi_{MD}$ for which $\hat{E}_{D}(\psi_{MD,A}) = S_{\min}(\psi_{MD,A})$. Remarkably, the gap between $S_{\min}$ and $S_{1}$ can be as large as $\Omega(n)$ vs. $o(1)$, underscoring the profound difference between these two classes, or phases, of states.

\subsubsection{Efficiently learnable quantum states: distillation-through-learning}\label{sec:distillation-learning}
Although seemingly unrelated, the task of quantum state tomography provides a valuable connection between the state-agnostic and state-aware scenarios discussed in this paper. Quantum state tomography can be described as follows: an agent is given multiple copies of an unknown quantum state \(\ket{\psi_{A,B}}\) and a precision parameter \(\varepsilon\). The agent's goal is to output a classical description \(\ket*{\hat{\psi}_{A,B}}\) that is \(\varepsilon\)-close to \(\ket{\psi_{A,B}}\) in trace distance. When \(\ket{\psi_{A,B}}\) admits an efficient classical description, quantum state tomography can serve as a subroutine to transition from a state-agnostic entanglement manipulation protocol to a state-aware manipulation protocol. Specifically, Alice and Bob can allocate a portion of the available state copies to generate a classical description, allowing them to implement state-aware entanglement manipulation techniques. As we show in \cref{sec:LOCCtomography}, if the state is efficiently learnable globally, then it is also efficiently learnable via LOCC.

To illustrate, suppose that \(k = O(\poly n)\) copies are sufficient to obtain a classical description for a given class of states. Assume further that, using this classical description, Alice and Bob design an LOCC protocol that produces \(k'R'\) ebits from \(k'\) copies of the state, where \(k' = O(\poly n)\). In the limit where \(k' \gg k\), the overall distillation rate \(R\) can be expressed as  
\begin{equation}
R = \frac{k'R'}{k' + k} \simeq R'\,.
\end{equation}  
which shows that, for \(k' \gg k\), the fraction of copies used for tomography becomes negligible.  

For general classes of states, however, the connection with tomography remains speculative. As shown in Ref.~\cite{PRXQuantum.5.040306}, performing quantum state tomography for states prepared using \(G\) native gates requires computational time of \(\Theta(\exp G)\). Consequently, learning states with efficient classical descriptions (\(G = O(\poly n)\)) remains computationally infeasible in the worst case, aligning with the findings of \cref{th:nogostateawaredistillation}.  

Nevertheless, certain classes of states may provide a more practical link between the state-agnostic and state-aware scenarios. A notable example is the class of $t$-doped stabilizer states, discussed in the previous subsection. This class is particularly interesting from the perspective of distillation-through-learning, as it satisfies the following conditions:  
1) It admits a partial but efficiently learnable classical description, which is sufficient to design an optimal and computationally efficient distillation protocol.  
2) It includes highly entangled states, surpassing the capabilities of the agnostic protocol discussed in \cref{section:protocol}.  

Other classes of efficiently learnable states that exhibit high entanglement include the class of free-fermionic states~\cite{bittel2025optimaltracedistanceboundsfreefermionic} and their deformations, such as $t$-doped fermionic states~\cite{Mele_2025}. However, whether one can design efficient LOCC protocols from their classical descriptions remains an open question and an exciting avenue for future research.  

\section{Fundamental limitations of efficient state compression}\label{sec:statecompression}
In the task of state compression \cite{PhysRevA.51.2738, PhysRevA.66.022311}, the sender Alice and the receiver Bob have access to limited uses of a noiseless quantum channel, and want to exploit it to transmit a given quantum information source using as few channel uses as possible. For this reason, Alice wants to compress the source via an encoding process and send the compressed version to Bob, who will then retrieve the original quantum message via a proper decoding. More precisely, the quantum source will be an ensemble of (not necessarily orthogonal) pure states $\{p_i,\ket{\psi_i}_A \}$. Equivalently, we can consider any purification of the source $\ket{\phi}_{E,A} = \sum_i \sqrt{p_i} \ket{u_i}_E\ket{\psi_i}_A$, and assume that Alice has only access to her register $A$, and not to the environment $E$. Given local sample-access to $k$ copies of the purified source, $\ket{\phi}_{E,A}^{\otimes k}$, we say that a compression protocol achieving rate $R$ with error $\varepsilon$ with some success probability $p$ exists if there is an encoding CPTP map $\mathcal{E}^{A^k \to W}$ and a decoding CPTP map $\mathcal{D}^{W\to \hat{A}^k}$ such that $\log|W| \leq kR$, and with probability $p$, $(\mathcal{D}^{W\to \hat{A}^k} \circ \mathcal{E}^{A^k \to W})(\phi_{E,A}^{\otimes k})$ is $\varepsilon$-close to $\phi_{E,A}^{\otimes k}$ in trace distance. It is well known that with many, but still $O(\poly n)$, samples the minimum achievable rate of compression is given by the von Neumann entropy of the source $S_1(\rho_A)$ \cite{PhysRevA.51.2738}. Due to our no-go results on the entanglement cost, we easily get a corresponding no-go result on compression protocols, for which essentially the same limitations apply. Again, we consider three cases: (1) the protocol is sample-efficient and state-agnostic; (2) the protocol is computationally efficent; (3) the protocol is computationally efficient and state-aware, that is it has full knowledge of the circuit preparing the quantum source $\rho_A = \sum_i p_i \psi_i$ (but not the individual states $\{ \ket{\psi_i}_A \}_i$). 
\begin{corollary}[No-go on state compression]
Any approximate $k$-shot with $k=O(\operatorname{poly}n)$ and $k = \omega(n_A^2)$ coding protocol having success probability $p\geq\sqrt{64/65}$ and error $\varepsilon\leq 10^{-4}$ which is either (1) state-agnostic, (2) computationally efficient, or (3) computationally efficient and state-aware (see \cref{app:stateaware}), must use, in the worst-case scenario, $\tilde{\Omega}(n_A^{1-\delta})$ (for arbitrary $\delta\in(0,1)$) many encoded qubits per copy, regardless of the value of $S_1(\rho_A)\geq o(1)$. Thus, in the worst case, no efficient compression is possible.
\begin{proof}
We prove the statement for the state-agnostic case, the other cases follow similarly due to \cref{appth:lowerboundentcost,thapp:stateawaredilution}. Let us assume, for the sake of contradiction, that there exists a state-agnostic and sample-efficient approximate compression protocol achieving a rate $\tilde{o}(n_A)$ when $S_1(\rho_A) = \tilde{o}(n_A)$ on every quantum source. Then, this agnostic protocol could be used to dilute any state using $\tilde{o}(n_A)$ ebits, thus violating our previous no-go result \cref{th:entanglementcostnogo}. Indeed, Alice could first compress $k$ copies of the state locally, and then consume $\tilde{o}(n_A)$ ebits per copy to teleport the compressed state, which must then be decoded by Bob. The final outcome of this procedure provides a a state-agnostic dilution rate that exceeds the one in \cref{th:entanglementcostnogo} and therefore our claim follows.
\end{proof}
\end{corollary}

\section{Efficient protocols for pure-state LOCC tomography}\label{sec:LOCCtomography}
In \cref{sec:distillation-learning}, we argued that it is possible to design an LOCC distillation scheme by first learning the input state. Here, we demonstrate the existence of a synergy between the problem of agnostic entanglement distillation and quantum learning theory. Specifically, we investigate the following question: assuming that Alice and Bob share copies of an unknown entangled state, how can they learn the state using LOCC operations alone? Moreover, how does the sample complexity compare to the case when one can implement fully general (global) measurements?

In the following, we show that, surprisingly, due to the existence of the protocol analyzed in \cref{section:protocol}, performing LOCC tomography in a two-party setting essentially incurs no disadvantage compared to general measurements strategies. We begin by formally introducing the problem.

\begin{definition}[two-party LOCC tomography] Let $\rho_{AB}$ a bipartite quantum state shared between two parties $A$ and $B$. The task of two-party LOCC tomography consists in designing an LOCC quantum algorithm acting on multiple copies of $\rho_{AB}$ that outputs a classical description $\hat{\rho}_{AB}$ such that, with success probability $\ge 2/3$
\be
\|\rho_{AB}-\hat{\rho}_{AB}\|_{1}\le\varepsilon.
\ee
\end{definition}
In the following, we analyze the tomography algorithm for the restricted case where a pure entangled state $\ket{\psi_{AB}}$ is shared between the parties. The idea behind the algorithm is straightforward: we treat the shared entangled state as a resource, use it to agnostically distill Bell pairs via \cref{algo1}, and then employ these Bell pairs to teleport a sufficient number of copies, allowing the execution of an optimal single-party tomography algorithm.

\begin{lemma}\label{lem:productstructuresmin}
    Let $\psi_{AB}$ be a bipartite pure state and $\rho_A,\rho_B$ its reduced density matrices. Let $\|\rho_B\|=\|\rho_A\|>1/2$, then the unique eigenvectors $\ket{\psi_{A}},\ket{\psi_B}$ corresponding to the maximum eigenvalue $\lambda_1\equiv\|\rho_A\|=\|\rho_B\|$ satisfy
    \be
\|\psi_{AB}-\psi_A\otimes\psi_B\|_1=\sqrt{2(1-\|\rho_A\|)}.
    \ee
\begin{proof}
    It is easy to be convinced that if $\|\rho_A\|>1/2$, then there is a unique eigenvector corresponding to the eingenvalue $\lambda_1$ for each party $A$ and $B$. Therefore, the Schmidt decomposition of the state reads $\ket{\psi_{AB}}=\sqrt{\lambda_1}\ket{\psi_{A}}\ket{\psi_B}+\sum_{i\neq 1}\sqrt{\lambda_{i}}\ket{\phi_{iA}}\ket{\phi_{iB}}$. The result then follows by a trivial calculation. 
\end{proof}
\end{lemma}

\begin{lemma}\label{lem:closeststate}
    Let $\rho$ be a quantum state, with $\|\rho\|>1/2$. There is a unique closest rank-1 approximation in trace distance and it corresponds to the unique eigenvector corresponding to the eingenvalue $\|\rho\|$. In other words, it holds that $\operatorname{argmin}_{\psi}\|\rho-\psi\|_1=\psi_1$ where $\rho\ket{\psi_1}=\lambda_1\ket{\psi_1}$.
    \begin{proof}
If $\|\rho\|>1/2$, then there is a unique eigenvector corresponding to the eigenvalue $\lambda_1=\|\rho\|$, denoted as ${\psi_1}$. Let us first show that $\min_{\psi}\|\rho-\psi\|_1=2(1-\lambda_1)=\|\rho-\psi_1\|_1$. First we have:
\be
\min_{\psi}\|\rho-\psi\|_1\le \|\rho-\psi_1\|_1=2(1-\lambda_1),
\ee
where we expressed $\rho=\lambda_1\psi_1+\sum_{i\neq 1}\lambda_i\psi_i$ and used that $\sum_{i\neq 1}\lambda_i=1-\lambda_1$.
Let us show the converse, i.e., $\min_{\psi}\|\rho-\psi\|_1\ge 2(1-\lambda_1)$. We recall an important property of the trace norm, namely $\|\rho-\sigma\|_1 = 2\max_{0\le\Lambda\le 1}\tr((\rho-\sigma)\Lambda)$ for any quantum states $\rho$ and $\sigma$. Therefore, we have
\be
\min_{\psi}\|\rho-\psi\|_1&=2\min_{\psi}\max_{0\le\Lambda\le 1}\tr(\Lambda(\psi-\rho))\\&\ge 2\max_{0\le\Lambda\le 1}\min_{\psi}\tr(\Lambda(\psi-\rho))\\&\ge2\min_{\psi}\tr(\psi(\psi-\rho))\\
&=2(1-\max_{\psi}\tr(\psi\rho))\\
&=2(1-\lambda_1).
\ee
Since $\|\rho-\psi_1\|_1=2(1-\lambda_1)$, then $\psi_1$ is a closest rank-1 approximation. We are just left to show that it is the unique closest approximation. 
Let us assume, by seek of contradiction, that there is $\tilde{\psi}$ being a rank-1 approximation, $\|\rho-\tilde{\psi}\|_1=2(1-\lambda_1)$, such that $\tilde{\psi}\neq\psi_1$. Hence we can write it as $|\tilde{\psi}\rangle=\sqrt{p}\ket{\psi_{1}}+\sqrt{1-p}\ket{\psi_{1}^{\perp}}$. We can now use Fuchs-van der Graaf, i.e., $\|\rho-\tilde{\psi}\|_1\ge 2(1-\tr(\rho\tilde{\psi}))$~\cite{Nielsen_1999}. Then
\be
(1-\lambda_1)\ge 1-(p\lambda_1+(1-p)\langle\psi_{1}^{\perp}|\rho|\psi_{1}^{\perp}\rangle),
\ee
from which, since $\langle\psi_{1}^{\perp}|\rho|\psi_{1}^{\perp}\rangle<\lambda_1$,  it follows that $p=1$ necessarily. Hence, the unique closest rank-1 approximation is the eigenvector $\psi_1$ with eigenvalue $\lambda_1$.
    \end{proof}
\end{lemma}

\begin{lemma}[Principal component tomography~\cite{odonnell2015efficientquantumtomography}]\label{lem:principalcomponenttomography} Let $\rho$ be a (possibly) mixed quantum state. Then there exists a quantum tomography algorithm consuming $O(d/\varepsilon^2)$ many copies of $\rho$ that outputs a state $\tilde{\rho}$ which is $\varepsilon$ close (in trace ditance) to the best rank-1 approximation of $\rho$ in trace distance. 
\end{lemma}
We now have all the necessary ingredients to prove the main result of this section—namely, that due to the existence of an agnostic distillation protocol that distills $\min\{\frac{1}{4}S_{\min},\frac{1}{4}\log k\}$ with probability $\ge 2/3$, the two-party LOCC tomography algorithm achieves the same sample complexity as standard tomography, up to a polynomial overhead in the number of qubits~$n$ and the inverse precision $1/\varepsilon$. In particular, this implies that for general pure states, the two-party LOCC tomography algorithm is as optimal as standard tomography, up to a polynomial overhead.

\begin{theorem}[Optimal two-party LOCC tomography algorithm for pure states]\label{th:LOCCtomography} Let $\ket{\psi_{AB}}$ a pure bipartite pure state of $n$ qubits. Then there exists a two-party LOCC algorithm consuming at most $\max\{\left(\frac{32n}{\varepsilon^2}+1\right)K_{n},K_{A}+K_{B}\}$, where $K_n,K_{A},K_{B}$ are the sample complexities for quantum state tomography for the class of states $\ket{\psi_{AB}}$ belongs to, and the principal component tomography of the reduced density matrices $\rho_A$, $\rho_B$ respectively. Furthermore, the scaling of this algorithm with the number of qubits is optimal for general pure quantum states up to a linear overhead in $n$, as it reduces to $O(nd/\varepsilon^4)$ in the worst case.
\begin{proof}
To prove the result, let us give an explicit two-party LOCC protocol and let us bound the sample complexity required. 

Let us for the moment assume that Alice and Bob know the operator norm $\|\rho_A\|$ of the reduced density matrix of $\ket{\psi_{AB}}$.
    \begin{enumerate}[label=(\roman*)]
    \item $S_{\min}\le \varepsilon^2/8$. Both Alice and Bob perform a principal component tomography on $\rho_A$ (resp. $\rho_B$). Both consume $K_A$ (resp. $K_B$) many copies to output the a state which is $\varepsilon'$ close to the unique best rank-1 approximation of $\rho$ in trace distance. Notice that for any choice of $0<\varepsilon<1$, then $\|\rho_A\|>1/2$, and \cref{lem:closeststate} ensures that there is only one closest state and it corresponds to the eigenvector of $\rho_A$ (resp. $\rho_B$) corresponding to the eigenvalue $\|\rho_A\|=\|\rho_B\|$.  Let $\hat{\psi}_A,\hat{\psi}_B$ the largest eignevectors and $\tilde{\psi}_A,\tilde{\psi}_B$ the respective $\varepsilon'$-approximations output by the algorithm. Outputting as a classical description  $\tilde{\psi}_A\otimes\tilde{\psi}_B$ is sufficient.  Indeed, thanks to \cref{lem:productstructuresmin}, we have
    \be
\|\psi_{AB}-\tilde{\psi}_A\otimes\tilde{\psi}_B\|_1\le \|\hat{\psi}_A\otimes\hat{\psi}_B-\tilde{\psi}_A\otimes\tilde{\psi}_B\|_1+\|\psi_{AB}-\hat{\psi}_A\otimes\hat{\psi}_B\|_1\le  2\varepsilon'+\sqrt{2(1-\|\rho_A\|)}
    \ee
choosing $\varepsilon'=\varepsilon/4$ and for $S_{\min}\le \varepsilon^2/8$, the algorithm succeeds.
    \item $S_{\min}> \varepsilon^2/8$. Denote $K_n$ the sample complexity of the optimal tomography algorithm for $\ket{\psi_{AB}}$. In this case, Alice' and Bob's strategy is to teleport $K_n$ copies of the shared entangled state and then proceed with such optimal tomography algorithm. To teleport $K_n$ copies, they need $nK_n$ perfect Bell pairs. We use the agnostic distillation scheme described in \cref{section:protocol} that, using $k$ copies of a shared entangled state is able to produce $k S_{\min}/4$ many perfect Bell pairs with success probability $\ge 2/3$. To produce $nK_n$ many perfect Bell pairs to perform perfect teleportation, we thus need to spend $k=4nK_n/S_{\min}\le \frac{32nK_n}{\varepsilon^2}$ many initial copies. Once $K_n$ copies are tranferred to Alice (or, equivalently, Bob) they can perform standard tomography with $K_n$ copies and achieve the task. Therefore, the total number of copies to achieve $\varepsilon$ trace distance error is simply $k+K_n=\left(\frac{32n}{\varepsilon^2}+1\right)K_n$.
\end{enumerate}
Therefore, with probability at least $2/3$ and using at most $\max\{\left(\frac{32n}{\varepsilon^2}+1\right)K_n,K_A+K_B\}$ many copies then the two-party LOCC tomography algorithm succeeds. Let us show that this algorithm is optimal for general pure quantum states. Thanks to \cref{lem:principalcomponenttomography}, we have that $K_{A}=O(d_{A}/\varepsilon'^{2})$ and $K_{B}=O(d_{B}/\varepsilon'^{2})$ for any reduced density matrix. Moreover, $K_n=O(d/\varepsilon^2)$ for pure state tomography. Therefore, the sample complexity of the algorithm in the worst case is $O(nd/\varepsilon^{4})$. To show the optimality, it we need a lower bound on two-party LOCC tomography, which can be simply obtained by noticing that it must be more inefficient than standard tomography for pure state. A lower bound for standard tomography is $\Omega(d/\varepsilon^2)$~\cite{Haah_2017}. Hence, up to a linear overhead of the order $O(n/\varepsilon^2)$, the two-party LOCC algorithm performs as well as the standard one.
\end{proof}
    
\end{theorem}

Since the algorithm provided is highly general, one might wonder whether it can be applied to classes of states that are learnable with a polynomial number of samples. This is precisely why we formulated the theorem without specifying the sample complexities $K_n, K_A, K_B$. Hence, it is evident that, at least in the regime of non-vanishing $S_{\min}$, the two-party LOCC tomography algorithm developed in \cref{th:LOCCtomography} enables the learning of efficiently learnable state classes with only a linear overhead in $n$.

However, the case of mixed states is profoundly different. The strategy behind our LOCC tomography algorithm relies on using the unknown state as a resource to distill Bell pairs, which are then used to teleport the state to one side, where standard tomography can be performed. In the mixed-state scenario, however, there exist \textit{bound entangled} states—(possibly highly) entangled states from which no Bell pairs can be distilled. As a result, our strategy inevitably fails for mixed states. In Ref.~\cite{tirone_2024}, a different approach was adopted for the mixed-state case: instead of distilling Bell pairs, the mixed state is teleported to one side using the closest separable approximation to a Bell pairs. This results in imperfect teleportation and introduces an additional dimensional factor $d$, leading to a sample complexity of $O(d^3)$ for general mixed states. An intriguing open question concerns the optimality of the strategy used in Ref.~\cite{tirone_2024}  for mixed states and whether, unlike in the pure-state case, LOCC tomography for mixed states is necessarily less efficient than standard tomographic approaches.

\twocolumngrid

\providecommand{\noopsort}[1]{}\providecommand{\singleletter}[1]{#1}%

\end{document}